\documentclass[reqno,10pt,letterpaper]{amsart}

\usepackage{lipsum}
\usepackage{amsmath}
\usepackage{amssymb}
\usepackage{amsthm}
\usepackage{mathrsfs}
\usepackage{accents}
\usepackage{calc}
\usepackage{arydshln}
\usepackage{upgreek}
\usepackage{slashed}
\usepackage{xifthen}
\usepackage{graphicx}
\usepackage{longtable}
\usepackage[inline]{enumitem}

\usepackage{xcolor}
\definecolor{winered}{rgb}{0.6,0,0}
\definecolor{lessblue}{rgb}{0,0,0.7}

\usepackage[pdftex,colorlinks=true,linkcolor=winered,citecolor=lessblue,urlcolor=lessblue,breaklinks=true,bookmarksopen=true]{hyperref}

\makeatletter
\newcommand{\myitem}[2]{\item[\rm(#2)]\def\@currentlabel{#2}\label{#1}}
\makeatother

\usepackage{titletoc}

\makeatletter

\def\@tocline#1#2#3#4#5#6#7{
\begingroup
  \par
    \parindent\z@ \leftskip#3 \relax \advance\leftskip\@tempdima\relax
                  \rightskip\@pnumwidth plus 4em \parfillskip-\@pnumwidth
    \ifcase #1 
       \vskip 0.6em \hskip 0em 
       \or
       \or \hskip 0em 
       \or \hskip 1em 
    \fi%
    #6
    \nobreak\relax{\leavevmode\leaders\hbox{\,.}\hfill}
    \hbox to\@pnumwidth {\@tocpagenum{#7}}
  \par
\endgroup
}

 \def\l@section{\@tocline{0}{0pt}{0pc}{}{}}

\renewcommand{\tocsection}[3]{%
  \indentlabel{\@ifnotempty{#2}{ 
    \ignorespaces\bfseries{#2. #3}}}
  \indentlabel{\@ifempty{#2}{\ignorespaces\bfseries{#3}}{}} 
    \vspace{1.5pt}}

\renewcommand{\tocsubsection}[3]{%
  \indentlabel{\@ifnotempty{#2}{
    \ignorespaces#2. #3}}
  \indentlabel{\@ifempty{#2}{\ignorespaces #3}{}}
    \vspace{1.5pt}}

\renewcommand{\tocsubsubsection}[3]{%
  \indentlabel{\@ifnotempty{#2}{
    \ignorespaces#2. #3}}
  \indentlabel{\@ifempty{#2}{\ignorespaces #3}{}}
    \vspace{1.5pt}}

\makeatother

\makeatletter
\def\@nomenstarted{0}
\newlength{\@nomenoldtabcolsep}

\newcommand{\nomenstart}
  {%
    \def\@nomenstarted{1}%
    \setlength{\@nomenoldtabcolsep}{\tabcolsep}%
    \setlength{\tabcolsep}{3.5pt}%
    \begin{longtable}{p{0.11\textwidth} p{0.86\textwidth}}
  }

\newcommand{\nomenitem}[2]{%
    \ifcase\@nomenstarted%
      \or 
      \or \\ 
    \fi%
    #1\,{\leavevmode\leaders\hbox{\,.}\hfill} & #2%
    \def\@nomenstarted{2}%
  }%
\newcommand{\nomenend}
  {\\%
      \end{longtable}%
      \setlength{\tabcolsep}{\@nomenoldtabcolsep}%
      \def\@nomenstarted{0}%
  }
\makeatother

\makeatletter
\newcommand{\BIG}{\bBigg@{3.5}}
\newcommand{\vast}{\bBigg@{4}}
\newcommand{\Vast}{\bBigg@{5}}
\newcommand{\VAST}[1]{\bBigg@{#1}}
\makeatother

\allowdisplaybreaks

\numberwithin{equation}{section}
\numberwithin{figure}{section}
\newtheorem{thm}{Theorem}[section]

\newtheorem{prop}[thm]{Proposition}
\newtheorem{lemma}[thm]{Lemma}
\newtheorem{cor}[thm]{Corollary}
\newtheorem{conj}[thm]{Conjecture}

\newtheorem*{thm*}{Theorem}
\newtheorem*{prop*}{Proposition}
\newtheorem*{cor*}{Corollary}
\newtheorem*{conj*}{Conjecture}

\theoremstyle{definition}
\newtheorem{definition}[thm]{Definition}

\theoremstyle{remark}
\newtheorem{rmk}[thm]{Remark}

\makeatletter
\newcommand{\fakephantomsection}{%
  \Hy@MakeCurrentHref{\@currenvir.\the\Hy@linkcounter}
  \Hy@raisedlink{\hyper@anchorstart{\@currentHref}\hyper@anchorend}%
  \Hy@GlobalStepCount\Hy@linkcounter%
}
\makeatother

\newcommand{\mc}{\mathcal}
\newcommand{\cA}{\mc A}

\newcommand{\cC}{\mc C}

\newcommand{\cI}{\mc I}

\newcommand{\cK}{\mc K}
\newcommand{\cL}{\mc L}
\newcommand{\cM}{\mc M}

\newcommand{\cO}{\mc O}

\newcommand{\cR}{\mc R}

\newcommand{\cU}{\mc U}
\newcommand{\cV}{\mc V}

\newcommand{\ms}{\mathscr}

\newcommand{\sE}{\ms E}

\newcommand{\sR}{\ms R}

\newcommand{\sV}{\ms V}
\newcommand{\sW}{\ms W}

\newcommand{\TT}{\mathbb{T}}

\newcommand{\C}{\mathbb{C}}
\newcommand{\N}{\mathbb{N}}
\newcommand{\R}{\mathbb{R}}

\newcommand{\Sph}{\mathbb{S}}

\newcommand{\sfb}{\mathsf{b}}

\newcommand{\sfG}{\mathsf{G}}

\newcommand{\bfB}{\mathbf{B}}

\newcommand{\fa}{\mathfrak{a}}

\newcommand{\fm}{\mathfrak{m}}
\newcommand{\fp}{\mathfrak{p}}

\newcommand{\ft}{\mathfrak{t}}

\newcommand{\slg}{\slashed{g}{}}

\newcommand{\ran}{\operatorname{ran}}

\renewcommand{\Re}{\operatorname{Re}}

\newcommand{\Id}{\operatorname{Id}}

\newcommand{\mathspan}{\operatorname{span}}
\newcommand{\supp}{\operatorname{supp}}

\newcommand{\tr}{\operatorname{tr}}

\newcommand{\diag}{\operatorname{diag}}

\newcommand{\dS}{{\mathrm{dS}}}

\newcommand{\Ups}{\Upsilon}

\newcommand{\eps}{\epsilon}

\newcommand{\hra}{\hookrightarrow}
\newcommand{\la}{\langle}

\newcommand{\ol}{\overline}
\newcommand{\pa}{\partial}
\newcommand{\dd}{{\mathrm d}}
\newcommand{\ra}{\rangle}

\newcommand{\ul}[1]{\underline{#1}{}}

\newcommand{\wh}{\widehat}

\newcommand{\ubar}[1]{\underaccent{\bar}#1}
\newcommand{\pfstep}[1]{$\bullet$\ \underline{\textit{#1}}}
\newcommand{\pfsubstep}[2]{{\bf#1}\ \textit{#2}}

\newcommand{\bop}{{\mathrm{b}}}

\newcommand{\cp}{{\mathrm{c}}}

\newcommand{\Diff}{\mathrm{Diff}}

\newcommand{\Vb}{\cV_\bop}

\newcommand{\Diffb}{\Diff_\bop}

\newcommand{\Tb}{{}^{\bop}T}

\newcommand{\half}{{\tfrac{1}{2}}}

\newcommand{\loc}{{\mathrm{loc}}}
\newcommand{\CI}{\cC^\infty}
\newcommand{\CIdot}{\dot\cC^\infty}

\newcommand{\CIc}{\cC^\infty_\cp}

\newcommand{\Hb}{H_{\bop}}

\newcommand{\Riem}{\mathrm{Riem}}
\newcommand{\Ric}{\mathrm{Ric}}

\newcommand{\bhm}{\fm}
\newcommand{\bha}{\fa}

\newcommand{\openbigpmatrix}[1]
  {%
    \def\@bigpmatrixsize{#1}%
    \addtolength{\arraycolsep}{-#1}%
    \begin{pmatrix}%
  }
\newcommand{\closebigpmatrix}
  {%
    \end{pmatrix}%
    \addtolength{\arraycolsep}{\@bigpmatrixsize}%
  }

\newlength{\enummargin}

\newcommand{\usref}[1]{{\upshape\ref{#1}}}

\DeclareGraphicsExtensions{.mps}

\makeatletter
\newcommand*{\fwbw}[1]{\expandafter\@fwbw\csname c@#1\endcsname}
\newcommand*{\@fwbw}[1]{\ifcase #1 \or {\rm fw}\or {\rm bw}\fi}
\AddEnumerateCounter{\fwbw}{\@fwbw}
\makeatother

\begin{document}

\title[Stability of the expanding region of KdS]{Stability of the expanding region of Kerr--de~Sitter spacetimes and smoothness at the conformal boundary}

\date{\today}

\begin{abstract}
  We give a new proof of the recent result by Fournodavlos--Schlue on the nonlinear stability of the expanding region of Kerr--de~Sitter spacetimes as solutions of the Einstein vacuum equations with positive cosmological constant. Our gauge is a modification of a generalized harmonic gauge introduced by Ringstr\"om in which the asymptotic analysis becomes particularly simple. Due to the hyperbolic character of our gauge, our stability result is local near points on the conformal boundary. We show furthermore that, in yet another gauge, the conformally rescaled metric is smooth down to the future conformal boundary, with the coefficients of its Fefferman--Graham type asymptotic expansion featuring a mild singularity at future timelike infinity of the black hole.
\end{abstract}


\author{Peter Hintz}
\address{Department of Mathematics, ETH Z\"urich, R\"amistrasse 101, 8092 Z\"urich, Switzerland}
\email{peter.hintz@math.ethz.ch}

\author{Andr\'as Vasy}
\address{Department of Mathematics, Stanford University, CA 94305-2125, USA}
\email{andras@math.stanford.edu}

\maketitle

\section{Introduction}
\label{SI}

We study the stability of expanding regions of solutions of the Einstein vacuum equations
\begin{equation}
\label{EqIEin}
  \Ric(g) - \Lambda g = 0
\end{equation}
where the cosmological constant $\Lambda$ is positive; we fix $\Lambda=3$ (which can always be achieved by scaling). Here $g$ is a Lorentzian metric (with signature $(-,+,+,+)$) on a 4-dimensional smooth manifold $M^\circ$. The basic example is the de~Sitter solution
\begin{equation}
\label{EqIdS}
  M^\circ = (0,\infty)_\tau \times \R^3_x,\qquad
  g_\dS = \frac{-\dd\tau^2+\dd x^2}{\tau^2}.
\end{equation}
The conformal rescaling $\tau^2 g_\dS=-\dd\tau^2+\dd x^2$ is smooth down to the boundary of
\[
  M:=[0,\infty)_\tau\times\R^3_x,
\]
which is called the \emph{conformal boundary}. The de~Sitter metric is often encountered in a different coordinate system $\tilde t=-\frac12\log(|x|^2-\tau^2)$, $\tilde r=\frac{|x|}{\tau}$, $\omega=\frac{x}{|x|}$, where it takes the form
\begin{equation}
\label{EqIgdS}
  g_\dS = -(\tilde r^2-1)^{-1}\,\dd\tilde r^2 + (\tilde r^2-1)\,\dd\tilde t^2 + \tilde r^2\slg,
\end{equation}
with $\slg$ is the standard metric on $\Sph_\omega^2$. These coordinates are valid in the expanding region $|x|>\tau$. (In the \emph{static region} $|x|<\tau$, setting $t=-\frac12\log(\tau^2-|x|^2)$ yields the same expression for the metric.)

The Schwarzschild--de~Sitter (SdS) metric describes a (non-rotating) black hole in de~Sitter space. The metric depends on a mass parameter $\bhm\in\R$ and is given by
\[
  g_\bhm = -\Bigl(\tilde r^2-\frac{2\bhm}{\tilde r}-1\Bigr)^{-1}\,\dd\tilde r^2 + \Bigl(\tilde r^2-\frac{2\bhm}{\tilde r}-1\Bigr)^{-1}\,\dd\tilde t^2 + \tilde r^2\slg
\]
in the expanding region, which is the region where $\tilde r$ is larger than the largest real root of $\tilde r^2-\frac{2\bhm}{\tilde r}-1$. Comparing this expression with~\eqref{EqIgdS}, the mass $\bhm$ thus contributes metric coefficients of relative size $\cO(\tilde r^{-3})$ as $r\to\infty$. In the coordinates $z=(\tau,x)$, one finds that
\begin{equation}
\label{EqISdS}
  g_\bhm = g_\dS + h_{\mu\nu}\frac{\dd z^\mu}{\tau}\frac{\dd z^\nu}{\tau}
\end{equation}
where $h_{\mu\nu}=h_{\mu\nu}(\tilde r)=\cO(\tilde r^{-3})$; in fact, $\tilde r^3 h_{\mu\nu}$ is smooth in $\tilde r^{-1}=\frac{\tau}{|x|}$ near $\tilde r=\infty$. The set (in an appropriate compactification, introduced below) where $\tilde t<\infty$, $\tilde r^{-1}=0$ defines the conformal boundary of the SdS black hole spacetime. A similar construction can be performed for the more general Kerr--de~Sitter (KdS) metric $g_\sfb$, $\sfb=(\bhm,\bha)$, describing a black hole of mass $\bhm$ and specific angular momentum $\bha$ in de~Sitter space.

We interpret this geometrically as follows: we blow up the point $(\tau,x)=(0,0)$ in $M$ to define a new manifold with corners $\breve M$; a neighborhood of the conformal boundary of $\breve M$ is then covered by the chart
\begin{equation}
\label{EqIbreveM}
  [0,\infty)_\rho \times [0,\infty)_R \times \Sph^2_\omega \subset \breve M,\qquad \rho:=\frac{\tau}{|x|}=\tilde r^{-1},\quad R:=|x|,
\end{equation}
with the conformal boundary being the interior of $\cI^+:=\{0\}\times[0,\infty)\times\Sph^2$, while the interior of $\cK:=[0,\infty)\times\{0\}\times\Sph^2$ contains all points at $\tilde t=\infty$ of the level sets of $\tilde r$ (as is evident from $R=\frac{\tilde r e^{-\tilde t}}{\sqrt{\tilde r^2-1}}\approx e^{-\tilde t}$). See Figure~\ref{FigIMfd}. Thus, $\cK$ is a blown-up version of $i^+$.\footnote{We use the letter $\cK$ since $i^+$ is too similar to $\cI^+$.}

\begin{figure}[!ht]
\centering
\includegraphics{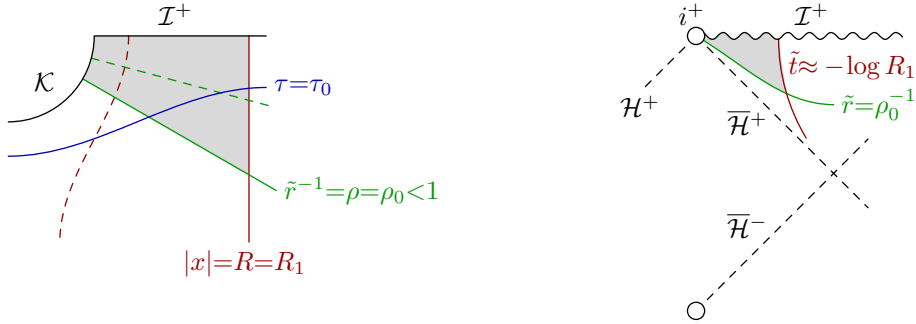}
\caption{\textit{On the left:} the manifold $\breve M$, with level sets of the coordinates $\tau$ (blue) and $|x|=R$ (red) on de~Sitter space; the level sets of $\tilde t$ in the region above the solid line $\rho=\rho_0<1$ (green) are approximately the same as those of $R$ and thus not shown. The level sets of $\rho=\tilde r^{-1}$ are shown in green. The highlighted region is part of the expanding region of SdS. \textit{On the right:} part of the Penrose diagram of SdS, with the part of the expanding region from the left highlighted.}
\label{FigIMfd}
\end{figure}

The point is that while the asymptotic behavior of the de~Sitter metric is the same at all points of the conformal boundary, the description of the asymptotic behavior of a KdS metric near its conformal boundary necessarily involves two asymptotic regimes ($\cI^+$ and $\cK$).

We shall study the initial value problem for~\eqref{EqIEin} when the initial data are given on a level set $\rho=\rho_0$ and asymptote to those of the KdS metric $g_\sfb$ as $R\to 0$ (i.e.\ as one approaches $\cK\approx i^+$). Recall that initial data are the first and second fundamental forms $\gamma,k$ of $\{\rho=\rho_0\}$, and they are subject to the constraint equations
\[
  R_\gamma-|k|_\gamma^2+(\tr_\gamma k)^2-2\Lambda=0,\qquad
  \delta_\gamma k+\dd\tr_\gamma k=0.
\]
(Here $R_\gamma$ is the scalar curvature, and $(\delta_\gamma k)_\mu=-\gamma^{\kappa\lambda}k_{\mu\kappa;\lambda}$.) Note that in the de~Sitter and KdS geometries, the endpoints of future causal curves starting at a point on $\rho=\rho_0$ lie in the interior of $\cI^+$ and are thus far from $\cK$. It is thus natural to expect that the spacetime metric evolving from such initial data still asymptotes to the same KdS metric $g_\sfb$ at $\cK$. On the other hand, away from $\cK$, we are, in a sense, studying a perturbation of de~Sitter space. (This is rigorously true when $R\geq R_1>0$ and $\rho=\rho_0\ll 1$ is small depending on $R_1$.) Thus classical theorems by Friedrich \cite{FriedrichDeSitterPastSimple} and Ringstr\"om \cite{RingstromEinsteinScalarStability} allow one to control the evolution of this part of the initial data. (We recall these results below.)

\begin{figure}[!ht]
\centering
\includegraphics{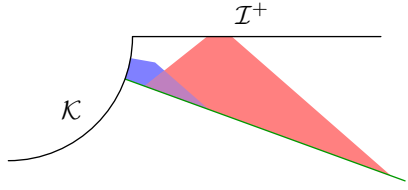}
\caption{Initial data are given on the green (spacelike) hypersurface. In the blue region, one can easily control the evolution of asymptotically-KdS data: the metrics remain asymptotic to the same KdS metric at $\cK$. In the red region, stability results for de~Sitter space apply. In a neighborhood of the corner, new methods are needed.}
\label{FigIRegions}
\end{figure}

The main difficulty is thus to control the evolving spacetime metric near the corner $\cK\cap\cI^+$; see Figure~\ref{FigIRegions}. This was first overcome in the recent work of Fournodavlos--Schlue \cite{FournodavlosSchlueExpanding} who considered initial data posed on a cylinder $\R_{\tilde t}\times\Sph_\omega^2$ on which the data asymptote to KdS data near $\{\pm\infty\}\times\Sph^2$ (see Figure~\ref{FigIData} below). Concretely, \cite[Theorem~1]{FournodavlosSchlueExpanding} shows that, under a smallness condition on the initial data relative to some fixed KdS data (measured in a Sobolev space with exponential weights as $|\tilde t|\to\infty$), the evolving spacetime metric $g$ can be written in the form
\[
  g(s,x) = -\Phi(s,x)\dd s^2 + g_{i j}(s,x)\dd x^i\,\dd x^j,\qquad s\geq 0,
\]
where $x=(\tilde t,\omega)$ denotes points on the cylinder $\R\times\Sph^2$. (The rough translation to present notation is $\rho\sim e^{-s}$ and $R\sim e^{-\tilde t}$.) Here $g_{i j}(s,x)=g_{i j}^\infty(x)e^{2 s}$ plus a $\cO(1)$ remainder as $s\to\infty$, while $\Phi(s,x)=1+e^{-2 s}\Phi^\infty(x)$ plus a $\cO(e^{-4 s})$ remainder. Furthermore, the various pieces of $\Phi,g_{i j}$, including $\Phi^\infty,g_{i j}^\infty$, differ from their KdS reference values by decaying amounts as $\tilde t\to\pm\infty$, i.e.\ as one approaches either of the two KdS black holes. This entails convergence to the de~Sitter type metric $-\dd s^2+e^{2 s}g_{i j}^\infty(x)\dd x^i\,\dd x^j$ (which for $g_{i j}^\infty=\delta_{i j}$ is the same as~\eqref{EqIdS} with $\tau=e^{-s}$) as $s\to\infty$. We highlight two features of the result and approach of \cite{FournodavlosSchlueExpanding}.
\begin{enumerate}
\item A parabolic gauge is used in which $\Phi$ is coupled to the mean curvature of the slices $s={\rm const}$. The non-hyperbolic nature of this gauge explains why Fournodavlos--Schlue work with a complete initial data set. It would be interesting to see if the arguments in \cite{FournodavlosSchlueExpanding} can be adapted to handle incomplete initial data sets.
\item Passing to $\rho=e^{-s}$, we have $\rho^2 g=-\dd\rho^2+g_{i j}^\infty(x)\dd x^i\,\dd x^j+\cO(\rho)$; thus, the conformally rescaled metric is Lipschitz down to $\tau=0$, but no higher order regularity is obtained (and it is not clear how much regularity one can expect in the chosen gauge).
\end{enumerate}

We revisit the stability problem of the expanding region of KdS spacetimes with the following goals in mind:
\begin{enumerate}
\item[(1*)] We use a generalized wave coordinate condition closely related to that of \cite{RingstromEinsteinScalarStability}. The hyperbolic character of this gauge condition allows us to prove a localized stability result.
\item[(2*)] We show how to upgrade the rough asymptotic control on the metric arising in the basic stability proof to smoothness of the conformally rescaled metric.
\end{enumerate}
Furthermore,
\begin{enumerate}
\setcounter{enumi}{2}
\item[(3*)] we introduce a robust framework for analyzing wave equations near the corner $\cI^+\cap\cK$; see~\S\ref{SssIExSdS}. This includes (higher order) energy estimates as well as a simple linear algebra mechanism (based on \emph{indicial roots}) for obtaining sharp decay and asymptotic expansions.
\end{enumerate}

In order to state our main result, we define the following norm for functions $u=u(R,\omega)$ defined for $R\leq R_0$:
\begin{equation}
\label{EqIHb}
  \|u\|_{R^\alpha\Hb^k}^2 := \sum_{j+|\gamma|\leq k} \int_{\Sph^2}\int_0^{R_0} | R^{-\alpha} (R\pa_R)^j \sV^\gamma u(R,\omega) |^2\,\frac{\dd R}{R}\,\dd\slg;
\end{equation}
here $\sV$ is the set of vector fields on $\Sph^2$ which generate rotations around the three coordinate axes. Since $R\sim e^{-\tilde t}$, the regularity here is regularity in $\tilde t$ and the angular variables.\footnote{The relationship of this norm with pointwise bounds is as follows. For $k\geq 2$, $\|u\|_{R^\alpha\Hb^k}<\infty$ implies $|u|\lesssim R^\alpha$; see Lemma~\ref{Lemma0bSob}. Conversely, $R^\alpha$ upper bounds for $u$ and $k$ of its derivatives imply $u\in R^{\alpha'}\Hb^k$ for all $\alpha'<\alpha$.} Below, we use the same notions for norms of tensors, and mean by that the sum of norms of their coefficients in smooth coordinate charts. We write $R^\alpha\Hb^k$ for the space of all functions with finite $\|\cdot\|_{R^\alpha\Hb^k}$-norm, and set $R^\alpha\Hb^\infty=\bigcap_{k\in\N_0}R^\alpha\Hb^k$. Furthermore, in order to ensure compatibility of the KdS metric $g_\sfb$ with the precise asymptotic expansion of the dynamical metric $g$ at $\cI^+$, we denote by
\[
  g_\sfb^{\rm FG}
\]
a presentation of the KdS metric $g_\sfb$ which is in \emph{Fefferman--Graham form} at $\cI^+$ (a notion we explain after the statement of the Theorem).

\begin{thm}[Main theorem, rough version]
\label{ThmIMain}
  Let $R_0>0$, and let $\rho_0>0$ be such that $\Sigma^\circ_{\rho_0,R_0}:=\{\rho_0\}\times(0,R_0]\times\Sph^2_\omega$ is spacelike for the KdS metric $g_\sfb$; denote the initial data on $\Sigma^\circ_{\rho_0,R_0}$ induced by $g_\sfb$ by $\gamma_\sfb,k_\sfb$. Suppose $\gamma,k$ are initial data on $\Sigma^\circ_{\rho_0,R_0}$ (i.e.\ solutions of the constraint equations) so that $\tilde\gamma:=\gamma-\gamma_\sfb$, $\tilde k:=k-k_\sfb$ lie in $R^\alpha\Hb^\infty$, and have $R^\alpha\Hb^d$-norms $<\eps$ where $\eps>0$ is small and $d\in\N$ is large. Then the maximal globally hyperbolic development of the data $\gamma,k$ contains a region isometric to
  \[
    (\Omega_{\rho_0,R_0}^\circ,g),
  \]
  where:
  \begin{enumerate}
  \item $\Omega_{\rho_0,R_0}^\circ$ is the domain defined by the inequalities $\rho\leq\rho_0$, $\rho R\geq\rho_0 R_0-\frac12(R_0-R)$, and the boundary hypersurfaces $\Sigma_{\rho_0,R_0}^\circ$ and $\Sigma_{\rho_0,R_0}^{+,\circ}:=\{\rho R=\rho_0 R_0-\frac12(R_0-R)\}$ are spacelike for $g$;
  \item the metric $g$ is of the form $g=g_\sfb^{\rm FG}+h$, where, in the frame $\tau\pa_\tau$, $\tau\pa_{x^i}$ ($i=1,2,3$),\footnote{We recall that $x=R\omega$ and $\tau=R\rho$.} the coefficients of $h=h(\rho,R,\omega)$ are smooth down to $\rho=0$ and of class $R^\alpha\Hb^\infty$ in $(R,\omega)$.
  \end{enumerate}
  More precisely, there exist $h_m=h_{m,i j}\frac{\dd x^i}{\tau}\frac{\dd x^j}{\tau}$ for $m=0,2,3,4,\ldots$ with $h_{m,i j}=h_{m,i j}(R,\omega)\in R^\alpha\Hb^\infty$ so that for all $N\in\N$ we have
  \begin{equation}
  \label{EqIMainExp}
    g - \biggl( g_\sfb^{\rm FG} + h_0 + \sum_{m=2}^N \rho^m h_m \biggr) = \cO(\rho^{N+1}),
  \end{equation}
  in the sense that this is the restriction of an element of $\rho^{N+1}\CI\bigl([0,\rho_0];R^\alpha\Hb^\infty([0,R_0]\times\Sph^2)\bigr)$ to $\Omega_{\rho_0,R_0}^\circ$. Furthermore:
  \begin{enumerate}
  \item $h_2$ is nonzero unless the metric $g_{(0)}:=\dd x^2+h_{(0)}$, $h_{(0)}:=h_{0,i j}(|x|,\frac{x}{|x|})\,\dd x^i\,\dd x^j$, on $\{0<|x|\leq R_0\}\subset\R^3$ is flat;
  \item denote by $g_3$ the $\rho^3$ coefficient of $g$. Then the tensor $g_{(3)}:=g_{3,i j}(|x|,\frac{x}{|x|})\dd x^i\,\dd x^j$ is a weighted TT (transverse-traceless) tensor, meaning $\tr_{g_{(0)}}g_{(3)}=0$ and $\delta_{g_{(0)}}(|x|^{-3}g_{(3)})=0$.
  \end{enumerate}
\end{thm}

See Theorem~\ref{ThmEStabEin} for the full result, and Figure~\ref{FigIData} for an illustration of the domain on which we work. Our approach to the proof is discussed in~\S\S\ref{SsIEx}--\ref{SsIE} below. The broader context of the precise asymptotic expansion~\eqref{EqIMainExp}, the weighted TT property, and the sense in which Theorem~\ref{ThmIMain} is optimal are explained in~\S\ref{SsIdS}. We immediately point out that the expansion~\eqref{EqIMainExp} is devoid of logarithmic terms; it therefore implies in particular that
\[
  \parbox{0.9\textwidth}{\emph{the conformally rescaled metric $\tau^2 g$, expressed in the frame $\tau\pa_\tau,\tau\pa_{x^1},\tau\pa_{x^2},\tau\pa_{x^3}$, is smooth across $\cI^+=\{\rho=0\}$.}}
\]

We now explain what it means for $g_\sfb^{\rm FG}$ to be in \emph{Fefferman--Graham form}: $g_\sfb^{\rm FG}$ has an expansion
\begin{equation}
\label{EqIgbFG}
  g_\sfb^{\rm FG}\sim g_{\sfb,0}+\sum_{m\geq 2}\rho^m g_{\sfb,m},\qquad \rho\to 0
\end{equation}
(i.e.\ equality of Taylor series at $\rho=0$), where $g_{\sfb,m}=(g_{\sfb,m})_{i j}\frac{\dd x^i}{\tau}\frac{\dd x^j}{\tau}$, with each $(g_{\sfb,m})_{i j}$ a smooth function on $[0,\infty)_R\times\Sph^2_\omega$, and with $g_{\sfb,(3)}:=(g_{\sfb,3})_{i j}\dd x^i\,\dd x^j$ a weighted TT tensor with respect to $\dd x^2$. (The existence of such a presentation $g_\sfb^{\rm FG}$ of the KdS metric is shown in Proposition~\ref{PropStCI}\eqref{ItStCICI}.)

\begin{figure}[!ht]
\centering
\includegraphics{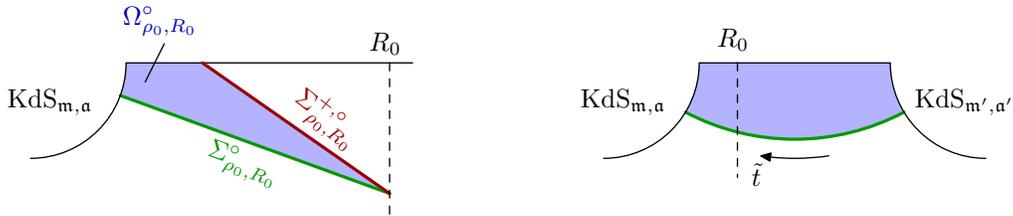}
\caption{\textit{On the left:} the Cauchy hypersurface $\Sigma_{\rho_0,R_0}^\circ$ (green), and the domain $\Omega_{\rho_0,R_0}^\circ$ (blue) on which the spacetime metric is controlled. \textit{On the left:} Theorem~\ref{ThmIMain}. \textit{On the right:} \cite[Theorem~1]{FournodavlosSchlueExpanding}.}
\label{FigIData}
\end{figure}

\begin{rmk}[Original of initial data]
\label{RmkIData}
  We do not concern ourselves here with the construction of initial data $\gamma,k$ satisfying the hypotheses of Theorem~\ref{ThmIMain}. Recall however that the nonlinear stability of the exterior region (more precisely, a neighborhood of the domain of outer communications---not to be confused with the cosmological region) of slowly rotating Kerr--de~Sitter black holes was proved in \cite{HintzVasyKdSStability,FangKdS} (with nontrivial initial data constructed in \cite[\S{11.3}]{HintzVasyKdSStability}). The region on which these stability results apply extends to any fixed finite value $\tilde r_+$ exceeding the radius of the cosmological horizon (with the required smallness of the initial data in the proofs of the references depending on $\tilde r_+$). The data induced by the spacetime metric on the corresponding level set $\rho=\rho_0=\tilde r_+^{-1}$ then satisfies the assumptions of Theorem~\ref{ThmIMain}. While subextremality is a crucial requirement for the nonlinear stability of KdS, the black hole parameters $\sfb=(\bhm,\bha)$ are \emph{unrestricted} in the setting of Theorem~\ref{ThmIMain}: it does not matter whether they are subextremal, extremal, superextremal, or even have negative mass, since only the asymptotic (de~Sitter) geometry as $\tilde r\to\infty$ (which is valid regardless of $\sfb$) matters for present purposes. We refer the reader to Remark~\ref{RmkEOrigInitial} for further discussion.
\end{rmk}

\begin{rmk}[Stability of de~Sitter space]
\label{RmkIdS}
  Our methods apply directly to the stability of (parts of) de~Sitter space in $(3+1)$-dimensions and thus yield a new proof of \cite{RingstromEinsteinScalarStability} (restricted to the case that the scalar field vanishes identically). In fact, our proof simplifies since one can work with standard Sobolev spaces on sets of bounded $x$: there is no more need for $R$-weights and b-regularity. More precisely, if $\tau_0>0$, then the spacetime evolving from sufficiently small and regular perturbations of the de~Sitter initial data at\footnote{For simplicity, we give ourselves plenty of room; posing data at $|x|<\tau_0+\delta$ for any fixed $\delta>0$ would suffice for the evolving spacetime to contain a piece of the conformal boundary.} $\{\tau=\tau_0,\ |x|<4\tau_0\}$ contains a region of the form
  \[
    \{ \tau\leq\tau_0,\ |x|<3\tau_0-2(\tau_0-\tau) \}
  \]
  equipped with a metric $g$ which has a full Taylor expansion
  \[
    g(\tau,x) \sim -\frac{\dd\tau^2}{\tau^2} + (\delta_{i j}+h_{0,i j}(x))\frac{\dd x^i}{\tau}\frac{\dd x^j}{\tau} + \sum_{m\geq 2} h_{m,i j}(x)\frac{\dd x^i}{\tau}\frac{\dd x^j}{\tau}
  \]
  at $\tau=0$, where $h_{m,i j}\in\CI(\R^3_x)$ and $h_{3,i j}(x)\dd x^i\,\dd x^j$ is transverse-traceless with respect to $(\delta_{i j}+h_{0,i j}(x))\dd x^i\,\dd x^j$. Also the global stability of de~Sitter space $-\dd t^2+(\cosh t)^2 g_{\Sph^3}$, as first proved in \cite{FriedrichDeSitterPastSimple}, follows from (a simpler version of) our arguments. --- We point out that we are able to obtain a conformally smooth solution by suitably modifying (in a constructive, and thus essentially explicit, manner) an already constructed solution in generalized harmonic gauge. In particular, we expect our approach to allow one to obtain sharp (Fefferman--Graham type) asymptotics for perturbations of de~Sitter space also in odd spacetime dimensions (where the existing results of \cite{FriedrichDeSitterPastSimple,AndersonStabilityEvenDS} do not apply, and logarithmic singularities are known to necessarily appear); we leave this to future work. Analogous results in the Riemannian setting of conformally compact Poincar\'e--Einstein metrics were obtained in \cite{ChruscielDelayLeeSkinnerConfCompReg}; see \cite{FeffermanGrahamAmbient,KichenassamyFeffermanGraham} for the analytic setting and \cite{AndersonBoundaryRegularity} for the $4$-dimensional case.
\end{rmk}

\begin{rmk}[Stability of larger regions]
\label{RmkILarger}
  The methods apply with only minor notational modifications to the initial data of \cite{FournodavlosSchlueExpanding}. The main change is that one now needs to work with two weights, one for each of the boundary hypersurfaces $\cK$ and $\cK'$ corresponding to future timelike infinity of the two KdS black holes (see Figure~\ref{FigIData}); the proof of our main energy estimate (Proposition~\ref{PropE2Reg}) is robust enough to handle this setting with only notational modifications. The conclusion is that on the blue region on the right in Figure~\ref{FigIData} we can put the dynamical spacetime metric into a form so that an expansion completely analogous to~\eqref{EqIMainExp} holds, where the coefficients $h_m$ are now elements of a doubly weighted Sobolev space on the cylinder $\R_{\tilde t}\times\Sph^2$. We remark that due to the domain-dependence of the gauge condition we use, we cannot simply patch together the local solutions which are produced when applying our proof of Theorem~\ref{ThmIMain} to various incomplete patches (such as $\Sigma_{\rho_0,R_0}^\circ$ or the data of Remark~\ref{RmkIdS}) of initial data, as the various local solutions are constructed in what might well be (slightly) different gauges; instead we must prove stability directly in the full desired region.
\end{rmk}

Valiente Kroon and collaborators have been developing an approach for studying the stability of the cosmological region of Schwarzschild--de~Sitter spacetimes based on an extension of Friedrich's conformal field equations \cite{FriedrichDeSitterPastSimple}. In Friedrich's equations, the conformal factor ($\tau$ in present notation) is one of the unknowns, and the equations (and their solutions) extend non-degenerately across the conformal boundary. This allowed Friedrich to reduce the global nonlinear stability of de~Sitter spacetime to a standard local-in-time result for his symmetric hyperbolic system. In the SdS case however, the conformal field equations cease to be regular at future timelike infinity $i^+$ of the SdS black hole (essentially because the SdS metric is not smooth in $\tau,x$ there, cf.\ the discussion of~\eqref{EqISdS}), and thus the results obtained using this approach are, at present, incomplete: \cite{GasperinValienteKroonAsymptoticSdS} constructs asymptotically SdS cosmological regions using an asymptotic initial value problem (closely related to \cite[Theorem~(3.2)]{FriedrichDeSitterPastSimple}), and \cite{MinucciValienteKroonSdSStab} controls solutions of initial value problems with near-SdS data in the domain of dependence of cylinders $[-T,T]_{\tilde t}\times\Sph^2$ in the notation of \cite{FournodavlosSchlueExpanding} and Figure~\ref{FigIData} which, in particular, does not contain a neighborhood of the corner $\cK\cap\cI^+$ (see also \cite[Figure~4]{MinucciValienteKroonSdSStab}).

Another possible avenue towards the stability of the cosmological region, based on geometric foliations and control of the Weyl tensor, was explored by Schlue in \cite{SchlueOpticaldS,SchlueWeylDecay}. For scalar waves propagating in the cosmological region of subextremal KdS spacetimes, decay results were obtained in \cite{SchlueCosmological}. Bernhardt \cite{BernhardtLinearExpanding} continued the study of linear scalar waves in this setting and obtained a partial asymptotic expansion at the conformal boundary (analogous to~\eqref{EqIMainExp} for $N=3$, and with $\alpha=0$) as well as a scattering result.

The plan for the remainder of this introduction is as follows.
\begin{enumerate}
\item In~\S\ref{SsIdS}, we put the asymptotic expansion~\eqref{EqIMainExp} into context and explain its optimality.
\item In~\S\ref{SsIEx}, we explain the basic ideas behind the analysis of wave equations near conformal boundaries, in both the de~Sitter type and KdS type settings.
\item In~\S\ref{SsIE}, we apply these ideas to the Einstein equations and explain our gauge choices, constraint damping, and the mechanism underlying the existence of the expansion~\eqref{EqIMainExp}.
\end{enumerate}

\subsection{Stability of de~Sitter space and asymptotics at the conformal boundary}
\label{SsIdS}

Friedrich's stability result for de~Sitter space \cite{FriedrichDeSitterPastSimple} demonstrates that small perturbations of de~Sitter initial data evolve into a spacetime $(M,g)$ where $M\subset\tilde M:=\R\times\Sph^3$ is given by the set $\{\Omega>0\}$, and $\Omega^2 g$ extends smoothly to $\bar M\subset\tilde M$ (and beyond). (Note that de~Sitter space itself is of this form for $\Omega=\cos s$ and $\Omega^2 g=-\dd s^2+g_{\Sph^3}$.) Moreover, \cite[Theorem~(3.2)]{FriedrichDeSitterPastSimple} establishes a 1-1 correspondence (at least near the conformal boundary) between such asymptotically simple solutions of the field equations and \emph{scattering data} defined at the future conformal boundary $S\subset\Omega^{-1}(0)$, also in cases where the spatial manifold $S$ is an arbitrary compact orientable 3-manifold: these data are a Riemannian metric $g_{(0)}$ and a TT tensor $g_{(3)}$ on $S$. (These are the restriction of $\Omega^2 g$ and certain components of the rescaled Weyl tensor of $g$ to $S$.)

Now, given such $g_{(0)},g_{(3)}$, it is a classical result by Fefferman--Graham \cite{FeffermanGrahamAmbient,FeffermanGrahamAmbientBook} that one can construct a formal solution of the Einstein vacuum equations~\eqref{EqIEin} of the Fefferman--Graham form (as described after Theorem~\ref{ThmIMain})
\begin{equation}
\label{EqIdSFG}
  g \sim \tau^{-2}\biggl( -\dd\tau^2 + g_{(0)}(x,\dd x) + \sum_{m\geq 2} \tau^m g_{(m)}(x,\dd x) \biggr),\qquad \tau\to 0;
\end{equation}
furthermore, the terms in the expansion~\eqref{EqIdSFG} are \emph{uniquely determined by $g_{(0)},g_{(3)}$}. It was moreover shown in \cite{RodnianskiShlapentokhRothmanSelfSimilar,HintzAsymptoticallydS} that this formal power series is the Taylor expansion of a true solution defined for $\tau<\tau_0(x)$ for some sufficiently small positive continuous function $\tau_0>0$. In combination (albeit in a rather indirect fashion), we can thus conclude that smooth perturbations of de~Sitter space are described by a metric of the form~\eqref{EqIdSFG}. (See also \cite{GrahamLeeConformalEinstein,KichenassamyFeffermanGraham,GurskySzekelyhidiLocalPoincareEinstein} for the construction of true solutions in the Riemannian setting.)

With this background, it is now clear that the description~\eqref{EqIMainExp} is optimal: the power series $g_\sfb^{\rm FG}+h_0+\sum_{m\geq 2}\rho^m h_m$ at $\rho=0$ is in Fefferman--Graham form and thus its coefficients are uniquely determined by the scattering data $g_{(0)}=\dd x^2+h_{(0)}$ and $g_{\sfb,(3)}+h_{(3)}$ in the notation of Theorem~\ref{ThmIMain} and \eqref{EqIgbFG} and the subsequent discussion. We remark that the weighted TT property of $h_{(3)}$, i.e.\ the TT property of $h^\natural_3(x,\dd x):=|x|^{-3}h_{(3)}(x,\dd x)$, becomes consistent with~\eqref{EqIdSFG} once we observe that this tensor appears in the Taylor expansion of $g$ at the conformal boundary $\cI^+$ via ($\tau^{-2}$ times) $\rho^3 h_{(3)}=(\rho|x|)^3 h^\natural_3=\tau^3 h^\natural_3$. 

\begin{rmk}[Black holes from scattering data]
\label{RmkIBHFromScat}
  The work of Mars--Pe\'on-Nieto \cite{MarsPeonNietoKdSFromScri} discusses the characterization of KdS metrics via their data at the conformal boundary. Also the construction of de~Sitter spacetimes containing several black holes in \cite{HintzGluedS,VerlemannGlueCharged} is, at least on a conceptual level, based on this scattering perspective.
\end{rmk}

It is then natural to make the following conjecture.

\begin{conj}[Scattering data]
\label{ConjIScatter}
  Fix KdS parameters $\sfb$. Suppose we are given tensors $h_{(m),i j}=h_{(m),i j}(R,\omega)\in R^\alpha\Hb^\infty$ for $m=0,3$ and $R\leq R_0$, with small norms in $R^\alpha\Hb^d$ for some sufficiently large $d$, so that $g_{(3)}:=g_{\sfb,(3)}+h_{(3)}$ is a weighted TT tensor with respect to $g_{(0)}:=\dd x^2+h_{(0)}$. Then there exists a solution $g$ of the Einstein vacuum equations in a neighborhood of $\{R\leq R_0\}\subset\cI^+$ with scattering data $g_{(0)},g_{(3)}$ which is asymptotic to the KdS metric with parameters $\sfb$ at $\cK$.
\end{conj}

Applying the uniqueness statements of \cite{HintzAsymptoticallydS} to the restriction of the scattering data $g_{(0)},g_{(3)}$ to $R\geq\delta>0$ and letting $\delta\searrow 0$, the solution $g$ in Conjecture~\ref{ConjIScatter} is easily seen to be necessarily unique (up to isometries) on an appropriate domain of dependence; see Figure~\ref{FigIScatter}. For the proof of a version of Conjecture~\ref{ConjIScatter} for linear scalar waves with $\alpha=0$, see \cite[Theorem~1.4]{BernhardtLinearExpanding}.

\begin{figure}[!ht]
\centering
\includegraphics{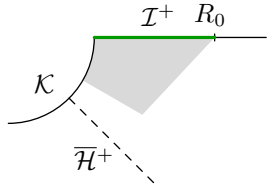}
\caption{The domain on which an asymptotically KdS metric is expected to exist, given scattering data on a part of the conformal boundary (highlighted in green). The dashed line indicates the cosmological horizon of KdS; to avoid blueshift instabilities when solving backwards, we stop at a positive distance from there.}
\label{FigIScatter}
\end{figure}

\subsection{Analysis near conformal boundaries}
\label{SsIEx}

Our analysis of the Einstein equations will build on a perspective which turns asymptotic analysis near conformal boundary into problems in linear algebra. In de~Sitter type settings, this perspective already played an important role in many works including \cite{FeffermanGrahamAmbient,RendallLambdaAsymptotics,VasyWaveOndS,HintzXiedS}, and in elliptic settings in \cite{MazzeoMelroseHyp,GrahamLeeConformalEinstein,MazzeoEdge}. A novelty of the present paper is an extension of this perspective which allows one to deal with singularities on the conformal boundary of the type given by future timelike infinity $i^+$ of KdS.

\subsubsection{de~Sitter space}
\label{SssIExdS}

For concreteness, we study the following toy model. Fix a smooth Riemannian metric $(g_{(0)})_{i j}(x)\dd x^i\,\dd x^j$ on the torus $x\in\TT^3$, and consider
\begin{equation}
\label{EqIEx}
  \Box\phi := \bigl(-(\tau\pa_\tau)^2 + 3\tau\pa_\tau + \tau^2 g_{(0)}^{i j}(x)\pa_{x^i}\pa_{x^j} \bigr)\phi(\tau,x) = 0.
\end{equation}
(For $g_{(0)}^{i j}=\delta^{i j}$, this is the covariant wave operator for the metric~\eqref{EqIdS}.) The initial data, at $\tau=1$, say, are assumed to be smooth in $x$; one can also allow for a nontrivial right hand side (with appropriate regularity and decay requirements as $\tau\to 0$), though we shall not do so here for the sake of exposition. If one only keeps the $\tau\pa_\tau$ terms of $\Box$, one obtains the \emph{indicial operator}, here
\[
  I(\tau\pa_\tau):=-(\tau\pa_\tau)^2+3\tau\pa_\tau,
\]
which is a constant coefficient regular-singular (or Fuchsian) ordinary differential equation (ODE). The \emph{indicial family} is its characteristic polynomial, so $I(\lambda):=-\lambda^2+3\lambda$, and the \emph{indicial roots} are its roots, $\lambda=0,3$. Since therefore $\Box(\tau^\lambda u)=\cO(\tau^{\lambda+1})$ for $\lambda=0,3$ and any $u=u(x)$, one anticipates that solutions of~\eqref{EqIEx} have the form $\phi(\tau,x)=\tau^0\phi_0(x)+\tau^3\phi_3(x)+\cdots$, where `$\cdots$' indicates terms that (at least in Taylor series at $\tau=0$) can be computed from $\phi_0,\phi_3$.

This heuristic can be made rigorous, as shown in the references above. We outline here a two-step strategy which generalizes easily to the more singular setting of the present paper.

\textit{Step~1.1. Basic energy estimate.} One can easily obtain a spacetime energy bound
\begin{equation}
\label{EqIExEn0}
  \|\phi\|_{\tau^{-N}H_0^1}^2 := \int_{\TT^3}\int_0^1 |\tau^N \phi|^2 + |\tau^N \tau\pa_\tau\phi|^2 + |\tau^N \tau\pa_x\phi|^2 \frac{\dd\tau}{\tau}\,\dd x \lesssim {\rm data}
\end{equation}
for some $N>0$. The particular value of $N$ will be of no concern to us; this will be advantageous when passing to tensorial equations for which precise energy estimates may be more difficult to obtain (e.g.\ due to computational complexities). The notation $\tau^{-N}H_0^1$ reflects the weight ($\phi$ is allowed to grow like $\tau^{-N}$) and the notion of regularity ($\tau\pa_\tau$, $\tau\pa_x$), which is \emph{0-regularity} on $[0,\infty)_\tau\times\TT^3_x$ in the parlance of Mazzeo--Melrose \cite{MazzeoMelroseHyp}. Importantly, the same value of $N$ works if $\Box$ is perturbed by terms which, relative to $\tau\pa_\tau$, $\tau\pa_x$, decay at $\tau=0$; we write such terms as $\cO(\tau)$ below.

\medskip

\textit{Step~1.2. Higher regularity.} Discarding $x$-derivatives in the above heuristic requires $\tau\pa_x\phi$ to be of lower order (in the sense of decay) than $\tau\pa_\tau\phi$. The estimate~\eqref{EqIExEn0} does not entail this. We thus need to improve~\eqref{EqIExEn0} to a higher regularity estimate, \emph{with the same weight $\tau^{-N}$}, in which we control $\pa_x$-derivatives of $\phi$. Concretely, we claim that
\begin{equation}
\label{EqIExEnk}
  \|\phi\|_{\tau^{-N}H_{0;\bop}^{1;k}}^2 := \sum_{j+|\alpha|\leq k} \|(\tau\pa_\tau)^j\pa_x^\alpha\phi\|_{\tau^{-N}H_0^1}^2 \lesssim {\rm data}.
\end{equation}
(The notation $H_{0;\bop}^{1;k}$ reflects the derivative types: we now control $k$ \emph{b-derivatives} ($\tau\pa_\tau,\pa_x$) in the parlance of \cite{MelroseTransformation,MelroseMendozaB,MelroseAPS} in addition to $1$ 0-derivative.) We accomplish this via a highly robust commutation argument which only relies on the \emph{structural properties} of $\tau\pa_\tau,\tau\pa_x$ and their relationships with $\tau\pa_\tau,\pa_x$. We illustrate this only for $k=1$: we then have the commuted equations
\begin{alignat*}{2}
  \Box(\tau\pa_\tau\phi) &= \tau\pa_\tau\Box\phi - 2\tau^2 g_{(0)}^{i j}(x)\pa_{x^i}\pa_{x^j}\phi &&= \cO(\tau\cdot\tau\pa_x) \pa_x\phi, \\
  \Box(\pa_x\phi) &= \pa_x\Box\phi + \tau^2(\pa_x g_{(0)}^{i j})\pa_{x^i}\pa_{x^j}\phi &&= \cO(\tau\cdot\tau\pa_x) \pa_x\phi.
\end{alignat*}
Therefore, the vector $\Phi:=(\tau\pa_\tau\phi,\pa_x\phi)$ satisfies a principally scalar system of equations which is
\[
  (\Box + \cO(\tau))\Phi = 0.
\]
To this equation, the estimate~\eqref{EqIExEn0} applies \emph{with the same value of $N$}, giving~\eqref{EqIExEnk} for $k=1$.

\begin{rmk}[General situation: triangular structure]
\label{RmkIExTriangle}
  If one replaced $3\tau\pa_\tau$ in~\eqref{EqIEx} by $(3+a(x))\tau\pa_\tau$, the equation for $\pa_x\phi$ would have an additional term $(\pa_x a)\tau\pa_\tau\phi$ on the right hand side, whose coefficients do \emph{not} decay. Instead, one now gets a strictly lower triangular system of the schematic form
  \begin{equation}
  \label{EqIExTriangle}
    \left(\begin{pmatrix} \Box & 0 \\ \pa_x a & \Box \end{pmatrix} + \cO(\tau)\right)\Phi = 0.
  \end{equation}
  For this system, one can still prove~\eqref{EqIExEn0} \emph{for the same $N$}, essentially by using~\eqref{EqIExEn0} for each component of $\Phi$ separately (with the `data' term now involving a norm on a spacetime source term) and taking a weighted sum of the two estimates to absorb the size of $\pa_x a$. In the tensorial equations of interest in this paper, we do encounter variable coefficients of this type (although they will leave the indicial roots unaffected, unlike $(3+a)\tau\pa_\tau$ here).
\end{rmk}

\medskip

\textit{Step 2. Decay.} Having arbitrarily many b-derivatives---in particular, $x$-derivatives---under control, we can now justify putting $x$-derivatives of $\phi$ on the right hand side of~\eqref{EqIEx}. This leads to
\[
  I(\tau\pa_\tau)\phi = \cO(\tau^2)\pa_x^2\phi.
\]
If $\phi\in\tau^{-N}H_{0;\bop}^{1;k}$, the right hand side lies in $\tau^{-N+2}H_{0;\bop}^{1;k-2}$. Integrating this ODE in $\tau$, with $x$ acting merely as a parameter, gives $\phi\in\tau^{-N+2}H_{0;\bop}^{1;k-2}$, when the indicial roots $0,3$ fall outside of the interval $[-N,-N+2]$. Iterating this argument thus allows one to show $\phi\in\tau^{-\delta}H_{0;\bop}^{1;k-J}$ for some $J$ (depending only on the growth rate $N$ allowed for by the basic energy estimate) for any $\delta>0$; repeating the argument one more time, the indicial root $0$ enters and produces
\begin{equation}
\label{EqIExPhi}
  \phi(\tau,x) \equiv \phi_0(x) \bmod \tau^{2-\delta}H_{0;\bop}^{1;k-J-2},\qquad \phi_0\in H^{k-J-2}(\TT^3).
\end{equation}
One can easily continue this scheme further and extract a full asymptotic expansion for $\phi$. This is the place where logarithmic terms can arise due to repeated roots and integer coincidences; see in particular~\cite{VasyWaveOndS,CicortasScatteringdS,BernhardtLinearExpanding}. In our proof of a basic nonlinear stability result (see Theorem~\ref{ThmESol}), it will suffice to get a leading order term plus a decaying remainder, and therefore we stop here. (The Fefferman--Graham asymptotics are largely obtained via formal arguments at $\rho=0$ and thus of a different flavor; see~\S\ref{SssIEFG} below.)

\medskip

We emphasize that the proofs of the (higher order) energy estimates in Step~1 do not rely on any particular structure of the underlying operator beyond the fact that it is built from $\tau\pa_\tau,\tau\pa_x$; in particular, they apply immediately to tensorial wave equations (with arbitrary lower order terms) as well. Similarly, the arguments in Step~2 only rely on the properties of the indicial family $I(\lambda)$, which in the case of tensorial equations is a polynomial with values in square matrices of the appropriate dimension; the asymptotic behavior of $\phi$ is then determined by the indicial roots ($\lambda\in\C$ with $\det I(\lambda)=0$) and the corresponding spaces $\ker I(\lambda)$ of \emph{indicial solutions}.\footnote{In full generality, the indicial roots may depend on $x$, as they do in the setting of Remark~\ref{RmkIExTriangle}. This does not happen in the settings considered in the present paper. Asymptotic expansions in the presence of variable indicial roots are studied in \cite{KrainerMendozaBundle}.}

\medskip

This approach to linear decay estimates can easily be combined with simple nonlinear methods (Moser-type product estimates, tame estimates combined with bootstrap arguments or a Nash--Moser iteration) to show the small data global well-posedness of suitable nonlinear equations, where `suitable' refers to the requirement that the nonlinear terms, when applied to $\phi$ of the form~\eqref{EqIExPhi} given by linear theory, produce decaying spacetime terms. (A simple toy example is $\Box\phi=(\tau\pa_\tau\phi)(\tau\pa_x\phi)$.) The proof of tame estimates in the present paper is essentially straightforward, although it does cause a significant bookkeeping overhead; thus we shall not comment on these standard nonlinear issues in the remainder of this introduction, instead referring the interested reader directly to~\S\ref{SsE2}.

\subsubsection{Expanding regions of de~Sitter black hole spacetimes.}
\label{SssIExSdS}

We wish to apply a similar approach in the expanding region of SdS and KdS spacetimes. As a consequence of the structure~\eqref{EqISdS}, the scalar wave operator $\Box_{g_\bhm}$ will again be a 0-differential operator, i.e.\ built from $\tau\pa_\tau$, $\tau\pa_x$, but its coefficients will no longer be smooth in $(\tau,x)$ (i.e.\ on $M$) but only in $\rho,R,\omega$ (i.e.\ on $\breve M$), as defined in~\eqref{EqIbreveM}. The expressions for $\tau\pa_\tau$, $\tau\pa_x$ in terms of $\rho,R,\omega$ are linear combinations of
\begin{equation}
\label{EqIEx0b}
  \rho\pa_\rho,\ \rho R\pa_R,\ \rho\pa_\omega,
\end{equation}
where we schematically write $\pa_\omega$ for derivatives on $\Sph^2$. (In $R\gtrsim 1$, i.e.\ in the cosmological region far from $i^+\approx\cK$, these are $\rho\pa_\rho\sim\tau\pa_\tau$ and $\rho\pa_R,\rho\pa_\omega\sim\tau\pa_x$; on the other hand, in $\rho\gtrsim 1$, i.e.\ far from the conformal boundary, they are $\pa_\rho\sim\pa_{\tilde r}$, $R\pa_R\sim\pa_{\tilde t}$, and $\pa_\omega$, i.e.\ the natural derivatives for analysis in spatially compact regions of a(n asymptotically) stationary black hole spacetime.) The wave operator is thus of the form
\[
  \Box_{g_\bhm} = \sum_{i+j+|\gamma|\leq 2} \ell_{i j\gamma}(\rho,R,\omega) (\rho\pa_\rho)^i(\rho R\pa_R)^j(\rho\pa_\omega)^\gamma,\qquad \ell_{i j\gamma}\in\CI([0,\bar\rho]_\rho\times[0,R_0]_R\times\Sph^2).
\]
(Here $\bar\rho,R_0>0$. The key point is that the coefficients $\ell_{i j\gamma}$ are smooth down to $\rho=0$ and $R=0$.)

Analogously to~\eqref{EqIEx}, we consider an initial value problem for
\[
  \Box_{g_\bhm}\phi = 0,
\]
with initial data posed at $\rho=\rho_0>0$ (i.e.\ $\tilde r=\rho_0^{-1}$), which we assume to be a spacelike hypersurface as in Figure~\ref{FigIMfd}; this happens for sufficiently large $\tilde r$. The analogues of the 2 steps in~\S\ref{SssIExdS} are as follows. (For easier readability, we are imprecise with the specification of domains of integrations etc.\ below; they are to be taken according to domain of dependence considerations.)

\medskip

\textit{Step~1.1. Basic energy estimate.} Since we now need to distinguish weights near the black hole ($\cK$) from weights at the conformal boundary ($\cI^+$), we work with doubly weighted norms
\[
  \|\phi\|_{\rho^{-N}R^\alpha H_{0,\bop}^1}^2 := \int_{\Sph^2}\int_0^{R_0}\int_0^{\rho_0} |\rho^N R^{-\alpha} (\rho\pa_\rho,\rho R\pa_R,\rho\pa_\omega)^{\leq 1}\phi|^2 \frac{\dd\rho}{\rho}\frac{\dd R}{R}\dd\slg.
\]
(The notation reflects the 0-nature---i.e.\ the vanishing---of the derivatives~\eqref{EqIEx0b} at $\rho=0$, and the b-nature---i.e.\ the tangency to $R=0$---at $R=0$. Note also that the integral over $\rho=\rho_0$ without the $\rho\pa_\rho$-derivative matches~\eqref{EqIHb}.) For fixed $\alpha$, given by the decay rate (or growth) of the initial data of $\phi$, there then exists $N$ so that
\[
  \|\phi\|_{\rho^{-N}R^\alpha H_{0,\bop}^1} \lesssim {\rm data}.
\]
The initial data norm here is the $R^\alpha\Hb^1\oplus R^\alpha\Hb^0$-norm. For the proof of this estimate, one can use the (future timelike) vector multiplier $-R^{-2\alpha}\rho^{2 N}\rho\pa_\rho$ for which the bulk term (deformation tensor) has a good sign when $N$ is large enough. (See Step~1 in the proof of Proposition~\ref{PropE2Reg}.)

\medskip

\textit{Step~1.2. Higher regularity.} This step is completely analogous to before: one now considers the system of commuted equations satisfied by $\rho\pa_\rho\phi$, $R\pa_R\phi$, $\pa_\omega\phi$ (which away from $R=0$ are equivalent to $\tau\pa_\tau\phi$, $\pa_x\phi$ as known from the de~Sitter discussion, and away from $\rho=0$ to $\pa_{\tilde r}\phi$, $\pa_{\tilde t}\phi$, $\pa_\omega\phi$). Due entirely to structural properties of the vector fields~\eqref{EqIEx0b} in relation to the vector fields $\rho\pa_\rho$, $R\pa_R$, $\pa_\omega$, this system has, at worst, a lower triangular structure analogous to~\eqref{EqIExTriangle}. This gives
\[
  \|\phi\|_{\rho^{-N}R^\alpha H_{0,\bop;\bop}^{1;k}}^2 := \sum_{i+j+|\gamma|\leq k} \|(\rho\pa_\rho)^i(R\pa_R)^j\pa_\omega^\gamma\phi\|_{\rho^{-N}R^\alpha H_{0,\bop}^1}^2 \lesssim {\rm data},
\]
where $N$ is \emph{fixed} and $k$ is only limited by the regularity of the initial data. (See Step~2 in the proof of Proposition~\ref{PropE2Reg}.)

\medskip

\textit{Step~2. Decay.} We can now regard all derivatives in the expression for $\Box_{g_\bhm}$ except for those \emph{only} involving $\rho\pa_\rho$ as error terms. That is, we rewrite the equation for $\phi$ as
\[
  I(\rho\pa_\rho,R,\omega)\phi = {\rm error} \in \rho^{-N+1}R^\alpha H_{0,\bop;\bop}^{1;k-2},\qquad
  I(\rho\pa_\rho,R,\omega) := \sum_{i=0}^2 \ell_{i 0 0}(0,R,\omega) (\rho\pa_\rho)^i.
\]
This is a family of ODEs in $\rho$ with parametric dependence on $R,\omega$, and can be integrated from initial data (or indeed from $\rho=\rho_1$ for any $\rho_1\in(0,\rho_0)$). Since $g_\bhm$ and $g_\dS$ agree to leading order at the conformal boundary $\rho=0$, the indicial operator is in fact \emph{the same} as for the wave equation on de~Sitter space (and in the toy model under consideration here independent of $R,\omega$). Therefore, the asymptotic behavior of $\phi=\phi(\rho,R,\omega)$ at $\rho=0$ is fully determined by the indicial roots (here $0,3$). In the present case, if $-N<0<-N+1$, we encounter the indicial root $0$ and thus obtain
\[
  \phi(\rho,R,\omega) \equiv \phi_0(R,\omega) \bmod \rho^{-N+1}R^\alpha H_{0,\bop;\bop}^{1;k-2},\qquad
  \phi_0 \in R^\alpha\Hb^{k-2}.
\]
(See Proposition~\ref{PropE2Decay} for details.) The only differences to the de~Sitter setting are thus:
\begin{enumerate}
\item the expansion is in terms of powers of $\rho$, not $\tau=\rho R$;
\item the terms in the expansion do not lie in standard Sobolev spaces in $x$, but in $R^\alpha\Hb^k$ (with the same $\alpha$ for all terms in the expansion).
\end{enumerate}

\medskip

We stress that the asymptotics of $\phi$ at $\cK=\{R=0\}$ and $\cI^+=\{\rho=0\}$ are completely decoupled: the decay rate $\alpha$ (or growth rate, if negative) of the initial data at $\cK$ propagates along $\cK$, but it has no bearing on the powers of $\rho$ appearing in the asymptotic expansion at the conformal boundary.

\subsection{The Einstein equations, gauges, constraint damping}
\label{SsIE}

Analogously to~\S\ref{SsIEx}, we first consider the de~Sitter setting (in $(3+1)$ dimensions) before explaining the simple modifications (given the framework explained in~\S\ref{SssIExSdS}) required for the Schwarzschild--de~Sitter case.

\subsubsection{de~Sitter space}
\label{SssIEdS}

The Einstein vacuum equations~\eqref{EqIEin} being nonlinear, we first consider their linearization
\[
  L_{g_\dS}:=D_{g_\dS}\Ric-\Lambda
\]
around $g_\dS$; this is $\frac12$ times $\Box_{g_\dS}+2\sR_{g_\dS}-\delta_{g_\dS}^*\delta_{g_\dS}\sfG_{g_\dS}-2\Lambda$ where $(\Box_g u)_{\mu\nu}=-g^{\kappa\lambda}u_{\mu\nu;\kappa\lambda}$ is the tensor wave operator, $\sR_{g_\dS}$ is a curvature operator, and we write $(\delta_g^*\omega)_{\mu\nu}=\frac12(\omega_{\mu;\nu}+\omega_{\nu;\mu})$, $(\delta_g h)_\mu=-g^{\kappa\lambda}h_{\mu\kappa;\lambda}$, and $\sfG_g h=h-\frac12 g\tr_g h$.

Expressing metric perturbations in (the symmetric second tensor power of) the frame $\frac{\dd\tau}{\tau}$, $\frac{\dd x^i}{\tau}$ ($i=1,2,3$), one can then write $L_{g_\dS}$ as a matrix of 0-differential operators (i.e.\ built from $\tau\pa_\tau$, $\tau\pa_x$). (It is not of wave type and not principally scalar, due to gauge issues addressed below.) The indicial family $I(L_{g_\dS},\lambda)$ is correspondingly a matrix-valued second order polynomial in $\lambda$ (see~\eqref{EqStCIIndRic} for the explicit expression). It is not invertible for \emph{any} $\lambda$, corresponding to the infinitesimal diffeomorphism invariance of the linearized Einstein equations, i.e.\ $L_{g_\dS}\circ\delta_{g_\dS}^*=0$ (so $\ran I(\delta_{g_\dS}^*,\lambda)\subset\ker I(L_{g_\dS},\lambda)$ for all $\lambda$). Since one wishes to disregard infinitesimal diffeomorphisms (Lie derivatives) as unphysical and expects to be able to eliminate them by suitable gauge choices, the more pertinent question is then to characterize the quotient space
\begin{equation}
\label{EqIEdSQuotient}
  \ker I(L_{g_\dS},\lambda) / \ran I(\delta_{g_\dS}^*,\lambda).
\end{equation}
This is a simple problem in linear algebra and solved in Lemma~\ref{LemmaStCISeq}. The upshot is that this space (with a mild modification required for the special value $\lambda=-1$) is trivial unless $\lambda=0,3$.\footnote{We argue that this should be regarded as the correct statement of mode stability in the de~Sitter context!} (This is already highly suggestive of the fact that the asymptotic degrees of freedom of perturbations of de~Sitter space are the coefficients $g_{(0)}$, $g_{(3)}$ in expansions such as~\eqref{EqIdSFG}.) Furthermore, the quotient space for $\lambda=0,3$ is spanned by tangential-tangential tensors $h_{i j}\frac{\dd x^i}{\tau}\frac{\dd x^j}{\tau}$ which are trace-free ($\sum_{i=1}^3 h_{i i}=0$). For now, our aim is to understand the nonlinear stability, in particular asymptotics and decay (however mild) towards $\tau=0$, and thus we focus on the indicial root $\lambda=0$. Since the indicial family governs asymptotics at each point of the conformal boundary individually, one may thus reasonably expect that a perturbation of de~Sitter space asymptotes to a metric of the form
\begin{equation}
\label{EqIdSUpdate}
  g_\dS + h_0 = -\frac{\dd\tau^2}{\tau^2} + \tau^{-2}(\dd x^2 + h_{0,i j}(x)\dd x^i\,\dd x^j),\qquad h_0=h_{0,i j}(x)\frac{\dd x^i}{\tau}\frac{\dd x^j}{\tau}.
\end{equation}
To go beyond heuristics, we need to supplement the (linearized) Einstein equations with a gauge condition in order to turn them into a wave equation admitting a well-posed initial value problem.

\medskip
\noindent
{\bf Gauges, I: eliminating non-decaying pure gauge solutions.} We deal with the diffeomorphism invariance by working with a (generalized) harmonic gauge. To motivate our particular choice, consider first the simple wave map (or DeTurck \cite{DeTurckPrescribedRicci}) gauge
\begin{equation}
\label{EqIdSUps}
  \Ups_\mu(g;g_\dS) := g_{\mu\nu}g^{\kappa\lambda}(\Gamma(g)_{\kappa\lambda}^\nu-\Gamma(g_\dS)_{\kappa\lambda}^\nu)=0.
\end{equation}
(This is a well-defined 1-form since the difference of two connections is a tensor; see~\eqref{EqEGauge} for a manifestly covariant expression.) The standard procedure for solving $\Ric(g)-\Lambda g=0$ in the gauge $\Ups(g;g_\dS)=0$ is then to consider the gauge-fixed Einstein equation
\[
  P_0(g) = \Ric(g)-\Lambda g - \delta_g^*\Ups(g;g_\dS) = 0.
\]
Given initial data $\gamma,k$, one then constructs Cauchy data for $g$ inducing $\gamma,k$ at $\rho=\rho_0$ for which moreover $\Ups(g;g_\dS)=0$ at $\rho=\rho_0$; once one has solved $P_0(g)=0$, the constraint equations imply that also the transversal derivative of $\Ups(g;g_\dS)$ at $\rho=\rho_0$ vanishes, and since the second Bianchi identity $\delta_g\sfG_g\Ric(g)=0$ implies the decoupled equation $\delta_g\sfG_g\delta_g^*\Ups(g;g_\dS)=0$, we conclude that $\Ups(g;g_\dS)=0$ and thus $\Ric(g)-\Lambda g=0$.

Consider now the linearization $L_0:=D_{g_\dS}P_0$. (This equals $\frac12\Box_{g_\dS}+\sR_{g_\dS}-\Lambda$ and is thus a wave operator on symmetric 2-tensors.) Let us determine the residual gauge freedom by computing those indicial solutions which are `pure gauge': this amounts to computing the indicial roots of
\begin{equation}
\label{EqIdSGauge}
  I\bigl(D_1|_{g_\dS}\Ups(\cdot;g_\dS)\circ\delta_{g_\dS}^*,\lambda\bigr)=I(-\delta_{g_\dS}\sfG_{g_\dS}\circ\delta_{g_\dS}^*,\lambda).
\end{equation}
It turns out that $\lambda=\frac12(3-\sqrt{33})\in(-2,-1)$ is one of them, with indicial solution denoted $\omega$. Thus, solutions of $L_0 h=0$ typically feature $\tau^\lambda$ \emph{growth} (with spatial profile an $x$-dependent multiple of the symmetric 2-tensor $I(\delta_{g_\dS}^*,\lambda)\omega$).

It is advantageous (see also Remark~\ref{RmkIdSHarmGauge} below) to devise a better gauge condition, namely one for which all `pure gauge' indicial roots have $\Re\lambda>0$. This would ensure that, modulo decaying remainders, solutions of $L_0 h=0$ are free of non-decaying gauge artefacts (and thus should be of the form $h_{0,i j}(x)\frac{\dd x^i}{\tau}\frac{\dd x^j}{\tau}+o(1)$, following the previous discussion). We arrange this by working with the gauge condition
\[
  \Ups_E(g;g_\dS) := \Ups(g;g_\dS) + E_{g_\dS}(g-g_\dS) = 0,
\]
where $E_{g_\dS}$ is a suitably chosen bundle map mapping 1-forms to symmetric 2-tensors; the requirement is simply that `pure gauge mode stability holds', i.e.\ all indicial roots of~\eqref{EqIdSGauge} with $\Ups_E$ in place of $\Ups$ are all positive. A possible choice for $E_{g_\dS}$ is given in~\eqref{EqEGauge}.\footnote{There is a small caveat, namely $-1$ is an indicial root, regardless of the choice of gauge modification, since $(\pa_{x^i})^\flat=\tau^{-2}\dd x^i=\tau^{-1}\frac{\dd x^i}{\tau}$ (corresponding to spatial translations) is a Killing 1-form on de~Sitter space. Since the symmetric gradient of this vanishes, it does not contribute non-decaying terms to solutions of the linearized gauge-fixed Einstein equations. See~\S\ref{SssE1Aux}.}

\begin{rmk}[Harmonic gauge]
\label{RmkIdSHarmGauge}
  In \cite[Appendix~C]{HintzVasyKdSStability}, the nonlinear stability of the \emph{static patch} of de~Sitter space, or more precisely of the slightly larger region $\frac{|x|}{\tau}<1+\delta$, is proved using an unmodified harmonic gauge. The indicial root $\frac12(3-\sqrt{33})$ arises there as a \emph{resonance}. The nonlinear iteration scheme of \cite{HintzVasyKdSStability} is capable of dealing with growing modes of the linearized equation by means of a black-box mechanism which, from a growing pure gauge mode $\delta^*\omega$, computes a 1-form modification $\theta=\theta(\omega)$ of the gauge condition so that in the gauge condition $\Ups-\theta=0$, the mode $\delta^*\omega$ does not arise \emph{at that particular iteration step}. (The modification $\theta$ in which global stability ultimately holds is thus part of the unknown.) --- In the present setting, where we are interested in the stability of a region of de~Sitter space which contains a nonempty open subset of the conformal boundary, the required gauge modifications would need to lie in an infinite-dimensional space in order to eliminate the $x$-dependent growing mode contribution $\delta^*\omega$ at all points $(0,x)$ on the conformal boundary at once. It is, however, not clear at present how to implement the black-box mechanism in this infinite-dimensional setting in a sufficiently robust manner so that it applies in a nonlinear iteration scheme.
\end{rmk}

\begin{rmk}[Ringstr\"om's gauge, I]
\label{RmkIdSRingstrom}
  In \cite[(46)--(50)]{RingstromEinsteinScalarStability}, Ringstr\"om introduces a gauge condition which is expressed \emph{in the particular global coordinate system $(-\log\tau,x)$}, namely $g^{\alpha\beta}\Gamma(g)_{\mu\alpha\beta}-3 g_{0\mu}=0$. (Since $g_\dS^{\alpha\beta}\Gamma(g_\dS)_{\mu\alpha\beta}=-3\delta_{0\mu}=3(g_\dS)_{0\mu}$, this gauge condition, at least for metrics with $g_{0\mu}=-\delta_{0\mu}+o(1)$, is equivalent to $g^{\alpha\beta}\Gamma(g)_{\mu\alpha\beta}-g_\dS^{\alpha\beta}\Gamma(g_\dS)_{\mu\alpha\beta}=0$.) While somewhat similar to~\eqref{EqIdSUps}, it has the conceptual disadvantage of not being covariant. Nonetheless, pure gauge mode stability does hold for this gauge (see Remark~\ref{RmkE1GaugeModOrigin}), albeit just barely since $0$ \emph{is} an indicial root.
\end{rmk}

\begin{rmk}[Weak global stability via patching static patches]
  In the context of Remark~\ref{RmkIdSRingstrom}, we remark that one could modify the nonlinear stability proof of \cite[Appendix~C]{HintzVasyKdSStability} to take place in Ringstr\"om's gauge (and with constraint damping, discussed below); one could then dispense of all gauge modifications (i.e.\ work in the \emph{fixed} gauge), and allow the final metric to deviate from the de~Sitter model. Applying such a result on each static patch (parameterized by the point $(0,x)$ on the conformal boundary) separately, one would thus obtain a global stability result for de~Sitter space since all perturbed static patches would automatically fit together. However, the regularity of the resulting solution would only be b-regularity in each static patch, meaning 0-regularity globally, which is far too weak to draw conclusions such as strong asymptotic expansions~\eqref{EqIMainExp}. See however \cite[\S{4.5}]{HintzVasySemilinear} for such an approach for the solution of nonlinear toy models.
\end{rmk}

\medskip
\noindent
{\bf Gauges, II: adjusting the background metric.} For the linearization of
\begin{equation}
\label{EqIdSP1}
  P_1(g) := \Ric(g) - \Lambda g - \delta_g^*\Ups_E(g;g_\dS)
\end{equation}
around $g=g_\dS$, the `physical' indicial root at $0$ (cf.\ the discussion following~\eqref{EqIEdSQuotient}) persists, and indeed one may expect solutions to asymptote to some tensor $h_0=h_{0,i j}(x)\frac{\dd x^i}{\tau}\frac{\dd x^j}{\tau}$. (It turns out that general solutions of the linearized equation still grow due to the existence of an indicial root which is not physical or pure gauge, but rather arises from an unfortunate cancellation of the Einstein and the gauge part of the operator; we deal with this using constraint damping below, and ignore this issue for the time being.) In a nonlinear iteration scheme, one might, in the next step, expect to have to consider the linearization of $P_1$ around a metric of the form $g=g_\dS+h_0+\tilde h$ (cf.\ \eqref{EqIdSUpdate}), with $\tilde h$ decaying towards $\tau=0$.

It turns out, however, that the gauge condition $\Ups_E(g_\dS+h_0+\tilde h;g_\dS)=0$ (and also $\Ups(g_\dS+h_0+\tilde h;g_\dS)$) \emph{cannot hold}, even to leading order at $\tau=0$ (i.e.\ ignoring $\tilde h$), for general $h_0$.\footnote{More precisely, in the second step of the iteration, $h_{(0)}=h_{0,i j}\dd x^i\,\dd x^j$ is trace-free with respect to $\dd x^2$, as discussed before~\eqref{EqIdSUpdate}; this saves the gauge condition $\Ups=0$ at $\tau=0$. But the next step would involve the linearization around $g_\dS+h_0+h'_0$ where $h'_0$ is trace-free with respect to $\dd x^2+h_{(0)}$; but $h_{(0)}+h'_{(0)}$ is typically no longer trace-free with respect to $\dd x^2$. This is why we need to study the gauge condition for $g_\dS+h_0$ assuming only that $h_0$ is tangential-tangential.} The idea is thus to replace the background metric $g_\dS$, which no longer captures the correct final geometry, by the new final geometry $g_\dS+h_0$. We implement this by regarding the leading order term $h_0$ and the decaying tail $\tilde h$ as separate unknowns, thus considering
\[
  P_2(h_0,\tilde h) := \Ric(g) - \Lambda g - \delta_g^*\Ups_E(g;g_0),\ \ \text{where}\ \ g=g_\dS+h_0+\tilde h,\ \ g_0=g_\dS+h_0.
\]
For any fixed $h_0$, this is a quasilinear wave equation for $\tilde h$. The change of the final background metric from $g_\dS$ to $g_\dS+h_0$ couples to the decaying remainder $\tilde h$ of the spacetime metric when evaluating the gauge 1-form $\Ups_E$; this necessitates the introduction of a further, decaying, gauge modification $\theta$. (Concretely, $\theta$ will lie in $\tau^\beta\Hb^k([0,1)\times\TT^3_x)$ for some $\beta\in(0,1)$.) See Lemma~\ref{LemmaPfGaugeMod}.

\begin{rmk}[Ringstr\"om's gauge, II]
\label{RmkIdSRingstrom2}
  An advantage of Ringstr\"om's gauge \cite[(46)--(50)]{RingstromEinsteinScalarStability} is that it is satisfied to leading order at $\tau=0$ for \emph{all} metrics of the form $g_\dS+h_0$. Thus, no adjustments of the gauge condition at $\tau=0$ are needed in this case. The fact that in our more geometric gauge we do adjust the background metric ultimately leads to significantly simpler computations of the indicial families (which end up being independent of the background metric in suitable bundle splittings, see~\eqref{EqEIndRootPf}), at a very minor technical expense (essentially Lemma~\ref{LemmaPfGaugeMod}).
\end{rmk}

\medskip
\noindent
{\bf Constraint damping.} For now, we return to the linearization $L_1$ of the operator~\eqref{EqIdSP1} around $g=g_\dS$. In the above discussion, we have in effect assumed that solutions of $L_1 h=0$ are sums of physical solutions (as in~\eqref{EqIdSUpdate}), pure gauge solutions, and decaying remainders. This is, however, not true: there is a negative indicial root, again at $\lambda=\frac12(3-\sqrt{33})\in(-2,-1)$, for which the corresponding indicial solution neither lies in $\ker I(D_{g_\dS}\Ric-\Lambda,\lambda)$ nor satisfies the linearized gauge condition. The corresponding growing ($\cO(\tau^\lambda)$) solution would arise for general initial data, which one does need to consider in a Nash--Moser iteration scheme for the solution of the nonlinear equation (or when solving the gauge-fixed Einstein equations numerically, as already pointed out in \cite{RingstromEinsteinScalarStability}).

The fix, going back to \cite{BrodbeckFrittelliHubnerReulaSCP,GundlachCalabreseHinderMartinConstraintDamping,PretoriusBinaryBlackHole}, is to modify the symmetric gradient $\delta_g^*$ coupling the gauge condition and the Ricci tensor. This was also used in an ad hoc fashion in \cite[(51)--(54)]{RingstromEinsteinScalarStability}, and played a crucial role in the nonlinear stability proof \cite{HintzVasyKdSStability}. (In a bootstrap approach, it can be avoided \cite{FangKdS}, but since it is easy to arrange, we might as well arrange it.) To wit, we replace $\delta_g^*$ in~\eqref{EqIdSP1} by $\tilde\delta_g^*:=\delta_g^*+\tilde E$ for a suitably chosen bundle map from symmetric 2-tensors to 1-forms. The only requirement is that
\begin{equation}
\label{EqIdSCD}
  \parbox{0.8\textwidth}{\centering all indicial roots of $I(\delta_{g_\dS}\sfG_{g_\dS}\circ\tilde\delta_{g_\dS}^*,\lambda)$ are positive.}
\end{equation}
Possible choices of $\tilde E$ are given in~\eqref{EqECD} (corresponding to \cite[(51)--(54)]{RingstromEinsteinScalarStability}) or \cite[(C.8)]{HintzVasyKdSStability} (see Remark~\ref{RmkE1ChoicetildeE}). We then consider the linearization $\tilde L_1$ of
\[
  \tilde P_1(g) := \Ric(g) - \Lambda g - \tilde\delta_g^*\Ups_E(g;g_\dS).
\]
The utility of~\eqref{EqIdSCD} is the following: if $\lambda<0$ and $I(\tilde L_1,\lambda)h=0$, then the linearized second Bianchi identity implies that also $I(\delta_{g_\dS}\sfG_{g_\dS}\circ\tilde\delta_{g_\dS}^*,\lambda)\eta=0$ where $\eta=I(D_1|_{g_\dS}\Ups_E(\cdot;g_\dS),\lambda)h$ measures the extent to which $h$ violates the linearized gauge condition. But then $\eta$ must vanish in view of~\eqref{EqIdSCD}, and thus $h$ is in fact automatically an indicial solution of the linearized Einstein vacuum equations, i.e.\ $I(D_{g_\dS}\Ric-\Lambda,\lambda)h=0$. Since $\lambda<0$, it is pure gauge, and since the gauge condition disallows growing pure gauge solution, it must vanish.

For the gauge-fixed Einstein equation with gauge modification and constraint damping,
\begin{equation}
\label{EqIEdSP}
\begin{split}
  &P(h_0,\tilde h,\theta) := \Ric(g) - \Lambda g - \tilde\delta_g^*(\Ups_E(g;g_0)-\theta), \\
  &\hspace{6em} \text{where}\ \ g=g_\dS+h_0+\tilde h,\ \ g_0=g_\dS+h_0,
\end{split}
\end{equation}
we thus expect that all indicial roots are $\geq 0$; and this is indeed the case. See Lemma~\ref{LemmaEIndRoot}.

Using the analytic techniques explained in~\S\ref{SsIEx}, one can then prove the existence of a global solution $h_0,\tilde h,\theta$ of $P(h_0,\tilde h,\theta)=0$ with given (gauged) initial data close to those of de~Sitter space. (See Theorem~\ref{ThmESol} for the black hole case.) The strategy is to obtain precise asymptotics for solutions of the linearization of $P$ in $\tilde h$ (which is a tensorial wave equation), read off updates for $h_0$ (related to the final spatial metric), $\tilde h$ (the decaying remainder of the metric perturbation), and $\theta$ (the decaying gauge modification), and close the iteration using a Nash--Moser scheme.

\subsubsection{Expanding regions of de~Sitter black hole spacetimes}
\label{SssIESdS}

The considerations in~\S\ref{SssIEdS} are entirely on the level of indicial roots (except for the adjustment of the background metric). Using the analytic modifications needed to pass from de~Sitter to KdS already discussed in~\S\ref{SssIExSdS}, it is thus clear that we can prove the nonlinear stability of the expanding region of KdS \emph{using the same nonlinear operator}~\eqref{EqIEdSP}, except we need to replace $g_\dS$ by the KdS metric $g_\sfb$. See Theorem~\ref{ThmESol} for the resulting nonlinear small data global existence result.

\begin{rmk}[Ringstr\"om's gauge, III]
\label{RmkIESdSRingstrom}
  It is conceivable that one can prove the nonlinear stability of the expanding region entirely in Ringstr\"om's gauge, though for reasons of conceptual clarity (as explained in Remarks~\ref{RmkIdSRingstrom} and \ref{RmkIdSRingstrom2}) we do not pursue this here.
\end{rmk}

\subsubsection{Fefferman--Graham type expansion at the conformal boundary}
\label{SssIEFG}

Having solved the Einstein equations~\eqref{EqIEin} in a gauge $\Ups_E(g;g_0)-\theta=0$, the spacetime metric $g$ is typically not conformally smooth. At this point, we disregard the evolution character of the initial value problem, and instead aim to construct diffeomorphisms $\phi=\Id+\cO(\rho^\beta)$ ($\beta\in(0,1)$) so that $\phi^*g$ becomes as regular as possible at the conformal boundary $\rho=0$. We do this in two steps.

\medskip
\textit{Step~1. Simplify the gauge condition.} The idea is that, for $\phi=e^V$ with $V=\cO(\rho^\beta)$ a small 0-vector field, $\Ups_E(\phi^*g;g_0)$ differs from $\Ups_E(g;g_0)$ by a term which is roughly equal to $D_1|_g\Ups_E(\cL_V g;g_0)$; this is a wave operator acting on $V$. By inverting the indicial operator of this wave operator, one can then find successively better choices of $V$ so that $\Ups_E(\phi^*g;g_0)$ has successively higher orders of vanishing at $\rho=0$. A Borel lemma type argument produces a diffeomorphism $\phi$ with $\Ups_E(\phi^*g;g_0)\equiv 0$ (i.e.\ infinite order vanishing at $\rho=0$); see Proposition~\ref{PropStGauge}. Once this is done, an indicial root based analysis of the gauge-fixed Einstein equations shows that $\phi^*g$ is log-smooth down to $\rho=0$ (Lemma~\ref{LemmaStLog}).

\medskip
\textit{Step~2. Obtaining Fefferman--Graham asymptotics.} This part of the argument applies to any solution of $\Ric(g)-\Lambda g=0$ which asymptotes to an asymptotically de~Sitter metric at the conformal boundary and is log-smooth there. To wit, we consider each term in the generalized Taylor expansion of $g$ separately (starting with the $\rho(\log\rho)^m$ terms); using the Einstein equation and simple linear algebra based on the description of~\eqref{EqIEdSQuotient}, one can easily eliminate all $\rho^1$ and all $\rho^2(\log\rho)^m$, $m\geq 1$, terms. The $\rho^3$ and $\rho^4$ levels are more delicate due to the fact that the solvability theory for $I(D\Ric-\Lambda,\lambda)h=f$ for $\lambda=3,4$ is somewhat delicate and requires $f$ to have a special structure which must be verified. The elimination of log terms at subsequent levels ($\rho^5$, $\rho^6$, etc.) is again straightforward using the triviality of the quotient space~\eqref{EqIEdSQuotient}. The details are given in (the proof of) Proposition~\ref{PropStCI}.

\medskip
The full nonlinear stability result for the Einstein equations is then Theorem~\ref{ThmEStabEin}.

\subsection{Outline of the paper}
\label{SsIO}

In~\S\ref{S0b}, we discuss 0-, b-, and (0,b)-operators (as motivated in~\S\ref{SsIEx}) in more detail, and how to describe de~Sitter and Kerr--de~Sitter metrics using related notions. Furthermore, we define the corresponding weighted Sobolev spaces and prove some of their properties as required for linear and nonlinear analysis.

In~\S\ref{SE}, we analyze the gauge-fixed Einstein vacuum equations (in the form motivated in~\S\ref{SssIEdS}) in detail. This includes the study of the indicial roots of their linearizations, and the proofs of (higher order) energy estimates and sharp asymptotics, with tame estimates, for solutions of their linearizations (following the outline given in~\S\S\ref{SssIExdS}--\ref{SssIExSdS}). This section concludes with a proof of Theorem~\ref{ThmESol}, i.e.\ small data global existence for the gauge-fixed Einstein equations.

In~\S\ref{SSt}, we demonstrate how to improve the asymptotic behavior of the spacetime metrics $g$ produced by Theorem~\ref{ThmESol} when they arise from initial data satisfying the constraint equations (and thus $g$ satisfies the Einstein vacuum equations). The main result are the precise asymptotics stated in Theorem~\ref{ThmEStabEin}.

\subsection*{Acknowledgments}

The authors would like to thank Grigorios Fournodavlos and Volker Schlue for many fruitful conversations on the topic of this paper. A.V.~gratefully acknowledges support from the National Science Foundation under grant number DMS-2247004 as well as a Simons Visiting Professorship at the Mathematisches Forschungsinstitut Oberwolfach.

\section{Kerr--de~Sitter space and 0-b-structures}
\label{S0b}

\subsection{Kerr--de~Sitter metrics as asymptotically de~Sitter metrics}
\label{Ss0bMet}

Recall from~\eqref{EqIdS} that the half space model of $(1+3)$-dimensional de~Sitter space is
\begin{equation}
\label{Eq0bdS}
  M^\circ := (0,\infty)_\tau \times \R^3_x,\qquad g_\dS := \frac{-\dd\tau^2+\dd x^2}{\tau^2}.
\end{equation}
(This is more commonly expressed in terms of $\tau=e^{-t_*}$ as $-\dd t_*^2+e^{-2 t_*}\dd x^2$.) The Kerr--de~Sitter metric $g_\sfb$ \cite{CarterHamiltonJacobiEinstein} with parameters $\sfb=(\bhm,a)$, $\bhm,a\in\R$, and $\Lambda=3$ is
\begin{equation}
\label{Eq0bKdS}
\begin{split}
  &g_\sfb = -\frac{\mu(r)}{\varrho^2}\Bigl(\dd t-\frac{a\sin^2\theta}{\Delta_0}\dd\phi\Bigr)^2 + \varrho^2\Bigl(\frac{\dd r^2}{\mu(r)}+\frac{\dd\theta^2}{\Delta_\theta}\Bigr) + \frac{\Delta_\theta\sin^2\theta}{\varrho^2}\Bigl(\frac{r^2+a^2}{\Delta_0}\dd\phi-a\,\dd t\Bigr)^2, \\
  &\quad \mu:=(r^2+a^2)(1-r^2)-2\bhm r,\ \Delta_0:=1+a^2,\ \Delta_\theta:=1+a^2\cos^2\theta,\ \varrho^2:=r^2+a^2\cos^2\theta.
\end{split}
\end{equation}
(This differs from the expression in Boyer--Lindquist coordinates by a constant rescaling of $t$ by $\Delta_0$; cf.\ \cite[(5.2)--(5.4)]{SchlueCosmological}.) Both metrics satisfy~\eqref{EqIEin} with $\Lambda=3$. To explain the sense in which $g_\sfb$ can be regarded as a black hole in de~Sitter space, tending to a point $\fp$ on the conformal boundary of de~Sitter space~\eqref{Eq0bdS}, say $(\tau,x)=(0,0)$ at $\fp$, we first rewrite the dS metric in two steps. The first step, following \cite[\S{2.1}]{HintzGluedS}, is to introduce polar coordinates
\begin{equation}
\label{Eq0bdSPolar}
  R := |x| \in [0,\infty),\qquad
  \omega := \frac{x}{|x|} \in \Sph^2,
\end{equation}
so $\dd x^2=\dd R^2+R^2\slg$ where $\slg$ is the standard metric on $\Sph^2$, and then defining
\begin{equation}
\label{Eq0bTildeCoord}
  \tilde t:=-\frac12\log(R^2 - \tau^2),\qquad \tilde r:=\frac{R}{\tau}.
\end{equation}
(This change of variables is valid \emph{outside} the cosmological region, not in the whole region where the metric~\eqref{Eq0bdS} is defined.) Write $\tilde\theta,\tilde\phi$ for polar coordinates on $\Sph^2$. In the coordinates $(\tilde t,\tilde r,\tilde\theta,\tilde\phi)$, we then have
\begin{equation}
\label{Eq0bdStilde}
  g_\dS = -(\tilde r^2-1)^{-1}\dd\tilde r^2 + (\tilde r^2-1)\,\dd\tilde t^2 + \tilde r^2(\dd\tilde\theta^2+\sin^2\tilde\theta\,\dd\tilde\phi^2).
\end{equation}
The second step, following \cite[Appendix~B]{SchlueCosmological}, is to define $t,r,\theta,\phi$ in $\tilde r>1$ via\footnote{The equation for $\tilde r^2$ can be solved for $r$ using $\sin^2\theta=1-\frac{\tilde r^2\cos^2\tilde\theta}{r^2}$; explicitly,
\[
  r^2=\frac{\tilde r^2(1+a^2\sin^2\tilde\theta)-a^2}{2}+\sqrt{\left(\frac{\tilde r^2(1+a^2\sin^2\tilde\theta)-a^2}{2}\right)^2+a^2\tilde r^2\cos^2\tilde\theta}.
\]
Moreover, $(1+a^2\cos^2\theta)r^2=(1+a^2)\tilde r^2-a^2\sin^2\theta\geq(1+a^2)\tilde r^2-a^2\sin^2\theta\,\tilde r^2=(1+a^2\cos^2\theta)\tilde r^2$ implies that $r\geq\tilde r$, and thus $\theta\in(0,\pi)$ is well-defined if we require $\theta-\frac{\pi}{2}$ have the same sign as $\tilde\theta-\frac{\pi}{2}$.}
\begin{equation}
\label{Eq0bCorotating}
  \tilde t = t,\qquad
  \tilde r^2 = \frac{1}{\Delta_0}(r^2\Delta_\theta + a^2\sin^2\theta),\qquad
  \tilde r\cos\tilde\theta = r\cos\theta,\qquad
  \tilde\phi = \phi-a t.
\end{equation}
In the `co-rotating coordinates' $t,r,\theta,\phi$, one then finds\footnote{The coefficient of $\dd\theta^2$ is erroneously given as $\frac{\varrho^2}{\Delta_\theta^2}$ in \cite[(B.6)]{SchlueCosmological}, \cite[\S{4.1}]{HintzGluedS}.}
\begin{align*}
  g_\dS &= (r^2+a^2\sin^2\theta-1)\dd t^2 + \frac{\varrho^2}{(r^2+a^2)(1-r^2)}\dd r^2 + \frac{\varrho^2}{\Delta_\theta}\dd\theta^2 + \frac{r^2+a^2}{\Delta_0}\sin^2\theta\,\dd\phi^2 \\
    &\qquad - 2 a\frac{r^2+a^2}{\Delta_0}\sin^2\theta\,\dd t\,\dd\phi.
\end{align*}
Comparing this with~\eqref{Eq0bKdS}, we find that
\begin{equation}
\label{Eq0bKdSDiffdS}
  g_\sfb = g_\dS + h_\sfb,\qquad
  h_\sfb = \frac{2\bhm r}{\varrho^2}\Bigl(\dd t-\frac{a\sin^2\theta}{\Delta_0}\dd\phi\Bigr)^2 + \frac{2\bhm r\varrho^2}{\mu|_{\bhm=0}\mu}\dd r^2.
\end{equation}
Note that the coefficients of $h_\sfb$ are of size $r^{-1}$, and the coefficient of $\dd r^2$ is of size $r^{-5}$.

We interpret this on a structural level as follows. Define first the manifold
\begin{equation}
\label{Eq0bM}
  M := [0,\infty)_\tau \times X,\qquad X=\R^3_x,
\end{equation}
whose boundary $\pa M=\tau^{-1}(0)$ is the future conformal boundary of dS. We recall from \cite{MazzeoMelroseHyp} that the 0-cotangent bundle is the smooth vector bundle ${}^0 T^*M\to M$ with frame $\frac{\dd\tau}{\tau}$, $\frac{\dd x^i}{\tau}$ ($i=1,2,3$). In the present paper, we always work in the splitting
\begin{equation}
\label{Eq0bT0Split}
  {}^0 T^*M = \R\frac{\dd\tau}{\tau} \oplus \tau^{-1}T^*X.
\end{equation}
Correspondingly, we split the second symmetric tensor power of this bundle via
\begin{equation}
\label{Eq0bST0Split}
  S^2\,{}^0 T^*M = \R\frac{\dd\tau^2}{\tau^2} \oplus \Bigl(2\frac{\dd\tau}{\tau}\otimes_s \tau^{-1}T^*X\Bigr) \oplus \tau^{-2}S^2 T^*X.
\end{equation}
In this splitting, the dS metric is thus given by $(-1,0,\dd x^2)$.

Next, blow up the point $\fp\in M$ given by $(\tau,x)=(0,0)$. This produces a manifold with corners \cite{MelroseDiffOnMwc}, in which we will only work in the region where $|x|>\tau$. Concretely, we introduce (consistently with~\eqref{Eq0bdSPolar}) the coordinates
\[
  R := |x| \in [0,\infty),\qquad
  \rho := \frac{\tau}{|x|} \in [0,1),\qquad
  \omega := \frac{x}{|x|} \in \Sph^2.
\]
In terms of $R,\rho$, the expressions in~\eqref{Eq0bTildeCoord} become
\begin{equation}
\label{Eq0bKdStilde}
  \tilde t=-\log R-\frac12\log(1-\rho^2),\qquad \tilde r=\rho^{-1}.
\end{equation}

\begin{definition}[Manifold]
\label{Def0bMfd}
  We define
  \begin{equation}
  \label{Eq0bMfd}
    \breve M := [0,\bar\rho]_\rho \times [0,\infty)_R \times \Sph^2_\omega,
  \end{equation}
  where $\bar\rho=\bar\rho(\sfb)\in(0,1)$ is chosen so that $g_\sfb$ is well-defined and $\dd(\rho^{-1})$ is timelike on $\breve M^\circ:=(0,\bar\rho]\times(0,\infty)\times\Sph^2$. We denote the blow-down map by
  \[
    \upbeta\colon\breve M\to M,\qquad
    (\rho,R,\omega)\mapsto(\tau,x)=(\rho R,\omega R),
  \]
  and the boundary hypersurfaces of $\breve M$ by
  \begin{equation}
  \label{Eq0bMfdBdy}
    \cI^+ := \rho^{-1}(0),\qquad
    \cK := R^{-1}(0).
  \end{equation}
  We moreover define $\cI^+_{R_0}:=\cI^+\cap\{R\leq R_0\}=[0,R_0]\times\Sph^2\subset\cI^+$.
\end{definition}

See Figure~\ref{Fig0bMfd}. The requirements on $\bar\rho>0$ in Definition~\ref{Def0bMfd} are satisfied for all sufficiently small $\bar\rho$ since $\frac{\dd(\rho^{-1})}{\rho^{-1}}=\frac{\dd\tilde r}{\tilde r}$ is timelike, and indeed has squared norm $1-\tilde r^{-2}\geq\frac34$ for $g_\dS$ in $\tilde r>2$ (by inspection of~\eqref{Eq0bdStilde}); in view of the decay of $h_\sfb$ in~\eqref{Eq0bKdSDiffdS} as $r\to\infty$ (equivalently, $\tilde r\to\infty$), this timelike character persists for sufficiently small $\tilde r^{-1}$.

\begin{figure}[!ht]
\centering
\includegraphics{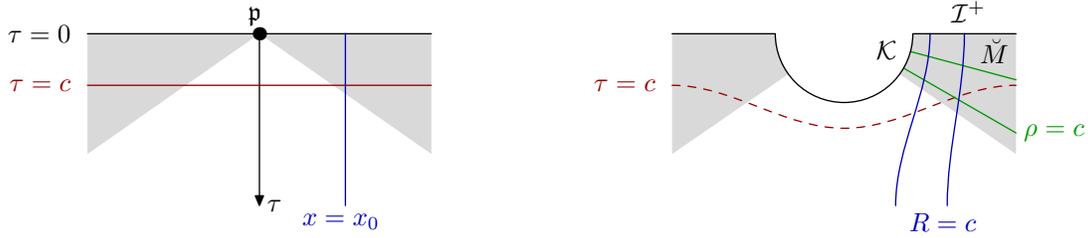}
\caption{\textit{On the left:} the de~Sitter spacetime manifold $M$ and its conformal boundary at $\tau=0$; some level sets of $\tau$ are shown in red, and level sets of $x$ in blue. The image $\upbeta(\breve M)\subset M$ is shaded in gray. \textit{On the right:} the blow-up $[M;\{\fp\}]$, the manifold with corners $\breve M$ (in gray), and level sets of $\tau$ (dashed red), $R=|x|$ (blue), $\rho$ (green).}
\label{Fig0bMfd}
\end{figure}

The benefit of working with $\breve M$ for the purpose of studying the KdS metric is thus that smooth functions of $\tilde r^{-1}$, say, are smooth on $\breve M$, whereas they are singular on $M$ near $\fp$; we return to this below. Since $\tau=\rho R$ and $x=R\omega$, we can write elements of ${}^0 T^*M$ as linear combinations (with smooth coefficients on $\breve M$) of $\frac{\dd\tau}{\tau}=\frac{\dd\rho}{\rho}-\frac{\dd R}{R}$, $\tau^{-1}\dd R$, and $\tau^{-1}R\,\dd\omega$, or equivalently in terms of $\frac{\dd\rho}{\rho}$, $\frac{\dd R}{\rho R}$, and $\frac{\dd\omega}{\rho}$. This motivates the following definition, which is studied further in~\S\ref{Ss0b} below.

\begin{definition}[0-b-cotangent bundle]
\label{Def0bT}
  We define the \emph{0-b-cotangent} bundle over the manifold $\breve M$ defined by~\eqref{Eq0bMfd} to be the direct sum
  \begin{equation}
  \label{Eq0bTSplit}
    {}^{0,\bop}T^*\breve M := \R\frac{\dd\rho}{\rho} \oplus \R\frac{\dd R}{\rho R} \oplus \rho^{-1}T^*\Sph^2.
  \end{equation}
  By this we mean that ${}^{0,\bop}T^*\breve M=\ul\R\oplus\ul\R\oplus T^*\Sph^2$ as a vector bundle where $\ul\R=\breve M\times\R$ is the trivial bundle; and an element $(a,b,\eta)$ with $a,b\in\R$ and $\eta\in T^*\Sph^2$ is identified with the covector $a\frac{\dd\rho}{\rho}+b\frac{\dd R}{\rho R}+\rho^{-1}\eta\in T^*M^\circ$.
\end{definition}

The above discussion shows that the identity map on $T^*M^\circ$ over $\breve M^\circ$ extends to a smooth bundle isomorphism\footnote{As an illustration, $g_\dS$ is a smooth non-degenerate Lorentzian signature section of $S^2\,{}^{0,\bop}T^*\breve M$; explicitly,
\[
  g_\dS = -\Bigl(\frac{\dd\rho}{\rho}\Bigr)^2 - 2\rho\frac{\dd\rho}{\rho}\frac{\dd R}{\rho R} + (1-\rho^2)\Bigl(\frac{\dd R}{\rho R}\Bigr)^2 + \frac{\slg}{\rho^2}.
\]}
\begin{equation}
\label{Eq0bTIso}
  {}^{0,\bop}T^*\breve M\cong\upbeta^*({}^0 T^*M)\ \ \text{over}\ \ \breve M.
\end{equation}

From~\eqref{Eq0bKdStilde} and \eqref{Eq0bCorotating}, it follows that $e^{-t}=e^{-\tilde t}=R\sqrt{1-\rho^2}$ and $r^{-1}$ are smooth positive multiples of $R$ and $\rho$ on $\breve M$. Therefore, the coefficients of $h_\sfb$ in~\eqref{Eq0bKdSDiffdS} in the frame $r\,\dd t$, $\frac{\dd r}{r}$, $r\,\dd\theta$, $r\,\dd\phi$ are smooth multiples of $r^{-3}$ on $\breve M$; and this persists in the frame $\tilde r\,\dd\tilde t$, $\frac{\dd\tilde r}{\tilde r}$, $\tilde r\,\dd\tilde\theta$, $\tilde r\,\dd\tilde\phi$. Since $\tilde r\,\dd\tilde t=-\frac{\dd R}{\rho R}+\frac{\dd\rho}{1-\rho^2}$ and $\frac{\dd\tilde r}{\tilde r}=-\frac{\dd\rho}{\rho}$, we conclude that in the decomposition of $h_\sfb$ according to~\eqref{Eq0bTSplit} (or equivalently according to~\eqref{Eq0bST0Split}, i.e.\ into $\frac{\dd\tau^2}{\tau^2}$, $2\frac{\dd\tau}{\tau}\otimes_s\frac{\dd x^i}{\tau}$, $\frac{\dd x^i}{\tau}\otimes_s\frac{\dd x^j}{\tau}$), each component is $\rho^3$ times a smooth function on $\breve M$. Using~\eqref{Eq0bTIso}, we can summarize our discussion as follows.

\begin{lemma}[Structure of KdS and dS metrics]
\label{Lemma0bStruct}
  Define $g_\dS$ by~\eqref{Eq0bdS}, and define $g_\sfb$ as a metric on $\breve M^\circ$ by~\eqref{Eq0bKdS} via the coordinate transformations~\eqref{Eq0bdSPolar}, \eqref{Eq0bTildeCoord}, and \eqref{Eq0bCorotating}. Then\footnote{This sharpens \cite[\S{4.1}]{HintzGluedS} insofar as we now control $g_\sfb,g_\dS$ and their difference uniformly down to $\cK$.} $g_\sfb,g_\dS$ are smooth Lorentzian signature sections of $S^2\,{}^{0,\bop}T^*\breve M=\upbeta^*(S^2\,{}^0 T^*M)$ over $\breve M$, and
  \[
    g_\sfb - g_\dS \in \rho^3\CI(\breve M;S^2\,{}^{0,\bop}T^*\breve M) = \rho^3\CI\bigl(\breve M;\upbeta^*(S^2\,{}^0 T^*M)\bigr).
  \]
  Furthermore, $\frac{\dd\rho}{\rho}$ is uniformly timelike for $g_\sfb$ in the sense that
  \begin{equation}
  \label{Eq0bStructTime}
    0 > g_\sfb\Bigl(\frac{\dd\rho}{\rho},\frac{\dd\rho}{\rho}\Bigr) \in \CI(\breve M).
  \end{equation}
\end{lemma}

The perturbations of $g_\sfb$ arising in the solution of nonlinear stability problem will similarly be considered as sections of $\upbeta^*(S^2\,{}^0 T^*M)$ with suitable regularity and decay properties on $\breve M$. For concreteness, we fix the types of domains on which we will study perturbations of $g_\sfb$ as follows:

\begin{definition}[Domains with spacelike boundaries]
\label{Def0bDom}
  We use the notation of Definition~\usref{Def0bMfd}. Let $R_0\in(0,\infty)$. For all $\rho_0\in(0,\bar\rho]$ so that $\Sigma_{\rho_0,R_0}^+:=\{\rho\leq\rho_0,\ \rho=R^{-1}(\rho_0 R_0-\half(R_0-R))\}$ is spacelike for $g_\sfb$, we set
  \begin{equation}
  \label{Eq0bDom}
    \Omega_{\rho_0,R_0} := \Bigl\{ (\rho,R,\omega) \in \breve M \colon \rho\leq\rho_0,\ R\leq R_0,\ \rho\geq R^{-1}\bigl(\rho_0 R_0-\half(R_0-R)\bigr) \Bigr\}.
  \end{equation}
  We denote by $\Sigma_{\rho_0,R_0}:=\{(\rho,R,\omega)\colon \rho=\rho_0,\ R\leq R_0\}$ the initial boundary hypersurface of $\Omega_{\rho_0,R_0}$. The final boundary hypersurface is $\Sigma_{\rho_0,R_0}^+$.
\end{definition}

Note that the final inequality in the definition of $\Omega_{\rho_,R_0}$ is equivalent to $\tau=\rho R\geq \rho_0 R_0-\half(R_0-R)$; the hypersurface $\Sigma_{\rho_0,R_0}^+$ (which intersects $\tau=0$ at $R=(1-2\rho_0)R_0$) is a spacelike hypersurface for $g_\dS$, and therefore it is also one for $g_\sfb$ in the region where $\rho\leq\rho_0$, provided $\rho_0$ is sufficiently small. Moreover, the future timelike vector field $\pa_r$ is outward pointing, and thus it is a final boundary hypersurface for purposes of solving wave equations (i.e.\ no data need to be imposed there). We will pose initial data at $\Sigma_{\rho_0,R_0}$. See Figure~\ref{Fig0bDom}.

\begin{figure}[!ht]
\centering
\includegraphics{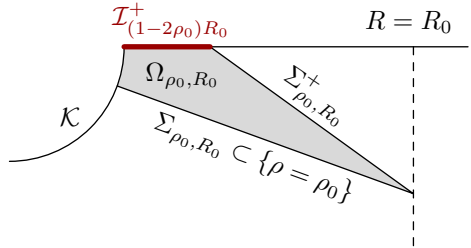}
\caption{Illustration of a domain $\Omega_{\rho_0,R_0}\subset\breve M$, only showing the coordinates $\rho,R$.}
\label{Fig0bDom}
\end{figure}

\subsection{Basics of 0-b-analysis; weighted b-Sobolev spaces}
\label{Ss0b}

We now introduce the operator classes and function spaces used in the analysis of the stability problem. For definiteness, we work on the 4-dimensional manifold $\breve M$ from Definition~\ref{Def0bMfd}, although all notions and results in this section admit straightforward generalizations to the case of general dimensions.

\subsubsection{Vector fields and operators}

The bundle dual to~\eqref{Eq0bTSplit} is the 0-b-tangent bundle
\[
  {}^{0,\bop}T\breve M=\R(\rho\pa_\rho)\oplus\R(\rho R\pa_R)\oplus\rho T\Sph^2.
\]
The space of its smooth sections is denoted $\cV_{0,\bop}(\breve M)$; it is spanned over $\CI(\breve M)$ by $\rho\pa_\rho$, $\rho R\pa_R$, and $\rho\Omega$ where $\Omega\in\cV(\Sph^2)$. Thus, a smooth vector field $V$ on $\breve M$ is a 0-b-vector field (i.e.\ lies in $\cV_{0,\bop}(\breve M)$) if and only if $V$ vanishes at $\cI^+=\rho^{-1}(0)$ (hence the subscript `$0$' \cite{MazzeoMelroseHyp}) and is tangent to $\cK=R^{-1}(0)$ (hence the subscript `$\bop$' \cite{MelroseMendozaB,MelroseAPS}).

\begin{definition}[0-b-operators]
\label{Def0bOp}
  For $m\in\N_0$, we write $\Diff_{0,\bop}^m(\breve M)$ for the space of all differential operators on $\breve M$ which are finite sums of up to $m$-fold compositions of elements of $\cV_{0,\bop}(\breve M)$. If $E\to\breve M$ is a vector bundle, we write $\Diff_{0,\bop}^m(\breve M;E)$ for operators $\CIc(\breve M^\circ;E)\to\CIc(\breve M^\circ;E)$ which in each local trivialization of $E$ are matrices of operators of class $\Diff_{0,\bop}^m$.
\end{definition}

In this paper, we only work with trivial(ized) bundles, and thus with matrices of scalar operators. A scalar operator can conversely be regarded as an operator on a trivialized bundle by acting component-wise. We will henceforth only discuss scalar operators explicitly, leaving the straightforward notational modifications to operators on trivial(ized) bundles (and the simple generalization to non-trivial bundles) to the reader.

By duality from~\eqref{Eq0bTIso}, we can equivalently define $\cV_{0,\bop}(\breve M)$ to be the $\CI(\breve M)$-span of the set $\{V|_{\breve M^\circ}\colon V\in\cV_0(M)\}$ of 0-vector fields on $M$; here $\cV_0(M)$ consists of all smooth vector fields on $M$ which vanish at $\pa M$, and we thus have $\cV_0(M)=\CI(M;{}^0 T M)$ where ${}^0 T M=\R(\tau\pa_\tau)\oplus\tau T X$ is the dual bundle to ${}^0 T^*M$ in~\eqref{Eq0bT0Split}.

A larger class of operators is
\[
  \Diffb^m(\breve M),
\]
which is defined analogously to $\Diff_{0,\bop}^m(\breve M)$ but using \emph{b-vector fields} $\Vb(\breve M)$, which are precisely those smooth vector fields on $\breve M$ which are tangent to all boundary hypersurfaces $\cI^+,\cK$ of $\breve M$. Thus, $\Vb(\breve M)$ is spanned over $\CI(\breve M)$ by $\rho\pa_\rho$, $R\pa_R$, and $\cV(\Sph^2)$.\footnote{The benefit of the definitions of $\cV_{0,\bop}$, $\cV_0$, $\Vb$ solely in terms of the smooth structure of $\breve M$ is that it allows one to determine frames for ${}^{0,\bop}T\breve M$ in local coordinates on $\breve M$ without the need for, say, change of variables computations.} One can also consider spaces of weighted operators
\[
  R^\alpha\rho^\beta\Diff_{0,\bop}^m(\breve M) = \{ R^\alpha\rho^\beta L \colon L\in\Diff_{0,\bop}^m(\breve M) \},
\]
similarly for $R^\alpha\rho^\beta\Diffb^m(\breve M)$. Elements of these spaces define bounded linear maps on $\CIc(\breve M^\circ)$.

Fixing a finite spanning set $\sV=\{V_a\}\subset\cV(\Sph^2)$ over $\CI(\Sph^2)$ (e.g.\ rotation vector fields around coordinate axes), we can express any $L\in\Diff_{0,\bop}^m(\breve M)$ in the form
\begin{equation}
\label{Eq0bL}
  L = \sum_{j+k+|\alpha|\leq m} \ell_{j k\alpha}(\rho,R,\omega) (\rho\pa_\rho)^j (\rho R\pa_R)^k (\rho\sV)^\alpha,\qquad \ell_{j k\alpha}\in\CI(\breve M);
\end{equation}
here $(\rho\sV)^\alpha:=(\rho V_1)^{\alpha_1}(\rho V_2)^{\alpha_2}\cdots$. The \emph{indicial operator} of $L$ is defined by
\[
  I(L) := \sum_{j\leq m} \ell_{j 0 0}(0,R,\omega) (\rho\pa_\rho)^j \in \Diffb^m(\breve M);
\]
it is a family (parameterized by $(R,\omega)\in\cI^+$) of elements of $\Diffb^m([0,\infty)_\rho)$. We shall often regard a function $u=u(R,\omega)$ as a function $(\rho,R,\omega)\mapsto u(R,\omega)$ (without spelling this out explicitly); with this convention, we have $\ell_{j 0 0}-\ell_{j 0 0}|_{\rho=0}\in\rho\CI$. Since $\rho R\pa_R,\rho V_a\in\rho\Vb$, we thus conclude that
\[
  L - I(L) \in \rho\Diffb^m(\breve M),
\]
that is, $I(L)$ captures $L$ to leading order at $\cI^+$ \emph{as a b-differential operator}. The \emph{indicial family} is obtained by formally conjugating $I(L)$ by the Mellin transforming in $\rho$, thus
\[
  I(L,\lambda) := \sum_{j\leq m} \ell_{j 0 0}(0,R,\omega) \lambda^j.
\]
For each $(R,\omega)\in\cI^+$, this is a polynomial in $\lambda$ whose roots are called \emph{indicial roots}. (They can depend on $(R,\omega)$, though the indicial roots of all operators appearing in this paper will be constant along $\cI^+$.)

The wave operators of main interest in this paper will be of 0-b type, but their solutions will be shown to be regular under application of b-vector fields (which are stronger). One underlying structural reason is the following.

\begin{lemma}[Ideal]
\label{Lemma0bIdeal}
  $\cV_{0,\bop}(\breve M)\subset\Vb(\breve M)$ is an ideal. That is, if $V\in\cV_{0,\bop}(\breve M)$ and $W\in\Vb(\breve M)$, then $[W,V]\in\cV_{0,\bop}(\breve M)$.
\end{lemma}
\begin{proof}
  If the conclusion holds for $V$, then it also holds for $f V$ where $f\in\CI(\breve M)$ since $[W,f V]=(W f)V+f[W,V]$, with both summands lying in $\cV_{0,\bop}$. Similarly, if the conclusion holds for $W$, then for $f\in\CI(\breve M)$ we have $[f W,V]=f[W,V]-(V f)W$; note then that $V f\in\rho\CI(\breve M)$, and therefore $(V f)W\in\rho\Vb(\breve M)\subset\cV_{0,\bop}(\breve M)$. It thus suffices to consider $V=\rho\pa_\rho,\rho R\pa_R,\rho V_a$ and $W=\rho\pa_\rho,R\pa_R,V_a$, in which case the membership $[W,V]\in\cV_{0,\bop}$ is straightforward to check.
\end{proof}

\subsubsection{Function spaces}

We fix on $\breve M$ and $\cI^+$ the (unweighted b-)densities $|\frac{\dd\rho}{\rho}\,\frac{\dd R}{R}\,\dd\slg|$ and $|\frac{\dd R}{R}\,\dd\slg|$, respectively.

\begin{definition}[b- and (0,b)-spaces]
\label{Def0bL}
  Fix a finite spanning set $\sV=\{V_a\}\subset\cV(\Sph^2)$. Let $\Omega\subset\breve M$ be compact and equal to the closure of its interior $\Omega^\circ$. Then for $\alpha,\beta\in\R$ and $k\in\N_0$, we define $R^\alpha\rho^\beta\Hb^k(\Omega)$ to be the space of elements of $L^2_\loc(\Omega^\circ)$ with finite norm\footnote{While this definition uses a concrete set $\rho\pa_\rho,R\pa_R,\sV$ of vector fields to test for b-regularity, we remark that any other finite set which spans the set $\Vb(\Omega)$ of smooth vector fields tangent to $\cI^+$ and $\cK$ over $\CI(\Omega)$ produces equivalent norms. Similarly, the functions $R$ and $\rho$ used as weights can be replaced by any smooth defining functions of $\cK$ and $\cI^+$, respectively, without changing the norms (up to equivalence). In this sense, the function spaces defined here only depend on the structure of $\Omega$ as a smooth manifold with corners.}
  \[
    \|u\|_{R^\alpha\rho^\beta\Hb^k(\Omega)} := \sum_{j+l+|\gamma|\leq k} \| R^{-\alpha}\rho^{-\beta} (\rho\pa_\rho)^j (R\pa_R)^l \sV^\gamma u \|_{L^2(\Omega)}.
  \]
  The space $R^\alpha\rho^\beta H_{0,\bop}^k(\Omega)$ is defined analogously but using $(\rho\pa_\rho)^j(\rho R\pa_R)^l(\rho\sV)^\gamma$; and the space $R^\alpha\rho^\beta\cC_\bop^k(\Omega)$ is defined analogously but using the $\cC^0_\bop(\Omega)$-norm which is defined to be the sup norm on the space $\cC_\bop^0(\Omega)$ of all bounded continuous functions on $\Omega\setminus(\cK\cup\cI^+)$. We similarly define $R^\alpha\Hb^k(\cU)$ for precompact $\cU\subset\cI^+$ with $\cU=\ol{\cU^\circ}$, with norm
  \[
    \|u\|_{R^\alpha\Hb^k(\cU)} := \sum_{l+|\gamma|\leq k} \| R^{-\alpha} (R\pa_R)^l\sV^\gamma u \|_{L^2(\cU)},
  \]
  and analogously $R^\alpha\cC_\bop^k(\cU)$.
\end{definition}

When $\Omega\subset[0,\infty)\times\cU$, operator classes with non-smooth coefficients are denoted
\[
  R^\alpha\rho^\beta\Hb^k(\Omega)\Diff_{0,\bop}^m(\Omega),\qquad
  R^\alpha\Hb^k(\cU)\Diff_{0,\bop}^m(\Omega);
\]
elements of these spaces are finite sums of products $a L$ where $a$ lies in the stated function space and $L\in\Diff_{0,\bop}^m(\Omega)$. (Recall here our convention of regarding a function of $(R,\omega)$ as a $\rho$-independent function of $(\rho,R,\omega)$.) Given $L\in R^\alpha\rho^\beta\Hb^k(\Omega)\Diff_{0,\bop}^m(\Omega)$, we define its norm as follows: there is a unique way of writing
\[
  L = \sum_{j+|\gamma|\leq m} \ell_{j\gamma} (\tau\pa_\tau)^j (\tau\pa_x)^\gamma,\qquad \ell_{j\gamma}\in R^\alpha\rho^\beta\Hb^k(\Omega),
\]
and we then set
\begin{equation}
\label{Eq0bOpNorm}
  \|L\|_{R^\alpha\rho^\beta\Hb^k(\Omega)\Diff_{0,\bop}^m(\Omega)} := \sum_{j+|\gamma|\leq m} \|\ell_{j\gamma}\|_{R^\alpha\rho^\beta\Hb^k(\Omega)};
\end{equation}
similarly for $L\in R^\alpha\Hb^k(\cU)\Diff_{0,\bop}^m(\Omega)$. Since $\Vb(\breve M)$ is spanned over $\CI(\breve M)$ by $\tau\pa_\tau$, $\rho^{-1}\tau\pa_{x^i}=R\pa_{x^i}$, we can similarly define\footnote{We take advantage of the particular geometry of $\breve M$ here. A more systematic approach towards defining a norm, which ends up giving an equivalent norm, is to cover $\Omega$ by coordinate charts and to sum the norms of the coefficients in the local coordinate chart expressions of $L$ as a b-differential operator.} the norm of $L\in R^\alpha\rho^\beta\Hb^k(\Omega)\Diffb^m(\Omega)$ via
\begin{equation}
\label{Eq0bbOp}
  L = \sum_{j+|\gamma|\leq m}\ell_{j\gamma}(\tau\pa_\tau)^j(R\pa_x)^\gamma \implies \|L\|_{R^\alpha\rho^\beta\Hb^k(\Omega)\Diffb^m(\Omega)} := \sum_{j+|\gamma|\leq m} \|\ell_{j\gamma}\|_{R^\alpha\rho^\beta\Hb^k(\Omega)}.
\end{equation}

We next discuss the algebra properties of weighted b-Sobolev spaces. \emph{For concreteness, we henceforth work only on $\cU=\cI^+_{R_0}:=\cI^+\cap\{R\leq R_0\}=[0,R_0]_R\times\Sph^2$ and domains $\Omega=\Omega_{\rho_0,R_0}$ from Definition~\usref{Def0bDom}.} (In particular, $[0,\infty)_\rho\times\cU\supset\Omega$.) As a useful technical tool, we introduce extension operators. We write $\CIdot(\Omega)$ for the space of smooth functions which vanish to infinite order at $\cI^+\cup\cK$; we stress that we do not require its elements to vanish at any other boundary hypersurfaces of $\Omega$.

\begin{lemma}[Extension operators]
\label{Lemma0bExt}
  Let $0<R_0<R_+<R_1$, $0<\rho_0<\rho_+<\rho_1$. Then there exists a continuous linear map $\Xi\colon\CIdot(\Omega_{\rho_0,R_0})\to\CIdot(\Omega_{\rho_1,R_1})$ with the following properties.
  \begin{enumerate}
  \item{\rm (Extension.)} For all $u\in\CIdot(\Omega_{\rho_0,R_0})$, we have $(\Xi u)|_{\Omega_{\rho_0,R_0}}=u$.
  \item{\rm (Support.)} For all $u\in\CIdot(\Omega_{\rho_0,R_0})$, we have $\supp(\Xi u)\subset\Omega_{\rho_+,R_+}$.
  \item{\rm (Boundedness.)} $\Xi$ defines a bounded map $R^\alpha\rho^\beta\Hb^k(\Omega_{\rho_0,R_0})\to R^\alpha\rho^\beta\Hb^k(\Omega_{\rho_1,R_1})$ for all $\alpha,\beta\in\R$ and $k\in\N_0$, similarly for weighted $\cC_\bop^k$-spaces.
  \end{enumerate}
  There exist extension operators $\CIdot(\cI^+_{R_0})\to\CIdot(\cI^+_{R_1})$ with the analogous properties.
\end{lemma}
\begin{proof}
  This is a variant of Seeley's theorem \cite{SeeleyExtension}. We only discuss the extension problem near the boundary hypersurface $y:=R(1-2\rho)-R_0(1-2\rho_0)=0$ of $\Omega_{\rho_0,R_0}$, and thus consider $u$ with support in $-\delta\leq y\leq 0$ and $0\leq\rho\leq\delta$. (Local extension operators can be patched together using a partition of unity.) Write $r=-\log\rho$, and denote by $\omega\in\R^2$, $|\omega|<2$, local coordinates on $\Sph^2$. Set $B(0,1)=\{\omega\in\R^2\colon|\omega|<1\}$. Then the norm of a function $u\in\CIc([-\log\delta,\infty)_r\times(-\delta,0]_y\times B(0,1)_\omega)$ on $\Hb$ and $\cC_\bop$ spaces are $L^2$-, resp.\ $\cC^0$-norms of the derivatives of $u$ along $\pa_r,\pa_y,\pa_\omega$. Fix $\chi\in\CIc((-1,0])$ to be equal to $1$ near $0$, and with $\supp\chi\subset(-\min(\delta,\eta),0]$ for some small $\eta>0$ fixed below. Set then
  \[
    \tilde u(r,y,\omega) := \begin{cases} u(r,y,\omega), & y<0, \\ \sum_{l=0}^\infty c_l\chi(-y/\delta_l)u(r,-y/\delta_l,\omega), & y>0, \end{cases}
  \]
  where we take $\delta_l=3^{-l}$ and define $c_l$ via $\sin(\frac{\pi}{2}z)=\sum_{l=0}^\infty c_l z^l$. This ensures that $\sum_{l=0}^\infty c_l\delta_l^{-j}=(-1)^j$ and $\sum_{l=0}^\infty |c_l|\delta_l^{-j}<\infty$ for all $j\in\N_0$, and thus $\tilde u$ is smooth across $y=0$. Moreover, by the support property of $\chi$, we have $\tilde u=0$ for $y\geq\eta$, which implies $\supp\tilde u$ is contained in (the intersection of $\{r\geq-\log\delta,\ \omega\in B(0,1)\}$ with) $\Omega_{\rho_+,R_+}$ when $\eta$ is sufficiently small. The boundedness of $u\mapsto\tilde u$ on $\Hb$- and $\cC_\bop$-spaces follows by direct differentiation.
\end{proof}

\begin{lemma}[Boundedness, restriction, and algebra properties]
\label{Lemma0bSob}
  Write $d_m:=\lceil\frac{m+1}{2}\rceil$ for the smallest integer larger than $\frac{m}{2}$, so $d_4=3$ and $d_3=2$. We write $D_\bop^p$ for any $p$-fold composition ($p\in\N_0$) of elements of $\{\rho\pa_\rho,R\pa_R,V_a\}$, resp.\ $\{R\pa_R,V_a\}$ acting on a function on $\Omega=\Omega_{\rho_0,R_0}$, resp.\ $\cU=\cI^+_{R_0}$.
  \begin{enumerate}
  \item\label{It0bSobEmbed}{\rm (Sobolev embedding.)} For every $k\in\N_0$, there exist constants $C^S_{\Omega,k}$, $C^S_{\cU,k}$ so that
    \begin{equation}
    \label{Eq0bSobEmbed}
      \|u\|_{\cC_\bop^k(\Omega)}\leq C^S_{\Omega,k}\|u\|_{\Hb^{k+d_4}(\Omega)},\qquad
      \|v\|_{\cC_\bop^k(\cU)}\leq C^S_{\cU,k}\|v\|_{\Hb^{k+d_3}(\cU)}.
    \end{equation}
  \item\label{It0bSobTrace}{\rm (Restrictions.)} Let $\rho_1\in(0,\rho_0]$, and write $\Sigma:=\Omega_{\rho_0,R_0}\cap\{\rho=\rho_1\}$. Then the restriction map $u\mapsto u|_\Sigma$ defines a bounded linear map
    \begin{equation}
    \label{Eq0bSobTrace}
      \Hb^k(\Omega) \to \Hb^{k-1}(\Sigma),\qquad k\in\N.
    \end{equation}
  \item\label{It0bSobProd}{\rm (Estimates for products.)} For every $k\in\N_0$, there exist constants $C_{\Omega,k}$, $C_{\cU,k}$, and $C_{\Omega,\cU,k}$ so that for all functions $u_1,u_2$ on $\Omega$ and $v_1,v_2$ on $\cU$, and for $a,b\in\N_0$ with $a+b=k$, we have
    \begin{subequations}
    \begin{align}
    \label{Eq0bSobProd1}
      \|(D_\bop^a u_1)(D_\bop^b u_2)\|_{L^2(\Omega)} &\leq C_{k,\Omega}\bigl( \|u_1\|_{\Hb^{d_4}(\Omega)}\|u_2\|_{\Hb^k(\Omega)} + \|u_1\|_{\Hb^{k+d_4}(\Omega)}\|u_2\|_{L^2(\Omega)} \bigr), \\
    \label{Eq0bSobProd2}
      \|(D_\bop^a v_1)(D_\bop^b v_2)\|_{L^2(\cU)} &\leq C_{k,\cU}\bigl( \|v_1\|_{\Hb^{d_3}(\cU)}\|v_2\|_{\Hb^k(\cU)} + \|v_1\|_{\Hb^{k+d_3}(\cU)}\|v_2\|_{L^2(\cU)} \bigr), \\
    \label{Eq0bSobProd3}
      \|(D_\bop^a v_1)(D_\bop^b u_2)\|_{L^2(\Omega)} &\leq C_{k,\Omega,\cU}\bigl( \|v_1\|_{\Hb^{d_3}(\cU)}\|u_2\|_{\Hb^k(\cU)} + \|v_1\|_{\Hb^{k+d_3}(\Omega)}\|u_2\|_{L^2(\Omega)} \bigr).
    \end{align}
    \end{subequations}
  \item\label{It0bSobFn}{\rm (Estimates for nonlinear expressions.)} Let $\delta>0$, and let $F\colon[-\delta,\delta]\to\R$ be a smooth function which vanishes at $0$. Then for all $u\in\Hb^k(\Omega)$ with\footnote{It suffices to assume $\|u\|_{L^\infty}\leq\delta$ (which is implied by the stated Sobolev bound), though we do not need this (classical) result here.}  $\|u\|_{\Hb^{d_4}}\leq \delta/C_{\Omega,d_4}^S$, we have $F(u)\in\Hb^k(\Omega)$ and
    \[
      \|F(u)\|_{\Hb^k(\Omega)} \leq C_{k,\Omega,F}\|u\|_{\Hb^{k+d_4}(\Omega)}.
    \]
    Similarly, for $v\in\Hb^k(\cU)$ with $\|v\|_{\Hb^{d_3}}\leq \delta/C_{\cU,d_3}^S$ we have $F(v)\in\Hb^k(\cU)$, with norm bounded by a constant times $\|v\|_{\Hb^{k+d_3}(\cU)}$.
  \end{enumerate}
\end{lemma}

It is straightforward to prove stronger bounds using classical Moser estimates, see e.g.\ \cite[Chapter~13]{TaylorPDE3}; we opt for weaker statements which have simple self-contained proofs.

The estimates in this Lemma admit straightforward generalizations to weighted spaces, for example $\|u\|_{R^\alpha\rho^\beta\cC_\bop^k(\Omega)}\leq C_{\Omega,k,\alpha,\beta}\|u\|_{R^\alpha\rho^\beta\Hb^{k+d_4}(\Omega)}$ for all $\alpha,\beta\in\R$, which we do not spell out here (but use frequently in~\S\ref{SE} below). Regarding~\eqref{Eq0bSobTrace}, the weighted generalization reads $R^\alpha\rho^\beta\Hb^k(\Omega)\to R^\alpha\Hb^{k-1}(\Sigma)$ where $\alpha,\beta\in\R$; note that the $\rho$-weight is irrelevant since we restrict to $\rho=\rho_1>0$.

The main application of part~\eqref{It0bSobFn} is to control $\frac{1}{f+u}$ where $f$ is a smooth function bounded away from $0$ and $u\in\Hb^k(\Omega)$. (This arises when inverting matrices with Sobolev-regular coefficients.) Note that $\frac{1}{f+u}=\frac{1}{f}(1-F(u/f))$ where $F(x)=\frac{x}{1+x}$ is smooth on $[-\delta,\delta]$ for any $\delta<1$; if $\|u\|_{\Hb^{d_4}}\leq\frac{\delta}{C_{\cU,d_4}^S}\|f\|_{L^\infty}$, we thus obtain $\|\frac{1}{f+u}-\frac{1}{f}\|_{\Hb^k}\leq C_{f,k,\delta,\Omega}\|u\|_{\Hb^k}$.

\begin{proof}[Proof of Lemma~\usref{Lemma0bSob}]
  Defining $\tilde u_1:=\Xi u_1$ using the extension operator from Lemma~\ref{Lemma0bSob}, and defining $\tilde u_2,\tilde v_1,\tilde v_2$ similarly, it suffices to prove the claimed estimates with $\Omega,\cU$ replaced by $\Omega_{\rho_1,R_1},\cI^+_{R_1}$ (where $\rho_1>\rho_0$ and $R_1>R_0$ are arbitrary), and with all functions now vanishing near the boundary hypersurfaces of $\Omega$, resp.\ $\cU$ at $\rho=\rho_1$, $R(1-2\rho)=R_1(1-2\rho_1)$, resp.\ $R=R_1$. We now relabel $\rho_1$, $R_1$, $\tilde u_1$, etc.\ as $\rho_0$, $R_0$, $u_1$, etc.

  In the coordinates $t=-\log R$ and $r=-\log\rho$ then, we have $\rho\pa_\rho=-\pa_r$ and $R\pa_R=-\pa_t$. Therefore, the estimates in~\eqref{Eq0bSobEmbed} are simply instances of standard Sobolev embedding on $\R_t\times\R_r\times\Sph^2$, resp.\ $\R_t\times\Sph^2$. Similarly, part~\eqref{It0bSobTrace} is an instance of the standard trace theorem, in the form of the continuity of the restriction map $H^s(\R\times\R\times\Sph^2)\to H^{s-\frac12}(\R\times\{0\}\times\Sph^2)$, $s>\frac12$; we work only with integer orders here and thus a fortiori obtain~\eqref{Eq0bSobTrace}.

  We follow the proof of \cite[Lemma~3.33]{HintzMink4Gauge} for part~\eqref{It0bSobProd}. Introducing local coordinates $\omega\in\R^2$ on $\Sph^2$, we need to prove estimates for functions on $\R^4_{t,r,\omega}$, resp.\ $\R^3_{t,\omega}$. It suffices to consider compactly supported $u_i,v_i$. We write $x=(t,\omega)$, and we write $\pa\in\{\pa_t,\pa_{\omega^1},\pa_{\omega^2},\pa_r\}$ for any coordinate derivative. Then
  \[
    \|\pa u\|_{L^2}^2=\iint \pa(u\ol{\pa u})\,\dd x\,\dd r + \iint u\ol{\pa^2 u}\,\dd x\,\dd r = \iint u\ol{\pa^2 u}\,\dd x\,\dd r\leq\|u\|_{L^2}\|\pa^2 u\|_{L^2},
  \]
  which implies for any $0\leq p<q$ the estimate
  \begin{equation}
  \label{Eq0bSobL2}
    \|\pa^p u\|_{L^2} \leq C_{p q}\|u\|_{L^2}^{\frac{q-p}{q}}\|\pa^q u\|_{L^2}^{\frac{p}{q}}.
  \end{equation}
  Passing back to b-spaces in our notation, we now estimate
  \[
    \|(D_\bop^a u_1)(D_\bop^b u_2)\|_{L^2}\leq \|D_\bop^a u_1\|_{L^\infty}\|D_\bop^b u_2\|_{L^2}\leq C\|D_\bop^a u_1\|_{\Hb^{d_4}}\|D_\bop^b u_2\|_{L^2}
  \]
  (using Sobolev embedding), further $\|D_\bop^a u_1\|_{\Hb^{d_4}}\leq C\sum_{d=0}^{d_4}\|D_\bop^{a+d}u_1\|_{L^2}$, and finally, for $d=0,\ldots,d_4$ and recalling that $a+b=k$,
  \begin{align*}
    \|D_\bop^a(D_\bop^d u_1)\|_{L^2}\|D_\bop^b u_2\|_{L^2} &\lesssim \|D_\bop^d u_1\|_{L^2}^{\frac{b}{k}}\|D_\bop^{k+d}u_1\|_{L^2}^{\frac{a}{k}} \cdot \|u_2\|_{L^2}^{\frac{a}{k}}\|D_\bop^k u_2\|_{L^2}^{\frac{b}{k}} \\
      &\lesssim \bigl(\|u_1\|_{\Hb^{d_4}}\|u_2\|_{\Hb^k} + \|u_1\|_{\Hb^{k+d_4}}\|u_2\|_{L^2}\bigr),
  \end{align*}
  which is~\eqref{Eq0bSobProd1}. The proofs of~\eqref{Eq0bSobProd2}--\eqref{Eq0bSobProd3} are completely analogous, now using Sobolev embedding on $\R^3_x$.

  To prove part~\eqref{It0bSobFn}, write $F(x)=x F_1(x)$; then $F(u)=u F_1(u)$ with $F_1(u)$ bounded (since $\|u\|_{L^\infty}<\delta$), and therefore $\|F(u)\|_{L^2}\lesssim \|u\|_{L^2}$. Next, $D_\bop^k F(u)$ is a sum of terms of the form $\prod_{j=1}^N(D_\bop^{k_j}u)F^{(N)}(u)$ where $k_1,\ldots,k_N\geq 1$ with $\sum_{j=1}^N k_j=k$ (thus $N\leq k$). It thus remains to estimate, via Sobolev embedding on the first $N-1$ factors,
  \begin{align*}
    \|(D_\bop^{k_1}u)\cdots(D_\bop^{k_N}u)\|_{L^2} &\lesssim \sum_{d=0}^{d_4} \|D_\bop^{k_1}(D_\bop^d u)\|_{L^2}\cdots\|D_\bop^{k_{N-1}}(D_\bop^d u)\|_{L^2}\|D_\bop^{k_N}u\|_{L^2} \\
      &\lesssim \sum_{d=0}^{d_4}\|D_\bop^d u\|_{L^2}^{\frac{k-k_1}{k}}\|D_\bop^d u\|_{\Hb^k}^{\frac{k_1}{k}}\cdots \|D_\bop^d u\|_{L^2}^{\frac{k-k_{N-1}}{k}}\|D_\bop^d u\|_{\Hb^k}^{\frac{k_{N-1}}{k}} \|u\|_{L^2}^{\frac{k-k_N}{k}}\|u\|_{\Hb^k}^{\frac{k_N}{k}} \\
      &\lesssim \|u\|_{\Hb^{d_4}}^{\frac{k_N}{k}}\|u\|_{\Hb^{k+d_4}}^{\frac{k-k_N}{k}} \|u\|_{L^2}^{\frac{k-k_N}{k}}\|u\|_{\Hb^k}^{\frac{k_N}{k}},
  \end{align*}
  which is bounded by $C\|u\|_{\Hb^{k+d}}$, as claimed.
\end{proof}

For Nash--Moser purposes, we record:

\begin{lemma}[Smoothing operators]
\label{Lemma0bSmooth}
  Let $\alpha,\beta\in\R$. There exist continuous linear maps
  \[
    S_\theta\colon R^\alpha\rho^\beta L^2(\Omega_{\rho_0,R_0})\to R^\alpha\rho^\beta\Hb^\infty(\Omega_{\rho_0,R_0}),\qquad \theta>1,
  \]
  so that
  \begin{alignat*}{3}
    k&\leq k' &&\implies \|S_\theta u-u\|_{R^\alpha\rho^\beta\Hb^k}\;&\leq C_{k,k'}\theta^{k-k'}\|u\|_{R^\alpha\rho^\beta\Hb^{k'}}, \\
    k&\geq k' &&\implies \|S_\theta u\|_{R^\alpha\rho^\beta\Hb^k} &\leq C_{k,k'}\theta^{k-k'}\|u\|_{R^\alpha\rho^\beta\Hb^{k'}}.
  \end{alignat*}
  There exist continuous linear maps $R^\alpha L^2(\cI^+_{R_0})\to R^\alpha\Hb^\infty(\cI^+_{R_0})$ with the analogous properties.
\end{lemma}
\begin{proof}
  Given $u\in R^\alpha\rho^\beta L^2(\Omega_{\rho_0,R_0})$, set $\tilde u:=\Xi u$ where $\Xi$ is an extension operator from Lemma~\ref{Lemma0bExt}. Pass to $t=-\log R$ and $r=-\log\rho$, and work in local coordinates $\omega\in\R^2$, $|\omega|<3$ on $\Sph^2$; using a partition of unity, we only consider the case that $\tilde u$ is supported in the subset $\{|\omega|<1\}$ of $\R^4_z$ where $z=(t,r,\omega)$. Fix $\phi\in\CIc(\R^4)$ so that its Fourier transform $\hat\phi(\zeta)$ is equal to $1$ near $\zeta=0$. Let $\chi\in\CIc(\R^2_\omega)$ be equal to $1$ for $|\omega|<2$ and supported in $\{|\omega|<3\}$. We then define
  \[
    (S_\theta u)(z):=\chi(\omega)\int_{\R^4} \theta^4\phi(\theta w)u(z-w)\,\dd w,\qquad z=(t,r,\omega).
  \]
  Without the cutoff $\chi$, this is the same construction as in \cite[Appendix]{SaintRaymondNashMoser}. With the cutoff $\chi$ present, only the estimate on $S_\theta u-u$ requires a bit of care; but the point is simply that $(1-\chi(\omega))\int_{\R^4}\theta^4\phi(\theta w)u(z-w)\,\dd w$ is bounded in every Sobolev space by $\theta^{-N}$ for all $N$ since the supports of $1-\chi$ and $u$ are disjoint. (See \cite[Lemma~6.12]{HintzGlueLocII} for details in a more precise construction.)
\end{proof}

\section{Initial value problems for the gauge-fixed Einstein equations}
\label{SE}

We study perturbations of the KdS metric $g_\sfb$ from~\eqref{Eq0bKdS} on the domain $\Omega_{\rho_0,R_0}$ from Definition~\ref{Def0bDom}, with initial data posed at $\Sigma_{\rho_0,R_0}$. We fix cutoffs
\begin{equation}
\label{EqECutoffs}
  \chi=\chi(\rho)\in\CIc([0,\tfrac12\rho_0)),\ \ \chi|_{[0,\frac14\rho_0]}=1,\qquad
  \tilde\chi=\tilde\chi(\rho)\in\CIc([0,\rho_0)),\ \ \chi|_{[0,\frac12\rho_0]}=1,
\end{equation}
and regard them as functions on $\breve M$ (which thus equal $1$ near $\cI^+$). We shall work in the generalized harmonic gauge $\Ups(g;g_0)+E_{g_0}(g-g_0)=0$ for a suitable (dynamically chosen) `background metric' $g_0$ where
\begin{equation}
\label{EqEGauge}
\begin{split}
  \Ups(g;g_0) &:= g(g_0)^{-1}\delta_g\sfG_g g_0\quad \bigl(\text{in coordinates: $\Ups(g;g_0)_\mu=g_{\mu\nu}g^{\kappa\lambda}(\Gamma(g)_{\kappa\lambda}^\nu - \Gamma(g_0)_{\kappa\lambda}^\nu)$}\bigr), \\
  E_{g_0} h &:= \chi e^0(-2\tr_{g_0} h - h(e_0,e_0)),\qquad e_0:=\tau\pa_\tau,\ \ e^0:=\frac{\dd\tau}{\tau}.
\end{split}
\end{equation}
For sections $h_0$ of $\upbeta^*(\tau^{-2}S^2 T^*X)$ over $\cI^+_{R_0}$ (which will capture the leading order change to $g_\sfb$ at $\cI^+$) and $\tilde h$ of $\upbeta^*(S^2\,{}^0 T^*M)$ over $\Omega_{\rho_0,R_0}$ (which will capture further decaying corrections to $g_\sfb$), we then define the gauge-fixed Einstein operator
\begin{equation}
\label{EqEEinOp}
\begin{split}
  P(h_0,\tilde h,\theta) &:= 2\Bigl(\Ric(g_\sfb+\chi h_0+\tilde h) - \Lambda(g_\sfb+\chi h_0+\tilde h) \\
    &\quad\hspace{2em} - \tilde\delta_{g_\sfb+\chi h_0+\tilde h}^*\bigl(\Ups(g_\sfb+\chi h_0+\tilde h;g_\sfb+\chi h_0) + E_{g_\sfb+\chi h_0}\tilde h - \tilde\chi\theta\bigr)\Bigr),
\end{split}
\end{equation}
where (for $g=g_\sfb+\chi h_0+\tilde h$) we set
\begin{equation}
\label{EqECD}
  \tilde\delta_g^* = \delta_g^* + \tilde E,\qquad \tilde E\omega:=\chi\bigl(2\omega(e_0)e^0\otimes e^0-4 e^0\otimes_s\omega\bigr).
\end{equation}
The structure of this operator was already motivated in~\S\ref{SssIEdS}. The specific choices for the gauge modification $E_{g_0}$ and the modified symmetric gradient $\tilde\delta_g^*$ will be explained in~\S\ref{SssE1Aux} below. The choice $g_\sfb+\chi h_0$ of background metric ensures that, for $\tilde h=g-g_0$ which decay towards $\cI^+$ (together with their derivatives along 0-vector fields), the gauge condition is always satisfied to leading order at $\cI^+$. The decision to use the modified symmetric gradient with respect to $g_\sfb+\chi h_0+\tilde h$ is due to the fact that then its indicial operator involves the induced boundary metric $\dd x^2+h_{(0)}$ (where $h_{(0)}=\tau^2 h_0$), which is geometrically more natural than just $\dd x^2$ (cf.\ \eqref{EqE1Structdelstar} below) and thus ultimately makes the computation of the indicial roots $L_{h_{(0)},\tilde h}$ in Lemma~\ref{LemmaEIndRoot} more transparent.

The main goal of this section, achieved in~\S\ref{SsEPf}, is the global solution of small data initial value problems for $P$ in the following sense.

\begin{thm}[Solution of the gauge-fixed Einstein equations in the cosmological region]
\label{ThmESol}
  Let $\alpha>0$, $d\in\N$, $\delta_0>0$. Then there exist $D\in\N_0$ and $\eps>0$ so that the following holds. Let
  \[
    \ubar h_0,\ubar h_1\in R^\alpha\Hb^\infty(\Sigma_{\rho_0,R_0};\upbeta^*(S^2\,{}^0 T^*M)),
  \]
  and suppose that $\|h_{\Sigma,j}\|_{R^\alpha\Hb^D}<\eps$, $j=0,1$. Let $\beta\in(0,1)$. Then there exist
  \begin{equation}
  \label{EqESolh}
  \begin{split}
    h_0&\in R^\alpha\Hb^\infty\bigl(\cI^+_{R_0};\upbeta^*(\tau^{-2}S^2 T^*X)\bigr), \\
    \tilde h&\in R^\alpha\rho^\beta\Hb^\infty\bigl(\Omega_{\rho_0,R_0};\upbeta^*(S^2\,{}^0 T^*M)\bigr), \\
    \theta&\in R^\alpha\rho^\beta\Hb^\infty\bigl(\Omega_{\rho_0,R_0};\upbeta^*({}^0 T^*M)\bigr),
  \end{split}
  \end{equation}
  with weighted $\Hb^d$-norms less than $\delta_0$, so that $P(h_0,\tilde h,\theta)=0$, and so that $\tilde h$ (and thus also $\chi h_0+\tilde h$) satisfies the initial conditions $\tilde h|_{\Sigma_{\rho_0,R_0}}=\ubar h_0$, $(\cL_{-\rho\pa_\rho}\tilde h)|_{\Sigma_{\rho_0,R_0}}=\ubar h_1$.
\end{thm}

The nonlinear stability of the expanding region $\Omega_{\rho_0,R_0}$ of KdS is a simple consequence. This is proved in~\S\ref{SSt}, together with sharper asymptotics for the metric in this case.

\begin{rmk}[Domain of definition of $h_0$]
\label{RmkEStabDomh0}
  Since $R\leq R_0$ on $\Omega_{\rho_0,R_0}$, the tensor $g_\sfb+\chi h_0+\tilde h$ is well-defined on $\Omega_{\rho_0,R_0}$. Note that $\Omega_{\rho_0,R_0}\cap\cI^+=\cI^+_{(1-2\rho_0)R_0}\subsetneq\cI^+_{R_0}$; thus, $h_0$ is defined on a larger set than the set where asymptotic data for $h$ should live or be relevant. In the proof, $h_0$ will simply arise via an extension operator from Lemma~\ref{Lemma0bExt} applied to a tensor on $\cI^+_{(1-2\rho_0)R_0}$.
\end{rmk}

\begin{rmk}[Origin of initial data]
\label{RmkEOrigInitial}
  The nonlinear stability result for slowly rotating KdS black holes proved in \cite{HintzVasyKdSStability,FangKdS} produces a solution of the (gauge-fixed) Einstein vacuum equations in a neighborhood $\{r_--\delta_-\leq r\leq r_++\delta_+\}$, $\delta_\pm>0$, of the domain of outer communications of the black hole; here $r_-$ and $r_+$ are the radii of the event and cosmological horizon, respectively. In fact, $\delta_+>0$ can be taken to be arbitrarily large (but fixed), with the smallness requirement on the initial data depending on $\delta_+$. The set $\{\rho=\rho_0,\ R<R_0\}$ is then contained in $r^{-1}((r_+,r_++\delta_+])$, and the decay assumptions on the initial data correspond exactly to exponential decay of the coefficients of $h_{\rm in}$ in the frame $\pa_t,\pa_r,\pa_\omega$ in $\tilde t=t\sim-\log R$ (cf.\ \eqref{Eq0bKdStilde}) together with all coordinate derivatives. --- More generally, without assuming that $g_\sfb$ is slowly rotating, if initial data for~\eqref{EqIEin} of class $R^\alpha\Hb^\infty$ are posed at $\rho=\rho_0$, then these data can be evolved in a standard generalized harmonic gauge (i.e.\ $\Ups(g;g_\sfb)=0$ in the notation of~\eqref{EqEGauge} below) up to any fixed hypersurface $\rho=\rho_1$, provided the data are sufficiently small (depending on $\rho_0,\rho_1$) in the $R^\alpha\Hb^N$-norm for some fixed $N$. This holds more generally for the solution of the corresponding gauge-fixed Einstein equations, without the need to require the validity of the constraint equations at $\rho=\rho_0$. This follows easily from energy estimates with multipliers $R^{-2\alpha}e^{-C r}\pa_r$ (or $e^{2\alpha t}e^{-C r}\pa_r$) for sufficiently large $C>1$, due to the timelike nature of $\dd r$ in the cosmological region; we leave the details to the reader. In a similar vein, one can evolve $R^\alpha\Hb^\infty$-perturbations of KdS data posed at $r=r_0$, with $r_0$ larger than the largest root of $\mu(r)$ in~\eqref{Eq0bKdS} up to $r=r_1$ for any fixed finite $r_1>r_0$, and thus again cover the set $\{\rho=\rho_0,\ R<R_0\}$.
\end{rmk}

\begin{rmk}[Gauge choices here and in~\cite{HintzVasyKdSStability}]
\label{RmkEKdS}
  Theorem~\ref{ThmESol} will yield the nonlinear stability of the expanding region in the gauge
  \begin{equation}
  \label{EqEKdSGaugeHere}
    \Ups(g;g_0)+E_{g_0}(g-g_0)-\tilde\chi\theta=0.
  \end{equation}
  Near the Cauchy hypersurface $\Sigma_{\rho_0,R_0}$, this reduces to $\Ups(g;g_\sfb)=0$. By contrast, the gauge in which the asymptotically KdS metric $g$ is found in \cite{HintzVasyKdSStability} is
  \begin{equation}
  \label{EqEKdSGaugeThere}
    \Ups(g;g_{\sfb_0})-\Ups(g_\sfb;g_{\sfb_0})-\theta'=0,
  \end{equation}
  where $\sfb=(\bhm,\bha)$ denotes the parameters of the final KdS black hole, $\sfb_0:=(\bhm,0)$, and $\theta'\in\CIc(M^\circ;T^*M^\circ)$. To make the output of \cite{HintzVasyKdSStability} directly compatible with the input of Theorem~\ref{ThmESol} (and Theorem~\ref{ThmEStabEin} below), one can simply modify the present gauge condition~\eqref{EqEKdSGaugeHere} so that near $\Sigma_{\rho_0,R_0}$---and thus away from $\rho=0$ which is the main focus of the current work---it becomes~\eqref{EqEKdSGaugeThere}. A concrete such choice is
  \begin{equation}
  \label{EqEKdSMix}
    \Ups(g;\chi g_0+(1-\chi)g_{\sfb_0}) - (1-\chi)\Ups(g_\sfb;g_{\sfb_0}) - (1-\chi)\theta' + E_{g_0}(g-g_0) - \tilde\chi\theta = 0.
  \end{equation}
  The unperturbed KdS case corresponds to $h_0=\tilde h=0$ and $\theta'=0$, in which case the gauge 1-form (i.e.\ the left hand side of~\eqref{EqEKdSMix}, with $\theta=0$) equals $\theta_0:=\Ups(g_\sfb;\chi g_\sfb+(1-\chi)g_{\sfb_0})-(1-\chi)\Ups(g_\sfb;g_{\sfb_0})$ (which satisfies $\theta_0=\tilde\chi\theta_0$). In the proof of Theorem~\ref{ThmESol} with general initial data, one thus constructs $\theta$ in the form $\theta_0+\theta^\flat$, with $\theta^\flat\in R^\alpha\rho^\beta\Hb^\infty$ determined by the iteration scheme.
\end{rmk}

\textit{For the remainder of this section, we fix $\alpha>0$ and $\beta\in(0,1)$.}

\subsection{Linearized gauge-fixed Einstein operator I: structure and indicial family}
\label{SsE1}

Assuming that $g:=g_\sfb+\chi h_0+\tilde h$ and $g_0:=g_\sfb+\chi h_0$ are Lorentzian metrics on $\Omega_{\rho_0,R_0}\cap\breve M^\circ$, we can use \cite[\S{2}]{GrahamLeeConformalEinstein} to compute the linearization of $P(h_0,\cdot)$ at $\tilde h$ to be
\begin{align}
\label{EqE1LExpr}
\begin{split}
  &L_{h_0,\tilde h} := D_2 P|_{\tilde h}(h_0,\cdot) = \Box_g - 2\Lambda + 2\tilde E\delta_g\sfG_g + 2\sR_g + 2\tilde\delta_g^*\circ(\sE_{g;g_0} - E_{g_0}), \\
  &\hspace{13em} + (D_g\tilde\delta^*_\cdot)(\Ups(g;g_0)+E_{g_0}\tilde h)
\end{split} \\
  &\qquad (\sR_g u)_{\mu\nu} = \Riem(g)_{\kappa\mu\nu\lambda}u^{\kappa\lambda} + \frac12\bigl(\Ric(g)_{\mu\lambda}u_\nu{}^\lambda + \Ric(g)_{\nu\lambda}u_\mu{}^\lambda\bigr), \nonumber\\
  &\qquad (\sE_{g;g_0} u)_\mu = \bigl(\Gamma(g)_{\kappa\nu}^\lambda - \Gamma(g_0)_{\kappa\nu}^\lambda\bigr)\bigl(g_{\mu\lambda}u^{\kappa\nu} - u_{\mu\lambda}g^{\kappa\nu}\bigr), \nonumber
\end{align}
where we raise and lower indices using $g$; moreover, $(D_g\tilde\delta^*_\cdot)\eta=(D_g\delta^*_\cdot)\eta$, for a fixed 1-form $\eta$, maps a symmetric 2-tensor $h$ to $\frac{\dd}{\dd s}(\delta^*_{g+s h}\eta)|_{s=0}$. We recall $\Lambda=3$ and record
\begin{equation}
\label{EqE1D1Ups}
  D_1\Ups|_g(\cdot;g_0) = -\delta_g\sfG_g - \sE_{g;g_0}.
\end{equation}

\begin{prop}[Structure of $L_{h_0,\tilde h}$; indicial operator]
\label{PropE1Struct}
  Let $k\geq 2$ and
  \[
    h_0\in R^\alpha\Hb^k(\cI^+_{R_0};\upbeta^*(S^2 T^*X)),\qquad \tilde h\in R^\alpha\rho^\beta\Hb^k(\Omega_{\rho_0,R_0};\upbeta^*(S^2\,{}^0 T^*M))
  \]
  Suppose that $\|h_0\|_{R^\alpha\Hb^{d_3+2}}<\delta_0$ and $\|\tilde h\|_{R^\alpha\rho^\beta\Hb^{d_4+2}}<\delta_0$ for some small $\delta_0>0$ (independently of $k$). Set $g:=g_\sfb+\chi h_0+\tilde h$ and $g_0:=g_\sfb+\chi h_0$.
  \begin{enumerate}
  \item{\rm (Structure.)} As differential operators on $\Omega_{\rho_0,R_0}$ acting on sections of $\upbeta^*(S^2\,{}^0 T^*M)$, we can write
    \begin{equation}
    \label{EqE1StructL}
    \begin{split}
      &L_{h_0,\tilde h} = L_{0,0} + L_{(0),h_0} + \tilde L_{h_0,\tilde h}, \\
      &\qquad L_{0,0}\in\Diff_{0,\bop}^2, \quad
      L_{(0),h_0} \in R^\alpha\Hb^{k-2}(\cI^+_{R_0})\Diff_{0,\bop}^2, \quad
      \tilde L_{h_0,\tilde h} \in R^\alpha\rho^\beta\Hb^{k-2}(\Omega_{\rho_0,R_0})\Diff_{0,\bop}^2,
    \end{split}
    \end{equation}
    where $L_{(0),h_0}$ and $\tilde L_{h_0,\tilde h}$ satisfy the tame estimates
    \begin{equation}
    \label{EqE1StructTame}
    \begin{split}
      \|L_{(0),h_0}\|_{R^\alpha\Hb^{k-2}\Diff_{0,\bop}^2} &\leq C_k\|h_0\|_{R^\alpha\Hb^{k+d_3}}, \\
      \|\tilde L_{h_0,\tilde h}\|_{R^\alpha\rho^\beta\Hb^{k-2}\Diff_{0,\bop}^2} &\leq C_k\bigl( \|h_0\|_{R^\alpha\Hb^{k+d_3}} + \|\tilde h\|_{R^\alpha\rho^\beta\Hb^{k+d_4}}\bigr)
    \end{split}
    \end{equation}
    in the notation of Lemma~\usref{Lemma0bSob} and equation~\eqref{Eq0bOpNorm}.
  \item\label{ItE1StructInd}{\rm (Indicial operator.)} Write
    \begin{equation}
    \label{EqE1Structg0h0}
      g_{(0)}=g_{(0)}(x,\dd x):=\dd x^2+h_{(0)},\qquad h_{(0)}:=\tau^2 h_0=h_0(\tau\pa_{x^i},\tau\pa_{x^j})\,\dd x^i\,\dd x^j,
    \end{equation}
    for the (rescaled) $\upbeta^*(\tau^{-2}S^2 T^*X)$ part of $g|_{\cI^+_{R_0}}$ in the splitting~\eqref{Eq0bST0Split}. Then the indicial operator of $L_{h_0,\tilde h}$ is given by
    \begin{subequations}
    \begin{equation}
    \label{EqE1StructIndOp}
      I_{g_{(0)}}(\rho\pa_\rho) := (\rho\pa_\rho)^2 - 3\rho\pa_\rho + 2\cdot\begin{pmatrix} -2\rho\pa_\rho+5 & 0 & (\rho\pa_\rho-2)\tr_{g_{(0)}} \\ 0 & -2\rho\pa_\rho+6 & 0 \\ -g_{(0)} & 0 & g_{(0)}\tr_{g_{(0)}} \end{pmatrix},
    \end{equation}
    in the sense that we can write
    \begin{equation}
    \label{EqE1StructInd}
    \begin{split}
      &L_{h_0,\tilde h} - I_{g_{(0)}}(\rho\pa_\rho) = R_0 + \tilde R_{h_0,\tilde h}, \\
      &\qquad R_0\in\rho\Diffb^2,\qquad \tilde R_{h_0,\tilde h}\in R^\alpha\rho^\beta\Hb^{k-2}\Diffb^2,
    \end{split}
    \end{equation}
    and so that
    \begin{equation}
    \label{EqE1StructIndTame}
      \|\tilde R_{h_0,\tilde h}\|_{R^\alpha\rho^\beta\Hb^{k-2}\Diffb^2} \leq C_k\bigl(\|h_0\|_{R^\alpha\Hb^{k+d_3}}+\|\tilde h\|_{R^\alpha\rho^\beta\Hb^{k+d_4}}\bigr).
    \end{equation}
    \end{subequations}
  \end{enumerate}
\end{prop}
\begin{proof}
  We work with the (dual) frames
  \begin{equation}
  \label{EqE1StructFrame}
    e_0=\frac{\dd\tau}{\tau},\ e_i=\frac{\dd x^i}{\tau},\qquad
    e^0=\tau\pa_\tau,\ e^i=\tau\pa_{x^i}
  \end{equation}
  of ${}^0 T^*M$ and ${}^0 T M$, respectively. We use Greek letters for spacetime indices $0,1,2,3$ and Latin letters for spatial indices $1,2,3$. The metric coefficients are
  \[
    g_{\mu\nu}=g(e_\mu,e_\nu)=(g_\sfb)_{\mu\nu}+\chi(h_0)_{\mu\nu}+\tilde h_{\mu\nu} = (g_0)_{\mu\nu} + \tilde h_{\mu\nu},
  \]
  where the coefficients of $h_0$ and $\tilde h$ are real-valued functions of class $R^\alpha\Hb^k(\cI^+_{R_0})$ and $R^\alpha\rho^\beta\Hb^k(\Omega_{\rho_0,R_0})$, respectively, and $(h_0)_{\mu\nu}=0$ if at least one of $\mu,\nu$ equals $0$. In the splitting~\eqref{Eq0bST0Split}, this means
  \[
    g = (g_\dS + h_0) + (g_\sfb-g_\dS) + \tilde h \in (-1,0,\dd x^2 + h_{(0)}) + \rho^3\CI + R^\alpha\rho^\beta\Hb^k.
  \]
  The components $g^{\mu\nu}=g^{-1}(e^\mu,e^\nu)$ of the inverse metric (which is well-defined when $h_0,\tilde h$ are small in $L^\infty$) are given by Cramer's rule; and we have $g^{-1}=g_0^{-1}-g_0^{-1}\tilde h g^{-1}$, which implies that
  \begin{equation}
  \label{EqE1StructgInv}
    (g^{-1})^{\mu\nu} - (g_0^{-1})^{\mu\nu} \in R^\alpha\rho^\beta\Hb^k(\Omega_{\rho_0,R_0}),
  \end{equation}
  with norm bounded by $C_k\|\tilde h\|_{R^\alpha\rho^\beta\Hb^{k+d_4}}$ by Lemma~\ref{Lemma0bSob}. Similarly, the norm of $g_0^{-1}-g_\sfb=-\chi g_\sfb^{-1}h_0 g_0^{-1}\in R^\alpha\Hb^k(\cI^+_{R_0})$ is bounded by $C_k\|h_0\|_{R^\alpha\Hb^{k+d_3}}$. In the bundle splitting~\eqref{Eq0bST0Split}, we therefore have
  \begin{equation}
  \label{EqE1StructGg}
    \sfG_g \equiv \begin{pmatrix} \half & 0 & \half\tr_{g_{(0)}} \\ 0 & I & 0 \\ \half g_{(0)} & 0 & I-\half g_{(0)}\tr_{g_{(0)}}  \end{pmatrix} \bmod \rho^3\CI+R^\alpha\Hb^k(\cI_{R_0}^+)+R^\alpha\rho^\beta\Hb^k(\Omega_{\rho_0,R_0})
  \end{equation}
  as an endomorphism of $\upbeta^*(S^2\,{}^0 T^*M)$.

  \pfstep{Structure of $L_{h_0,\tilde h}$.} We note that $[e_0,e_0]=0=[e_i,e_j]$ and $[e_0,e_i]=e_i=-[e_i,e_0]$. Moreover, $e_0(h_0)_{\mu\nu}=e_0(g_\dS)_{\mu\nu}=0$; by Lemma~\ref{Lemma0bStruct}, we thus have $e_0(g_\sfb)_{\mu\nu}\in\rho^3\CI$ and $e_0 g_{\mu\nu}\equiv e_0(g_\sfb)_{\mu\nu}\bmod R^\alpha\rho^\beta\Hb^{k-1}$. We can now compute the connection coefficients
  \begin{align*}
    \Gamma(g)_{\lambda\mu\nu} &= g(\nabla^g_{e_\mu}e_\nu,e_\lambda) \\
      &= \tfrac12\bigl( e_\mu g_{\nu\lambda} + e_\nu g_{\mu\lambda} - e_\lambda g_{\mu\nu} - g(e_\mu,[e_\nu,e_\lambda]) - g(e_\nu,[e_\mu,e_\lambda]) + g(e_\lambda,[e_\mu,e_\nu]) \bigr).
  \end{align*}
  Recall that $e_0,e_i\in\cV_0(M)$ and thus $e_0,e_i\in\cV_{0,\bop}(\breve M)\subset\Vb(\breve M)$. Therefore,
  \begin{equation}
  \label{EqE1StructChristoffel}
  \begin{alignedat}{2}
    \Gamma(g)_{0 0 0} &= \half e_0 g_{0 0}, \qquad&
    \Gamma(g)_{\ell 0 0} &= (e_0-1)g_{0\ell} - \half e_\ell g_{0 0}, \\
    \Gamma(g)_{0 i 0} &= \half e_i g_{0 0}, \qquad&
    \Gamma(g)_{\ell i 0} &= \half( e_i g_{0\ell} + e_0 g_{i\ell} - e_\ell g_{0 i}) - g_{i\ell}, \\
    \Gamma(g)_{0 0 j} &= \half e_j g_{0 0} + g_{0 j}, \qquad&
    \Gamma(g)_{\ell 0 j} &= \half( e_0 g_{j\ell} + e_j g_{0\ell} - e_\ell g_{0 j}), \\
    \Gamma(g)_{0 i j} &= \half( e_i g_{0 j} + e_j g_{0 i} - e_0 g_{i j} ) + g_{i j}, \qquad&
    \Gamma(g)_{\ell i j} &= \half( e_i g_{j\ell} + e_j g_{i\ell} - e_\ell g_{i j} );
  \end{alignedat}
  \end{equation}
  and we have
  \begin{align*}
    \Gamma(g)_{\lambda\mu\nu}-\Gamma(g_0)_{\lambda\mu\nu}&\in R^\alpha\rho^\beta\Hb^{k-1}(\Omega_{\rho_0,R_0}), \\
    \Gamma(g)_{\lambda\mu\nu}-\Gamma(g_\sfb)_{\lambda\mu\nu}&\in R^\alpha\Hb^{k-1}(\cI^+_{R_0}) + R^\alpha\rho^\beta\Hb^{k-1}(\Omega_{\rho_0,R_0}).
  \end{align*}
  In particular, $\Gamma(g)_{\lambda\mu\nu}\in\rho^3\CI+R^\alpha\Hb^{k-1}+R^\alpha\rho^\beta\Hb^{k-1}$ for all $\mu,\nu,\lambda$. In view of~\eqref{EqE1StructgInv}, we also obtain
  \[
    \Gamma(g)_{\mu\nu}^\kappa-\Gamma(g_0)_{\mu\nu}^\kappa\in R^\alpha\rho^\beta\Hb^{k-1},\qquad
    \Gamma(g)_{\mu\nu}^\kappa-\Gamma(g_\sfb)_{\mu\nu}^\kappa\in R^\alpha\Hb^{k-1}+R^\alpha\rho^\beta\Hb^{k-1},
  \]
  and $\Gamma(g)_{\mu\nu}^\kappa\in\rho^3\CI+R^\alpha\Hb^{k-1}+R^\alpha\rho^\beta\Hb^{k-1}$. This implies the same memberships but with $k-2$ in place of $k-1$ for the components of $\Riem(g)-\Riem(g_0)$ (here $\Riem(g)^\lambda{}_{\kappa\mu\nu}=e_\mu\Gamma_{\nu\kappa}^\lambda-e_\nu\Gamma_{\mu\kappa}^\lambda+\Gamma_{\mu\rho}^\lambda\Gamma_{\nu\kappa}^\rho-\Gamma_{\nu\rho}^\lambda\Gamma_{\mu\kappa}^\rho$), $\Riem(g)-\Riem(g_\sfb)$, and $\Riem(g)$ itself. By expressing $\Box_g$, $\delta_g$, $\sfG_g$, $\tilde\delta_g^*$, $\sR_g$, $\sE_{g;g_0}$ in the frame $e^\mu$, and noting that for $\eta:=\Ups(g;g_0)+E_{g_0}\tilde h\in R^\alpha\rho^\beta\Hb^{k-1}(\Omega_{\rho_0,R_0})$ we have
  \[
    ((D_g\tilde\delta^*_\cdot)\eta)_{\mu\nu} \colon h \mapsto -\half ( h_\mu{}^\kappa{}_{;\nu} + h_\nu{}^\kappa{}_{;\mu} - h_{\mu \nu}{}^{;\kappa} )\eta_\kappa,
  \]
  we thus obtain~\eqref{EqE1StructL}. The tame bounds~\eqref{EqE1StructTame} follow easily from Lemma~\ref{Lemma0bSob}.

  \pfstep{Indicial family.} Since $e_i\in\tau\Vb(\breve M)$, we only need to keep track of the $e_0$-derivatives (including those of order $0$) acting on the argument of $L_{h_0,\tilde h}$, whereas all $e_i$-derivatives can be dropped. Moreover, all contributions to $L_{h_0,\tilde h}$ arising from $\tilde h$ are (a fortiori) of class $R^\alpha\rho^\beta\Hb^{k-2}\Diffb^2$ and thus do not contribute to the indicial operator either; thus, we only need to compute the indicial operator of $L_{h_0,0}$, which amounts to working with $g=g_0$; the terms involving $\sE_{g;g_0}$ and $D_g\tilde\delta^*_\cdot$ in~\eqref{EqE1LExpr} thus vanish. Now, $\Gamma(g)_{\lambda\mu\nu}\equiv 0\bmod\rho^3\CI+R^\alpha\rho^\beta\Hb^{k-1}$ for all $\lambda,\mu,\nu$ except for
  \[
    \Gamma(g)_{\ell i 0} \equiv -g_{i\ell},\qquad
    \Gamma(g)_{0 i j} \equiv g_{i j},
  \]
  and therefore also $\Gamma(g)^\lambda_{\mu\nu}\equiv 0$ except for
  \begin{equation}
  \label{EqE1StructGamma}
    \Gamma(g)^\ell_{i 0} \equiv -\delta_i^\ell,\qquad
    \Gamma(g)^0_{i j} \equiv -g_{i j}.
  \end{equation}
  (Carefully note that $\Gamma(g)_{0 i}^\ell\equiv 0$. The connection coefficients in the frame $e_\mu$ are not symmetric.) This gives $R_{\kappa\mu\nu\lambda}=g(e_\kappa,([\nabla_{e_\nu},\nabla_{e_\lambda}]-\nabla_{[e_\nu,e_\lambda]})e_\mu)\equiv 0\bmod\rho^3\CI+R^\alpha\rho^\beta\Hb^{k-2}$ except for $R_{0 m 0 \ell} \equiv -g_{m\ell}$, $R_{k m n \ell} \equiv g_{k n}g_{m\ell} - g_{k\ell}g_{m n}$, and those coefficients obtained from these via the symmetries $R_{\kappa\mu\nu\lambda}=-R_{\mu\kappa\nu\lambda}=-R_{\kappa\mu\lambda\nu}$. Therefore, $\Ric(g)_{0 0}\equiv -3$, $\Ric(g)_{m\ell}\equiv g_{m\ell}$, $\Ric(g)_{0\ell}\equiv 0$; that is, $\Ric(g)\equiv 3 g$. From this, one easily computes that, in the bundle splitting~\eqref{Eq0bST0Split},
  \begin{equation}
  \label{EqE1StructR}
    \sR_g \equiv \begin{pmatrix} 3 I & 0 & \tr_{g_{(0)}} \\ 0 & 4 I & 0 \\ g_{(0)} & 0 & 4 I - g_{(0)}\tr_{g_{(0)}} \end{pmatrix}.
  \end{equation}
  (Cf.\ \cite[Lemma~2.4]{HintzGluedS}.)

  Next, we compute the indicial operator of $\Box_g$. If $u$ is a section of $\upbeta^*(S^2\,{}^0 T^*M)$, then modulo operators acting on $u$ which do not contribute to the indicial operator, we have $\Box_g u\equiv u_{\mu\nu;0 0}-g^{k\ell}u_{\mu\nu;k\ell}$. Using~\eqref{EqE1StructGamma} and $u_{\mu\nu;\kappa}=e_k u_{\mu\nu}-\Gamma_{\kappa\mu}^\rho u_{\rho\nu}-\Gamma_{\kappa\nu}^\rho u_{\mu\rho}$, we then compute
  \begin{alignat*}{2}
    u_{0 0;0}&\equiv e_0 u_{0 0}, &\qquad
    u_{0 0;k}&\equiv 2 u_{0 k}, \\
    u_{0 j;0}&\equiv e_0 u_{0 j}, &\qquad
    u_{0 j;k}&\equiv u_{j k}+g_{j k}u_{0 0}, \\
    u_{i j;0}&\equiv e_0 u_{i j}, &\qquad
    u_{i j;k}&\equiv u_{0 j}g_{i k}+u_{i 0}g_{j k},
  \end{alignat*}
  and then
  \begin{alignat*}{2}
    u_{0 0;0 0}&\equiv e_0 e_0 u_{0 0}, &\qquad
    u_{0 0;k\ell}&\equiv 2 u_{0\ell;k}+g_{k\ell}u_{0 0;0}, \\
    u_{0 j;0 0}&\equiv e_0 e_0 u_{0 j}, &\qquad
    u_{0 j;k\ell}&\equiv u_{\ell j;k}+u_{0 0;k}g_{j\ell}+u_{0 j;0}g_{k\ell}, \\
    u_{i j;0 0}&\equiv e_0 e_0 u_{i j}, &\qquad
    u_{i j;k\ell}&\equiv g_{i\ell}u_{0 j;k} + g_{j\ell}u_{i 0;k} + g_{k\ell}u_{i j;0}.
  \end{alignat*}
  In the splittings~\eqref{Eq0bT0Split} and \eqref{Eq0bST0Split}, this gives
  \begin{equation}
  \label{EqE1StructBoxg}
    \Box_g \equiv
      e_0 e_0-3 e_0 + \begin{pmatrix}
        -6 & 0 & -2\tr_{g_{(0)}} \\ 0 & -6 & 0 \\ -2 g_{(0)} & 0 & -2
      \end{pmatrix}, \qquad
    \delta_g \equiv
      \begin{pmatrix}
        e_0-3 & 0 & -\tr_{g_{(0)}} \\ 0 & e_0-4 & 0
      \end{pmatrix}.
  \end{equation}
  The modification terms $E_{g_0}$ and $\tilde E$ in~\eqref{EqEGauge} and \eqref{EqECD} are given by
  \begin{equation}
  \label{EqE1StructEE}
    E_{g_0} \equiv \begin{pmatrix} 1 & 0 & -2\tr_{g_{(0)}} \\ 0 & 0 & 0 \end{pmatrix},\qquad
    \tilde E \equiv \begin{pmatrix} -2 & 0 \\ 0 & -2 \\ 0 & 0 \end{pmatrix}.
  \end{equation}
  Similarly, if $\omega$ is a 1-form, we compute the covariant derivatives $\omega_{\mu;\nu}=e_\nu\omega_\mu-\Gamma_{\nu\mu}^\kappa\omega_\kappa$ to be $\omega_{0;0}\equiv e_0\omega_0$, $\omega_{0;j}\equiv \omega_j$, $\omega_{i;0}\equiv e_0\omega_i$, $\omega_{i;j}\equiv g_{i j}\omega_0$, and therefore
  \begin{equation}
  \label{EqE1Structdelstar}
    \delta_g^* \equiv \begin{pmatrix} e_0 & 0 \\ 0 & \half(e_0+1) \\ g_{(0)} & 0 \end{pmatrix},\qquad
    \tilde\delta_g^* \equiv \begin{pmatrix} e_0-2 & 0 \\ 0 & \half(e_0-3) \\ g_{(0)} & 0 \end{pmatrix}.
  \end{equation}
  With the indicial operator of $L_{h_0,\tilde h}$ being equal to that of $\Box_g-2\Lambda+2\tilde E\delta_g\sfG_g+2\sR_g-2\tilde\delta_g^*\circ E$, we thus obtain~\eqref{EqE1StructInd}.
\end{proof}

As a by-product of the computations in the above proof, we record:

\begin{lemma}[Mapping properties of $P$]
\label{LemmaE1Map}
  For $h_0,\tilde h,\theta$ as in~\eqref{EqESolh}, with $\|h_0\|_{R^\alpha\Hb^{d_3+2}}<\delta_0$ and $\|\tilde h\|_{R^\alpha\rho^\beta\Hb^{d_4+2}}<\delta_0$ for some small $\delta_0>0$, we have
  \[
    P(h_0,\tilde h,\theta)\in R^\alpha\rho^\beta\Hb^\infty(\Omega_{\rho_0,R_0};\upbeta^*(S^2\,{}^0 T^*M)).
  \]
  Moreover, for all $k\in\N_0$, we have the tame estimate
  \[
    \|P(h_0,\tilde h,\theta)\|_{R^\alpha\rho^\beta\Hb^k} \leq C_k\bigl( \|h_0\|_{R^\alpha\Hb^{k+2}} + \|\tilde h\|_{R^\alpha\rho^\beta\Hb^{k+2}} + \|\theta\|_{R^\alpha\rho^\beta\Hb^{k+1}} \bigr).
  \]
\end{lemma}

\subsubsection{Indicial roots of the constraint propagation and gauge potential wave operators}
\label{SssE1Aux}

Before continuing the study of $L_{h_0,\tilde h}$, we make the following observations regarding the linearization $L_\dS$ of $g\mapsto 2(\Ric(g)-\Lambda g-\tilde\delta^*_g(\Ups(g;g_0)+E_{g_0}(g-g_0))$ around $g=g_0=g_\dS$. (This is the linear operator one would need to invert when using a Newton type iteration to study the stability of de~Sitter space.) These observations are only made to motivate the choices of $E_{g_0},\tilde E$ in~\eqref{EqEGauge} and \eqref{EqECD}. To wit, if $L_\dS h=0$, then by the linearized second Bianchi identity, we have $2\delta_g\sfG_g\tilde\delta^*_g\eta=0$ where $\eta=D_1|_g\Ups(h;g_0)+E_{g_0}h$. Refining the splitting~\eqref{Eq0bST0Split} using
\begin{equation}
\label{EqE1AuxSplitRef}
  S^2 T^*X = \R g_{(0)} \oplus \ker\tr_{g_{(0)}},
\end{equation}
we can use~\eqref{EqE1StructGg} and~\eqref{EqE1StructBoxg}--\eqref{EqE1Structdelstar} to compute the indicial family of $2\delta_g\sfG_g\tilde\delta^*_g$ as
\begin{equation}
\label{EqE1AuxCD}
\begin{split}
  &\begin{pmatrix} \lambda-3 & 0 & -3 & 0 \\ 0 & \lambda-4 & 0 & 0 \end{pmatrix}\begin{pmatrix} \tfrac12 & 0 & \tfrac32 & 0 \\ 0 & 1 & 0 & 0 \\ \half & 0 & -\half & 0 \\ 0 & 0 & 0 & 1 \end{pmatrix} \left( \begin{pmatrix} \lambda & 0 \\ 0 & \half(\lambda+1) \\ 1 & 0 \\ 0 & 0 \end{pmatrix} + \begin{pmatrix} -2 & 0 \\ 0 & -2 \\ 0 & 0 \\ 0 & 0 \end{pmatrix} \right) \\
  &\qquad\qquad = \begin{pmatrix} (\lambda-2)(\lambda-3) & 0 \\ 0 & (\lambda-3)(\lambda-4) \end{pmatrix}.
\end{split}
\end{equation}
Its indicial roots are thus $\lambda=2,3,4$ and in particular all positive. Therefore, for any indicial solution of $I(L_\dS,\lambda)h=0$ with $\Re\lambda\leq 0$ (or more generally $\Re\lambda<2$), i.e.\ $L_\dS(\tau^\lambda h)=\cO(\tau^{\lambda+1})$, the gauge 1-form $\eta$ defined above necessarily vanishes modulo $\cO(\tau^{\lambda+1})$. Therefore, $h$ is an indicial solution also for the linearization of the ungauged operator $\Ric(g)-\Lambda g$ around $g=g_\dS$, and in particular satisfies an indicial operator version of the linearized constraints. (Moreover, it satisfies the linearized gauge condition on the indicial operator level.) Thus, the particular choice of $\tilde E$ in~\eqref{EqECD} leads to a damping of violations of the constraints, and in this sense acts as \emph{constraint damping}.

\begin{rmk}[Origin of the choice of $\tilde E$]
\label{RmkE1ChoicetildeE}
  The choice~\eqref{EqECD} corresponds exactly to the damping terms $M_{\mu\nu}$ in \cite[(51)--(53)]{RingstromEinsteinScalarStability}. Another choice is the one made in \cite[Appendix~C.3--C.4]{HintzVasyKdSStability}, which amounts to $\tilde E\omega=-2 e^0\otimes_s\omega+\omega(e_0)g_0$ (which now depends on $g_0$), the indicial operator of which is
  \[
    \begin{pmatrix} -1 & 0 \\ 0 & -1 \\ -1 & 0 \\ 0 & 0 \end{pmatrix}.
  \]
  The indicial roots of $2\delta_g\sfG_g\tilde\delta^*_g$ for this choice would be $1,4,6$. --- There is of course an open set of choices with the same damping effect, and for all such choices our arguments below go through (except for possibly having to reduce $\beta>0$ if an indicial root $\lambda$ in the right half plane gets close to $\{\Re\lambda=0\}$).
\end{rmk}

Next, we consider the linearization $-(\delta_g\sfG_g-E_{g_0})$ of the gauge 1-form $g\mapsto\Ups(g;g_0)+E_{g_0}(g-g_0)$ around $g=g_0=g_\dS$ (cf.\ \eqref{EqE1D1Ups}), and specifically ask about the indicial solutions which are pure gauge. By this we mean indicial solutions of the form $I(\delta_g^*,\lambda)\omega$ where $\omega$ is an indicial solution of the \emph{gauge potential wave operator} $2(\delta_g\sfG_g-E_{g_0})\delta_g^*$---whose indicial family, using~\eqref{EqE1StructEE}, is
\begin{equation}
\label{EqE1ChoiceBoxGauge}
\begin{split}
  &\left(\begin{pmatrix} \lambda-3 & 0 & -3 & 0 \\ 0 & \lambda-4 & 0 & 0 \end{pmatrix}\begin{pmatrix} \tfrac12 & 0 & \tfrac32 & 0 \\ 0 & 1 & 0 & 0 \\ \half & 0 & -\half & 0 \\ 0 & 0 & 0 & 1 \end{pmatrix} - \begin{pmatrix} 1 & 0 & -6 & 0 \\ 0 & 0 & 0 & 0 \end{pmatrix} \right) \begin{pmatrix} \lambda & 0 \\ 0 & \half(\lambda+1) \\ 1 & 0 \\ 0 & 0 \end{pmatrix} \\
  &\qquad\qquad = \begin{pmatrix} (\lambda-2)(\lambda-3) & 0 \\ 0 & (\lambda-4)(\lambda+1) \end{pmatrix}.
\end{split}
\end{equation}
The indicial roots are thus $-1,2,3,4$, and the indicial root $-1$ corresponds, on exact de~Sitter space, to the fact that $\pa_{x^i}=\tau^{-1}e_i$ is a Killing vector field. (More generally $\delta_g^*(\tau^{-1}e_i)=o(\tau^{-1})$ for general $g$ of the form studied in Proposition~\ref{PropE1Struct}, as follows from~\eqref{EqE1Structdelstar}.)

\begin{rmk}[Origin of the choice of $E_{g_0}$]
\label{RmkE1GaugeModOrigin}
  We found the modification $E_{g_0}$ of the (generalized) harmonic gauge by trial and error. The gauge modification used in \cite[(50)]{RingstromEinsteinScalarStability} corresponds to
  \begin{equation}
  \label{EqE1GaugeModRingstrom}
    E_{g_0}^{\text{\cite{RingstromEinsteinScalarStability}}} = \begin{pmatrix} 0 & 0 & -3 & 0 \\ 0 & -2 & 0 & 0 \end{pmatrix}.
  \end{equation}
  This leads to the gauge potential wave operator having an indicial root at $0$ (see Remark~\ref{RmkEIndRingstrom} for the consequence of this for the gauge-fixed linearized Einstein equation); our choice avoids this. If one took $E_{g_0}=0$, then there would, for example, be an indicial root at $\frac12(3-\sqrt{33})\in(-2,-1)$, corresponding to an exponentially growing pure gauge solution which would need to be removed from the asymptotics of the linearized metric perturbation by a gauge modification, as done in static patches in \cite[Appendix~C]{HintzVasyKdSStability}, as discussed in Remark~\ref{RmkIdSHarmGauge}.
\end{rmk}

\subsubsection{Indicial family of the linearized gauge-fixed Einstein operator}
\label{SssE1Ind}

We now compute the indicial roots of the linearized gauge-fixed Einstein operator~\eqref{EqE1LExpr}.

\begin{lemma}[Indicial roots of $L_{h_0,\tilde h}$]
\label{LemmaEIndRoot}
  Let $h_0,\tilde h$ be as in Proposition~\usref{PropE1Struct}, define $g_{(0)},h_{(0)}$ by~\eqref{EqE1Structg0h0}, and write $I_{g_{(0)}}(\lambda)=I(L_{h_0,\tilde h},\lambda)$. Then the indicial roots of $L_{h_0,\tilde h}$ are $0,2,3,4$. The space of indicial solutions corresponding to the root $0$ is $\ker I_{g_{(0)}}(0)=\tau^{-2}\ker\tr_{g_{(0)}}$ (as a subbundle of $\upbeta^*(S^2\,{}^0 T^*M)$ over $\cI^+_{R_0}$); and the root $0$ is simple in that $I_{g_{(0)}}(\lambda)^{-1}$ has a simple pole at $\lambda=0$.
\end{lemma}
\begin{proof}
  We split
  \begin{equation}
  \label{EqEIndRootSplit}
    S^2 T^*X=\R g_{(0)}\oplus\ker\tr_{g_{(0)}}
  \end{equation}
  (so $\tau^{-2}S^2 T^*X=\R\tau^{-2}g_{(0)}\oplus\tau^{-2}\ker\tr_{g_{(0)}}\subset\upbeta^*(S^2\,{}^0 T^*M)$). From~\eqref{EqE1StructIndOp}, we then have
  \begin{equation}
  \label{EqEIndRootPf}
    I_{g_{(0)}}(\lambda) = \lambda^2-3\lambda+2\cdot\begin{pmatrix} -2\lambda+5 & 0 & 3(\lambda-2) & 0 \\ 0 & -2\lambda+6 & 0 & 0 \\ -1 & 0 & 3 & 0 \\ 0 & 0 & 0 & 0 \end{pmatrix}.
  \end{equation}
  The solutions of $\det I_{g_{(0)}}(\lambda)=0$ are $\lambda=0,2,3,4$. Moreover, $\ker I_{g_{(0)}}(0)$ is spanned by $\ker\tr_{g_{(0)}}$. The final statement is a consequence of the fact that the determinant of the $4\times 4$ matrix given by the right hand side of~\eqref{EqEIndRootPf} has a simple zero at $\lambda=0$.
\end{proof}

\begin{rmk}[Comparison with Ringstr\"om's operator]
\label{RmkEIndRingstrom}
  Using the gauge modification~\eqref{EqE1GaugeModRingstrom} from \cite{RingstromEinsteinScalarStability}, one finds that $I_{g_{(0)}}(\lambda)$ is diagonal and given by $\diag((\lambda-2)(\lambda-3),(\lambda-2)(\lambda-3),\lambda(\lambda-3),\lambda(\lambda-3))$. Thus, the space of indicial solutions at $\lambda=0$ now consists of all tangential-tangential tensors.
\end{rmk}

\subsection{Linearized gauge-fixed Einstein operator II: estimates for solutions}
\label{SsE2}

The control of solutions of initial value problems for $L_{h_0,\tilde h}$ lies at the heart of our stability proof. We proceed in two steps.
\begin{enumerate}
\item First, we obtain an estimate on a space allowing for growth at $\cI^+$ but with arbitrary regularity (Proposition~\ref{PropE2Reg}).
\item We then use the information about the indicial operator of $L_{h_0,\tilde h}$ from Lemma~\ref{LemmaEIndRoot} to improve decay while giving up regularity (Proposition~\ref{PropE2Decay}).
\end{enumerate}

For definiteness, we fix $\sV\subset\cV(\Sph^2)$ to be the set $\sV=\{V_1,V_2,V_3\}$ where $V_a$ is the rotation vector field around the $a$-th coordinate axis in $\R^3\supset\Sph^2$. \emph{Throughout this section, we assume that $\alpha>0$, $\beta\in(0,1)$, and
\begin{equation}
\label{EqE2h}
  \|h_0\|_{R^\alpha\Hb^{2 d_3+4}(\cI^+_{R_0};\upbeta^*(\tau^{-2}S^2 T^*X))} < \delta_0,\qquad
  \|\tilde h\|_{R^\alpha\rho^\beta\Hb^{2 d_4+4}(\Omega_{\rho_0,R_0};\upbeta^*(S^2\,{}^0 T^*M))} < \delta_0,
\end{equation}
and $\delta_0>0$ is small.} Under these assumptions, we have $\Hb^{d_4}$- and thus $\cC^0$-bounds on the coefficients of $L_{h_0,\tilde h}$ by~\eqref{EqE1StructTame}.

\subsubsection{High regularity estimate on growing spaces}
\label{SssE2Reg}

Following the strategy outlined in~\S\ref{SsIEx}, we prove:

\begin{prop}[Tame bounds on growing spaces]
\label{PropE2Reg}
  There exists $N>0$ so that the following holds. Let
  \[
    v_0,v_1\in R^\alpha\Hb^\infty(\Sigma_{\rho_0,R_0};\upbeta^*(S^2\,{}^0 T^*M)),\qquad f\in R^\alpha\rho^{-N}\Hb^\infty(\Omega_{\rho_0,R_0};\upbeta^*(S^2\,{}^0 T^*M)),
  \]
  and define the norm
  \[
    \|(f,v_0,v_1)\|_{D^{k,\alpha,\beta}} := \|f\|_{R^\alpha\rho^\beta\Hb^k} + \|v_0\|_{R^\alpha\Hb^{k+1}} + \|v_1\|_{R^\alpha\Hb^k}.
  \]
  Then the initial value problem
  \begin{equation}
  \label{EqE2RegPDE}
    L_{h_0,\tilde h}v=f,\qquad (v,\cL_{-\rho\pa_\rho}v)|_{\Sigma_{\rho_0,R_0}} = (v_0,v_1),
  \end{equation}
  has a unique solution $v\in R^\alpha\rho^{-N}\Hb^\infty(\Omega_{\rho_0,R_0};\upbeta^*(S^2\,{}^0 T^*M))$ which, moreover, for all $k\in\N_0$ satisfies the tame estimate
  \begin{equation}
  \label{EqE2RegEst}
  \begin{split}
    \|v\|_{R^\alpha\rho^{-N}\Hb^k} &\leq C_k\Bigl( \|(f,v_0,v_1)\|_{D^{k,\alpha,-N}} \\
      &\quad \hspace{3em} + \bigl(\|h_0\|_{R^\alpha\Hb^{k+2 d_3+2}} + \|\tilde h\|_{R^\alpha\rho^\beta\Hb^{k+2 d_4+2}}\bigr) \|(f,v_0,v_1)\|_{D^{0,\alpha,-N}} \Bigr).
  \end{split}
  \end{equation}
\end{prop}
\begin{proof}
  We first prove an energy estimate on the level of $H^1$ (or more precisely weighted $H_{0,\bop}^1$). Higher regularity follows by commuting b-vector fields through the equation~\eqref{EqE2RegPDE} and using an $H^1$-level energy estimate for a \emph{system} of wave equations (each of which involves $L_{h_0,\tilde h}$). To facilitate this second step, we immediately phrase the basic energy estimate for such systems.

  \pfstep{Step 1. Basic energy estimate.} Let $K\in\N$, and suppose that $A_{I J}$, $1\leq I,J\leq K$, is a first order differential operator acting on sections of $\upbeta^*(S^2\,{}^0 T^*M)$ over $\Omega_{\rho_0,R_0}$ which is of the form $A_{I J} = A_{0,I J} + A_{(0),I J} + \tilde A_{I J}$ where
  \[
    A_{0,I J} \in \Diff_{0,\bop}^1,\quad
    A_{(0),I J} \in R^\alpha\cC_\bop^0(\cI^+_{R_0})\Diff_{0,\bop}^1,\quad
    \tilde A_{I J} \in R^\alpha\rho^\beta\cC_\bop^0(\Omega_{\rho_0,R_0})\Diff_{0,\bop}^1.
  \]
  Assume moreover that $A=(A_{I J})_{1\leq I,J\leq K}$ has an lower triangular structure at $\cI^+$ in that\footnote{In the actual setting arising in Step~2 below, the indexing is by b-derivatives applied to $v$ solving~\eqref{EqE2RegPDE}. The lower triangular structure will arise from the inclusion of $\cV_{0,\bop}\hra\Vb$, with the treatment of elements of $\Vb$ which are in the range of this inclusion (i.e.\ $\rho\pa_\rho$) different from those which are not $(R\pa_R$, spherical derivatives). Cf.\ also the discussion of higher regularity in~\S\ref{SsIEx}.}
  \[
    I\leq J \implies A_{0,I J} \in \rho\Diff_{0,\bop}^1,\quad A_{(0),I J}=0.
  \]
  For the system
  \begin{equation}
  \label{EqE2LOp}
    \cL := (\delta_{I J}L_{h_0,\tilde h} + A_{I J})_{1\leq I,J\leq K},
  \end{equation}
  we then consider the initial value problem
  \begin{equation}
  \label{EqE2IVP}
    \cL v = f,\qquad (v,\cL_{-\rho\pa_\rho}v)|_{\Sigma_{\rho_0,R_0}}=(v_0,v_1),
  \end{equation}
  where now $f$ and $v_0,v_1$ are $K$-tuples of elements of $R^\alpha\rho^{-N}\Hb^\infty$ and $R^\alpha\Hb^\infty$, respectively. We claim that there exists a constant $N$ \emph{which is independent of $h_0$, $\tilde h$, $K$, $A_{I J}$} so that
  \begin{equation}
  \label{EqE2RegBasicEst}
    \|v\|_{R^\alpha\rho^{-N}H_{0,\bop}^1} \leq C\|(f,v_0,v_1)\|_{D^{0,\alpha,-N}}.
  \end{equation}
  (The constant $C$ is allowed to depend on $K,A_{I J}$.)

  \pfsubstep{(1.1)}{Energy estimate for the scalar wave equation.} Note that $\frac{\dd\rho}{\rho}$ is timelike with respect to $g=g_\sfb+\chi h_0+\tilde h$ on $\Omega_{\rho_0,R_0}$; this follows from the corresponding property for $g_\sfb$ in~\eqref{Eq0bStructTime} since, by Sobolev embedding, $h_0,\tilde h$ are small in $L^\infty$ as sections of $\upbeta^*(S^2\,{}^0 T^*M)$ (which implies that $(\chi h_0+\tilde h)(\frac{\dd\rho}{\rho},\frac{\dd\rho}{\dd\rho})$ is small in $L^\infty$). This timelike nature implies that a weight $\rho^N$ can be used to give a positive bulk term in an energy estimate. In order to ensure that the value of $N$ required to get this positivity does not depend on the terms $A_{0,I J}$ and $\tilde A_{I J}$ of $\cL$ (which are lower order not just in the differential order, but importantly also in the sense of decay at $\rho=0$), we employ an additional weight $e^{\digamma\rho^{2\beta}/2\beta}$. For any fixed value of $\digamma$, this additional weight is smooth and bounded away from $0$ and $\infty$, but choosing $\digamma\gg 1$ allows us to give less weight to energy densities close to $\rho=0$ than near $\rho=\rho_0$, and thus gain some additional positivity for the bulk term away from $\cI^+$ (see~\eqref{EqE2rhodomega} below); moreover, this weight allows us to absorb lower order terms (in the sense of decay) to wave operators (see e.g.\ the discussion of \eqref{EqE2Lot} below). Let thus
  \[
    V_0 := -R^{-2\alpha}\rho\pa_\rho,\qquad V := \omega^2 V_0,\ \ \omega=\omega(\rho):=e^{\digamma\rho^{2\beta}/2\beta}\rho^N,
  \]
  acting component-wise both in the index $I$ and in the trivialization of $\upbeta^*(S^2\,{}^0 T^*M)$ induced by $e^0=\frac{\dd\tau}{\tau}$, $e^i=\frac{\dd x^i}{\tau}$ ($i=1,2,3$).

  We first prove an estimate for the scalar wave operator $\Box_g$; so consider $\Box_g v=f$, with initial data $(v_0,v_1)$ for $v$. The stress-energy-momentum tensor $T=T[v]$ of $v$ is $T_{\mu\nu}=(e_\mu v)(e_\nu v)-\frac12 g_{\mu\nu}g^{\kappa\lambda}(e_\kappa v)(e_\lambda v)$ where we recall $e_0=\tau\pa_\tau$, $e_i=\tau\pa_{x^i}$. The $J$-current associated with $v$ and $V$ is
  \[
    {}^{(V)}J = T(V,\cdot),\qquad {\rm div}_g({}^{(V)}J) = -(\Box_g v)V v + {}^{(V)}K,\quad {}^{(V)}K = T\cdot\cL_V g,
  \]
  with `$\cdot$' denoting tensor contraction (using $g$). For $\rho_1\in(0,\rho_0)$, define the domain
  \[
    \Omega_{\rho_0,R_0}^{\rho_1} := \Omega_{\rho_0,R_0} \cap \{\rho\geq\rho_1\},
  \]
  with boundary hypersurfaces $\Sigma_{\rho_0,R_0}=\{\rho=\rho_0\}$ and $\Sigma_{\rho_0,R_0}^{\rho_1}=\{\rho=\rho_1\}\cup\{(2-R)\rho=(2-R_0)\rho_0\}$. Then
  \begin{align*}
    &\int_{\Sigma_{\rho_0,R_0}^{\rho_1}} \la{}^{(V)}J,\nu\ra\,\dd\sigma + \int_{\Omega_{\rho_0,R_0}^{\rho_1}} {}^{(V)}K\,\dd g \\
    &\hspace{8em} = \int_{\Sigma_{\rho_0,R_0}} \la {}^{(V)}J,\nu\ra\,\dd\sigma + \int_{\Omega_{\rho_0,R_0}^{\rho_1}} (\Box_g v)V v\,\dd g
  \end{align*}
  where $\nu$ denotes the future unit normal at the respective boundary hypersurface. Since $R^{2\alpha}V_0$ is future timelike uniformly down to $R=0$, the integral over $\Sigma_{\rho_0,R_0}$ is bounded in absolute value by $\|v_0\|_{R^\alpha\Hb^1}^2+\|v_1\|_{R^\alpha L^2}^2$. The integral over $\Sigma_{\rho_0,R_0}^{\rho_1}$ is non-negative, and will be dropped in the estimate below. Now,
  \[
    {}^{(V)}K = {}^{(\omega^2 V_0)}K = \omega^2\cdot{}^{(V_0)}K + 2\omega T(\nabla\omega,V_0),
  \]
  where we further compute
  \[
    T(\nabla\omega,V_0) = R^{-2\alpha}\rho\omega'(\rho) T\Bigl(\Bigl(\frac{\dd\rho}{\rho}\Bigr)^\sharp, -\rho\pa_\rho\Bigr).
  \]
  Since $\frac{\dd\rho}{\rho}$ and $-\rho\pa_\rho$ are are uniformly (future) timelike in $\Omega_{\rho_0,R_0}$, there exists a constant $c_0>0$ so that $T((\frac{\dd\rho}{\rho})^\sharp,-\rho\pa_\rho)\geq c_0\sum_{\mu=0}^3 (e_\mu v)^2$. Using the simple upper bound $|{}^{(V_0)}K|\leq C_0 R^{-2\alpha}\sum_{\mu=0}^3 (e_\mu v)^2$, we thus obtain (upon taking $\rho_1\to 0$) the energy estimate
  \begin{align*}
    &\int_{\Omega_{\rho_0,R_0}} R^{-2\alpha} ( c_0\rho\omega' - C_0\omega ) \omega \sum_{\mu=0}^3 (e_\mu v)^2\,\dd g \\
    &\hspace{8em} \leq C\bigl( \|v_0\|_{R^\alpha\Hb^1}^2 + \|v_1\|_{R^\alpha L^2}^2 \bigr) + \int_{\Omega_{\rho_0,R_0}} R^{-2\alpha}\omega^2 |\Box_g v| |\rho\pa_\rho v|\,\dd g
  \end{align*}
  (The choice $\omega=\rho^N$ where $N>\frac{2 C_0}{c_0}$ would ensure that the term in parentheses on the left is bounded from below by $C_0\omega$, and thus one immediately gets an $\dot H^1$ type estimate.) We compute
  \begin{equation}
  \label{EqE2rhodomega}
    \rho\omega' = (N + \digamma\rho^{2\beta})\omega.
  \end{equation}

  We control $v$ itself by integrating (i.e.\ via a version of the Hardy inequality which is uniform in $\digamma$); to wit, for $W=\omega^2 W_0$ where $W_0=R^{-2\alpha}\rho\pa_\rho$, we have
  \begin{align*}
    {\rm div}_g W&=\omega^2\,{\rm div}_g W_0+2\omega g(\nabla\omega,W_0) \\
      &=\omega^2\,{\rm div}_g W_0 + 2\omega \rho\omega'(\rho) R^{-2\alpha} = \omega^2\bigl({\rm div}_g W_0 + 2(N+\digamma\rho^{2\beta}) R^{-2\alpha}\bigr).
  \end{align*}
  In the identity
  \[
    \int_{\Sigma_{\rho_0,R_0}^{\rho_1}} \la v^2 W,\nu\ra\,\dd\sigma + \int_{\Omega_{\rho_0,R_0}^{\rho_1}} {\rm div}_g(v^2 W)\,\dd g = \int_{\Sigma_{\rho_0,R_0}} \la v^2 W,\nu\ra\,\dd\sigma,
  \]
  we then note that $\la W,\nu\ra$ is bounded from above and below by a positive multiple of $R^{-2\alpha}\rho^{2 N}$ (and it is positive on $\Sigma_{\rho_0,R_0}^{\rho_1}$), while ${\rm div}_g W\geq (\frac32 N+\digamma\rho^{2\beta})R^{-2\alpha}\omega^2$ for sufficiently large $N$ (independently of $\digamma$); this uses that $|{\rm div}_g W_0|\lesssim R^{-2\alpha}$. Since ${\rm div}_g(v^2 W)=v^2\,{\rm div}_g W+2 v W v$, this gives (for a larger constant $C_0$)
  \begin{equation}
  \label{EqE2ScalarEst}
  \begin{split}
    &\int_{\Omega_{\rho_0,R_0}} R^{-2\alpha}\biggl( (N+\digamma\rho^{2\beta})\omega^2 v^2 + (c_0\rho\omega'-C_0\omega)\omega \sum_{\mu=0}^3 |e_\mu v|^2\biggr)\,\dd g \\
    &\hspace{8em} \leq C_{N,\digamma}\bigl( \|v_0\|_{R^\alpha\Hb^1}^2 + \|v_1\|_{R^\alpha L^2}^2 \bigr) + \int_{\Omega_{\rho_0,R_0}} R^{-2\alpha}\omega^2 |\Box_g v| |\rho\pa_\rho v|\,\dd g.
  \end{split}
  \end{equation}
  We leave this estimate as it is, but point out that choosing $N$ large enough and applying Cauchy--Schwartz to the final term, one would obtain an estimate of the form $\|v\|_{R^\alpha\rho^{-N}H_{0,\bop}^1(\Omega_{\rho_0,R_0})}\leq C\|\Box_g v\|_{R^\alpha\rho^{-N}L^2(\Omega_{\rho_0,R_0})}$.

  \pfsubstep{(1.2)}{The case that $A_{0,I J}\in\rho\Diff_{0,\bop}^1$, $A_{(0),I J}=0$ for all $I,J$.} Write
  \begin{align*}
    &L_{h_0,\tilde h} = \Box_g \otimes \Id_{\R^{10}} + Q,\qquad Q = Q_0 + Q_{(0)} + \tilde Q, \\
    &\qquad Q_0 \in \Diff_{0,\bop}^1,\quad Q_{(0)} \in R^\alpha\cC^0(\cI^+_{R_0})\Diff_{0,\bop}^1,\quad \tilde Q\in R^\alpha\rho^\beta\cC^0(\Omega_{\rho_0,R_0})\Diff_{0,\bop}^1,
  \end{align*}
  with $10$ being the rank of $\upbeta^*(S^2\,{}^0 T^*M)$; here $Q_0$ is independent of $h_0,\tilde h$, while $Q_{(0)}$ and $\tilde Q$ have small coefficients in view of~\eqref{EqE2h}. Applying~\eqref{EqE2ScalarEst} to a $10$-component vector $v$, the replacement of $\Box_g v$ on the right by $L_{h_0,\tilde h}$ creates error terms: the term arising from $Q_0$ is bounded by $C_{Q_0}\int_{\Omega_{\rho_0,R_0}} R^{-2\alpha}\omega^2(v^2+\sum_{\mu=0}(e_\mu v)^2)\,\dd g$, while the terms arising from $Q_{(0)}$ and $\tilde Q$ are bounded by the same expression (in fact with a small constant in view of~\eqref{EqE2h}). We put these terms on the left hand side and estimate $\int_{\Omega_{\rho_0,R_0}} R^{-2\alpha}\omega^2|L_{h_0,\tilde h}v||e_0 v|\,\dd g$ using Cauchy--Schwartz; this way we obtain the estimate~\eqref{EqE2ScalarEst} for $L_{h_0,\tilde h}$ in place of $\Box_g$, with a new constant $C_0$ which can be taken to be independent of $h_0,\tilde h$. Using~\eqref{EqE2rhodomega}, we have
  \[
    c_0\rho\omega'-C_0\omega=(N c_0-C_0 + \digamma\rho^{2\beta})\omega,
  \]
  so fixing $N$ with, say,
  \[
    N c_0-C_0\geq 1,
  \]
  we obtain, with $|\pa^{\leq 1}v|^2 := v^2+\sum_{\mu=0}^3 (e_\mu v)^2$,
  \begin{equation}
  \label{EqE2LEst}
  \begin{split}
    &\int_{\Omega_{\rho_0,R_0}} R^{-2\alpha}\omega^2(1 + \digamma\rho^{2\beta}) |\pa^{\leq 1}v|^2 \,\dd g \\
    &\qquad \leq C\|R^{-\alpha}\omega L_{h_0,\tilde h}v\|_{L^2}^2 + C_{N,\digamma}\bigl(\|v_0\|_{R^\alpha\Hb^1}^2 + \|v_1\|_{R^\alpha L^2}^2\bigr).
  \end{split}
  \end{equation}

  The estimate~\eqref{EqE2ScalarEst} holds also for $\cL=\Box_g\otimes\Id_{\R^{10 K}}$ (with $10 K$ being the rank of the direct sum of $K$ copies of $\upbeta^*(S^2\,{}^0 T^*M)$). Under the present assumptions on the $A_{I J}$, we can write the operator~\eqref{EqE2LOp} in the form
  \begin{align*}
    &\cL = \Box_g\otimes\Id_{\R^{10 K}} + Q\otimes\Id_{\R^K} + \cA, \\
    &\qquad \cA=(A_{I J})_{1\leq I,J\leq K}\in \rho\Diff_{0,\bop}^1 + R^\alpha\rho^\beta\cC^0(\Omega_{\rho_0,R_0})\Diff_{0,\bop}^1.
  \end{align*}
  Now, the estimate~\eqref{EqE2LEst} holds, with the same constants, also for the operator $\cL$ when $\cA=0$. The contribution of $\cA$ can be estimated by
  \begin{equation}
  \label{EqE2Lot}
    C\|R^{-\alpha}\omega\cA v\|_{L^2}^2 \leq C_\cA \int_{\Omega_{\rho_0,R_0}} R^{-2\alpha}\rho^{2\beta}\omega^2 |\pa^{\leq 1}v|^2\,\dd g.
  \end{equation}
  For sufficiently large $\digamma$, this can be absorbed into the left hand side of~\eqref{EqE2LEst}.

  \pfsubstep{(1.3)}{The general lower triangular case.} Writing $v=(v^1,\ldots,v^K)$, the estimate~\eqref{EqE2LEst} can be applied to each $v^I$ separately. The equation for $v^I$ reads $L_{h_0,\tilde h}v^I=f^I-\sum_{J=1}^K A_{I J}v^J$. Splitting the sum into $\sum_{J<I} A_{I J}+\sum_{J\geq I} A_{I J}$, we thus get
  \begin{align*}
    &\int_{\Omega_{\rho_0,R_0}} R^{-2\alpha}\omega^2(1+\digamma\rho^{2\beta}) |\pa^{\leq 1}v^I|^2\,\dd g \\
    &\quad\qquad \leq C_K\biggl( \|R^{-\alpha}\omega f^I\|_{L^2}^2 + C_\cA\sum_{J<I} \int_{\Omega_{\rho_0,R_0}} R^{-2\alpha}\omega^2 |\pa^{\leq 1}v^J|^2\,\dd g \\
    &\quad\qquad \hspace{9.21em} + C_\cA\sum_{J\geq I} \int_{\Omega_{\rho_0,R_0}} R^{-2\alpha}\rho^{2\beta}\omega^2|\pa^{\leq 1}v^J|^2\,\dd g\biggr).
  \end{align*}
  Calling this estimate ($*_I$), we then consider the sum of estimates $\sum_{I=1}^K \eps^I(*_I)$ with $\eps>0$ to be determined. The left hand side of the resulting estimate controls a weighted norm of $\sum_{I=1}^K \eps^I|\pa^{\leq 1}v^I|^2$, while the second term on the right hand side is bounded by $C_K C_\cA \sum_{I=2}^K\sum_{J=1}^{I-1} \eps^I|\pa^{\leq 1}v^J|^2\leq K C_K C_\cA \eps \sum_{J=1}^K \eps^J|\pa^{\leq 1}v^J|^2$; this can be absorbed into the left hand side when $K C_K C_\cA\eps<\frac12$. Having thus fixed $\eps$, we can then argue as in Step~(1.2) in order to absorb also the last term on the right (arising from the sums over $J\geq I$) into the left hand side upon choosing $\digamma$ sufficiently large. This completes the proof of the estimate~\eqref{EqE2RegBasicEst}.

  \pfstep{Step 2. Higher regularity.} Trivializing $\upbeta^*(S^2\,{}^0 T^*M)$ using the frame $e^\mu$, we have
  \begin{align}
    &L_{h_0,\tilde h} = \sum_{j+|\gamma|\leq 2} \ell_{j\gamma} (\tau\pa_\tau)^j (\tau\pa_x)^\gamma,\qquad \ell_{j\gamma}=\ell_{0,j\gamma}+\ell_{(0),j\gamma}+\tilde\ell_{j\gamma}, \nonumber\\
  \label{EqE2LExpr}
    &\qquad \hspace{4em} \ell_{0,j\gamma}\in\CI(\Omega_{\rho_0,R_0}),\quad \ell_{(0),j\gamma}\in R^\alpha\Hb^\infty(\cI^+_{R_0}),\quad\tilde\ell_{j\gamma}\in R^\alpha\rho^\beta\Hb^\infty(\Omega_{\rho_0,R_0}),
  \end{align}
  where the coefficients are $10\times 10$ matrices, with $\ell_{(0),j\gamma}$ and $\tilde\ell_{j\gamma}$ satisfying tame estimates in terms of $h_0,\tilde h$; this is the content of~\eqref{EqE1StructTame}. It is more transparent to pass to the coordinates $R=|x|$, $\rho=\frac{\tau}{|x|}$, and $\omega=\frac{x}{|x|}$. We thus write
  \[
    L_{h_0,\tilde h} = \sum_{j+|\gamma|+i\leq 2} \ell_{j\gamma i}(\rho,R,\omega)(\rho R\pa_R)^j(\rho\sV)^\gamma(\rho\pa_\rho)^i,
  \]
  where
  \begin{equation}
  \label{EqE2ellSplit}
    \ell_{j\gamma i}=\ell_{j\gamma i}(\rho,R,\omega)=\ell_{0,j\gamma i}(\rho,R,\omega)+\ell_{(0),j\gamma i}(R,\omega)+\tilde\ell_{j\gamma i}(\rho,R,\omega)
  \end{equation}
  analogously to~\eqref{EqE2LExpr}. We study the initial value problem~\eqref{EqE2RegPDE}.

  \pfsubstep{(2.1)}{Warm-up: gaining 1 b-derivative.} Step~1 gives $\|v\|_{R^\alpha\rho^{-N}H_{0,\bop}^1}\leq C\|(f,v_0,v_1)\|_{D^{0,\alpha,-N}}$. We claim that for all $k\in\N$,
  \begin{equation}
  \label{EqE2TameLo}
  \begin{split}
    &\|v\|_{R^\alpha\rho^{-N}H_{0,\bop;\bop}^{1;k}} := \|v\|_{R^\alpha\rho^{-N}\Hb^k} + \|\tau\pa_\tau v\|_{R^\alpha\rho^{-N}\Hb^k} + \|\tau\pa_x v\|_{R^\alpha\rho^{-N}\Hb^k} \\
    &\qquad \leq C_k\bigl(\|h_0\|_{R^\alpha\Hb^{k+2+2 d_3}},\|\tilde h\|_{R^\alpha\rho^\beta\Hb^{k+2+2 d_4}}\bigr) \bigl(\|(f,v_0,v_1)\|_{D^{k,\alpha,-N}} + \|v\|_{R^\alpha\rho^{-N}H_{0,\bop;\bop}^{1;k-1}}\bigr),
  \end{split}
  \end{equation}
  which gives
  \begin{equation}
  \label{EqE2TameLoFin}
    \|v\|_{R^\alpha\rho^{-N}H_{0,\bop;\bop}^{1;k}} \leq C_k\bigl(\|h_0\|_{R^\alpha\Hb^{k+2+2 d_3}},\|\tilde h\|_{R^\alpha\rho^\beta\Hb^{k+2+2 d_4}}\bigr) \|(f,v_0,v_1)\|_{D^{k,\alpha,-N}}.
  \end{equation}
  We shall prove this for $k=1$ to illustrate the structural properties of (0,b)-differential operators and their interaction with b-regularity. (We do not use the estimate~\eqref{EqE2TameLo} later on, and thus leave the discussion of $k\geq 2$ to the interested reader. Only the tame estimate, proved in Step~2.2 below, will be used.)

  To wit, we consider the equations satisfied by b-derivatives of $v$. The derivatives along $\rho\pa_\rho$ and $R\pa_R$, $V_a$ ($a=1,2,3$) play different roles (cf.\ Remark~\ref{RmkIExTriangle}, with $\tau\pa_\tau$ and $\pa_x$ playing the roles of $\rho\pa_\rho$ and $R\pa_R,V_a$). Thus,
  \begin{subequations}
  \begin{align}
  \label{EqE2CommRho}
  \begin{split}
    L_{h_0,\tilde h}(\rho\pa_\rho v) &= \rho\pa_\rho f + [ L_{h_0,\tilde h}, \rho\pa_\rho ]v \\
      &= \rho\pa_\rho f - \sum_{j+|\gamma|+i\leq 2} (\rho\pa_\rho\ell_{j\gamma i}) (\rho R\pa_R)^j (\rho\sV)^\gamma (\rho\pa_\rho)^i v \\
      &\quad\hspace{2.4em} - \sum_{j+|\gamma|+i\leq 2} \ell_{j\gamma i} (j+|\gamma|)(\rho R\pa_R)^j(\rho\sV)^\gamma(\rho\pa_\rho)^i  v,
  \end{split} \\
  \label{EqE2CommR}
    L_{h_0,\tilde h}(R\pa_R v) &= R\pa_R f - \sum_{j+|\gamma|+i\leq 2} (R\pa_R\ell_{j\gamma i})(\rho R\pa_R)^j(\rho\sV)^\gamma(\rho\pa_\rho)^i v, \\
  \label{EqE2CommV}
  \begin{split}
    L_{h_0,\tilde h}(V_a v) &= V_a f - \sum_{j+|\gamma|+i\leq 2} (V_a \ell_{j\gamma i})(\rho R\pa_R)^j(\rho\sV)^\gamma(\rho\pa_\rho)^i v \\
      &\quad\hspace{2em} - \sum_{j+|\gamma|+i\leq 2} \ell_{j\gamma i} (\rho R\pa_R)^j[V_a,(\rho\sV)^\gamma](\rho\pa_\rho)^i v.
  \end{split}
  \end{align}
  \end{subequations}
  Those terms on the right hand sides in which $j+|\gamma|+i\leq 1$ can be estimated in the space $R^{-\alpha}\rho^{-N}L^2$ by $\sum_{j,\gamma,i}\|\ell_{j\gamma i}\|_{\cC_\bop^1}\|v\|_{R^{-\alpha}\rho^{-N}H_{0,\bop;\bop}^{1;0}}$. By Sobolev embedding (Lemma~\ref{Lemma0bSob}) and using~\eqref{EqE1StructTame} (with the values $k-2=1+d_3$, resp.\ $k-2=1+d_4$), this is bounded by $C(1+\|h_0\|_{R^\alpha\Hb^{3+2 d_3}}+\|\tilde h\|_{R^\alpha\rho^\beta\Hb^{3+2 d_4}})\|v\|_{R^{-\alpha}\rho^{-N}H_{0,\bop;\bop}^{1;0}}$.

  Consider thus the terms with $j+|\gamma|+i=2$. In the first sum of~\eqref{EqE2CommRho}, note that $\rho\pa_\rho$ annihilates the leading order terms of $\ell_{j\gamma i}$ at $\rho=0$, so $\rho\pa_\rho\ell_{j\gamma i}\in\rho\CI+R^\alpha\rho^\beta\Hb^\infty$, with the $R^\alpha\rho^\beta\cC_0^\bop$-norm of the non-smooth contribution $\rho\pa_\rho\tilde\ell_{j\gamma i}$ (cf.\ \eqref{EqE2ellSplit}) bounded by $\|\tilde\ell_{j\gamma i}\|_{R^\alpha\rho^\beta\cC_\bop^1}\leq C\|\tilde\ell_{j\gamma i}\|_{R^\alpha\rho^\beta\Hb^{1+d_4}}\leq C'(\|h_0\|_{R^\alpha\Hb^{3+2 d_3}}+\|\tilde h\|_{R^\alpha\rho^\beta\Hb^{3+2 d_4}})$. When $i=1$ or $2$, we can thus regard $(\rho\pa_\rho\ell_{j\gamma i})(\rho R\pa_R)^j(\rho\sV)^\gamma(\rho\pa_\rho)^{i-1}$ as a contribution to $A_{1 1}$ (i.e.\ this is applied to $\rho\pa_\rho v$) in the notation of~\eqref{EqE2LOp} (where $K=1+1+3$). When $i=0$ on the other hand, and $j\geq 1$, say, then the term $(\rho\pa_\rho\ell_{j\gamma i}) (\rho R\pa_R)^{j-1}(\rho\sV)^\gamma\circ\rho R\pa_R$ gives rise to a contribution $\rho(\rho\pa_\rho\ell_{j\gamma i})(\rho R\pa_R)^{j-1}(\rho\sV)^\gamma$ to $A_{1 2}$ (i.e.\ this is applied to $R\pa_R v$), and more precisely to $\tilde A_{1 2}$ in view of the vanishing factor of $\rho$; when $j=0$ and $|\gamma|=2$, we instead get an analogous contribution to $A_{1 a}$ (in fact, to $\tilde A_{1 a}$) for $a=3,4,5$.

  Turning to the second sum of~\eqref{EqE2CommRho}, note that now at least one of $j,|\gamma|$ is nonzero. Say $j\geq 1$ (the case $|\gamma|\geq 1$ being completely analogous); then we can write
  \begin{equation}
  \label{EqE2CommjTerm}
    (\rho R\pa_R)^j(\rho\sV)^\gamma(\rho\pa_\rho)^i = \rho(\rho R\pa_R)^{j-1}(\rho\sV)^\gamma(\rho\pa_\rho)^i\circ R\pa_R,
  \end{equation}
  and thus regard $\rho\ell_{j\gamma i}(j+|\gamma|)(\rho R\pa_R)^{j-1}(\rho\sV)^\gamma(\rho\pa_\rho)^i$ as a further contribution to $A_{1 2}$ (i.e.\ this acts on $R\pa_R v$), more precisely to $\tilde A_{1 2}$---note again the presence of the factor $\rho$ here.

  Turning to~\eqref{EqE2CommR}, we can regard those terms for which $i\geq 1$ as contributions to $A_{2 1}$ (i.e.\ acting on $\rho\pa_\rho v$) by writing them as $(R\pa_R\ell_{j\gamma i})(\rho R\pa_R)^j(\rho\sV)^\gamma(\rho\pa_\rho)^{i-1}\circ\rho\pa_\rho$. (The coefficient $R\pa_R\ell_{j\gamma i}$ is of class $\CI+R^\alpha\Hb^\infty+R^\alpha\rho^\beta\Hb^\infty$; it does not need to vanish at $\rho=0$, which is why we allowed for such lower-triangular terms in Step~1.) When $i=0$ and thus one of $j,|\gamma|$ is nonzero, say $j\geq 1$, we again write~\eqref{EqE2CommjTerm} to obtain a contribution to $A_{2 2}$ (with coefficient vanishing at $\rho=0$). The equation~\eqref{EqE2CommV} is treated completely analogously (using that\footnote{It is not important that the $V_a$ are rotation vector fields here; it suffices that, by virtue of them spanning $\cV(\Sph^2)$ over $\CI(\Sph^2)$, we can write $[V_a,V_b]=\sum_c f_{a b}^c V_c$ for some $f_{a b}^c\in\CI(\Sph^2)$.} $[V_a,V_b]=\eps_{a b c}V_c$).

  Altogether, we thus obtain a $5\times 5$ system of the form~\eqref{EqE2LOp}--\eqref{EqE2IVP}, with $v$ and $f$ replaced by $v':=(\rho\pa_\rho v,R\pa_R v,V_1 v,V_2 v,V_3 v)$ and $f'=(f_I)_{1\leq I\leq 5}$, respectively, where $f_1,f_2,f_{a+2}$ ($a=1,2,3$) is given by the right hand sides of~\eqref{EqE2CommRho}--\eqref{EqE2CommV} without the terms with $j+|\gamma|+i=2$. The initial datum $v'|_{\Sigma_{\rho_0,R_0}}$ can be computed in terms of $v_0,R\pa_R v_0,V_a v_0,v_1$. For the initial datum $(\rho\pa_\rho v')|_{\Sigma_{\rho_0,R_0}}$, we only need to determine $(\rho\pa_\rho)^2 v|_{\rho=\rho_0}$; but
  \begin{equation}
  \label{EqE2CauchyData}
    (\rho\pa_\rho)^2 v = \frac{1}{\ell_{2 0 0}}\biggl( L_{h_0,\tilde h}v - \sum_{\genfrac{}{}{0pt}{}{j+|\gamma|+i\leq 2}{i\leq 1}} \ell_{j\gamma i} (\rho R\pa_R)^j(\rho\sV)^\gamma(\rho\pa_\rho)^i v \biggr).
  \end{equation}
  Note that $\ell_{2 0 0}=-g^{-1}(\frac{\dd\rho}{\rho},\frac{\dd\rho}{\rho})\equiv -g_\sfb^{-1}(\frac{\dd\rho}{\rho},\frac{\dd\rho}{\rho})\bmod R^\alpha\Hb^\infty(\cI^+_{R_0})+R^\alpha\rho^\beta\Hb^\infty(\Omega_{\rho_0,R_0})$ is equal to $-g_\sfb^{-1}(\frac{\dd\rho}{\rho},\frac{\dd\rho}{\rho})>0$ (cf.\ Definition~\ref{Def0bMfd}) plus a small correction (by~\eqref{EqE2h}), and thus bounded away from $0$. Furthermore, the restriction of $L_{h_0,\tilde h}v=f\in R^\alpha\rho^{-N}\Hb^1$ to $\rho=\rho_0$ lies in $R^\alpha L^2(\Sigma_{\rho_0,R_0})$ by Lemma~\ref{Lemma0bSob}\eqref{It0bSobTrace}. The estimate~\eqref{EqE2RegBasicEst}, with $C$ depending on the $A_{I J}$, now gives~\eqref{EqE2TameLo} for $k=1$.

  \pfsubstep{(2.2)}{Higher b-regularity with tame estimates.} We claim that
  \begin{equation}
  \label{EqE2TameHi}
  \begin{split}
    &\|v\|_{R^\alpha\rho^{-N}H_{0,\bop;\bop}^{1;k}} \\
    &\qquad \leq C_k\Bigl( \|(f,v_0,v_1)\|_{D^{k,\alpha,-N}} \\
    &\qquad \hspace{4em} + \bigl(\|h_0\|_{R^\alpha\Hb^{k+2+2 d_3}} + \|\tilde h\|_{R^\alpha\rho^\beta\Hb^{k+2+2 d_4}}\bigr) \|v\|_{R^\alpha\rho^{-N}H_{0,\bop}^1} + \|v\|_{R^\alpha\rho^{-N}H_{0,\bop;\bop}^{1;k-1}} \Bigr).
  \end{split}
  \end{equation}
  For the proof, we commute the equation $L_{h_0,\tilde h}v=f$ with
  \[
    \sW^\zeta := (\rho\pa_\rho, R\pa_R, (V_1, V_2, V_3))^\zeta,\qquad \zeta=(m,n,\sigma),\ \ |\zeta|:=m+n+|\sigma|=k.
  \]
  On the set $K:=\{\zeta\in\N_0^5 \colon |\zeta|=k\}$, we introduce a weak ordering by declaring $\zeta=(m,n,\sigma)\leq\zeta'=(m',n',\sigma')$ if and only if $m\geq m'$. For $\zeta=(m,n,\sigma)$ with $|\zeta|=k$, consider then
  \begin{align}
    L_{h_0,\tilde h}(\sW^\zeta v) &= \sW^\zeta f - [\sW^\zeta,L_{h_0,\tilde h}]v = \sW^\zeta f - \sum_{j+|\gamma|+i\leq 2} [\sW^\zeta,\ell_{j\gamma i}(\rho R\pa_R)^j(\rho\sV)^\gamma(\rho\pa_\rho)^i]v \nonumber\\
  \label{EqE2CommEq}
  \begin{split}
      &= \sW^\zeta f - \sum_{j+|\gamma|+i\leq 2}[\sW^\zeta,\ell_{j\gamma i}] (\rho R\pa_R)^j(\rho\sV)^\gamma(\rho\pa_\rho)^i v \\
      &\quad \hspace{2.5em} - \sum_{j+|\gamma|+i\leq 2} \ell_{j\gamma i} [ \sW^\zeta, (\rho R\pa_R)^j(\rho\sV)^\gamma(\rho\pa_\rho)^i ]v.
  \end{split}
  \end{align}

  \pfsubstep{(2.2.1)}{Lower order terms I.} On the right hand side of~\eqref{EqE2CommEq}, consider first the terms with $i+|\gamma|+j\leq 1$. We shall estimate these in $R^\alpha\rho^{-N}L^2(\Omega_{\rho_0,R_0})$. By Lemma~\ref{Lemma0bIdeal}, we can write
  \[
    [\sW^\zeta,(\rho R\pa_R)^j(\rho\sV)^\gamma(\rho\pa_\rho)^i]=\sum_{\genfrac{}{}{0pt}{}{j'+|\gamma'|+i'\leq 1}{|\zeta'|\leq k-1}} c^{\zeta,j\gamma i}_{\zeta',j'\gamma'i'}\sW^{\zeta'}(\rho R\pa_R)^{j'}(\rho\sV)^{\gamma'}(\rho\pa_\rho)^{i'}
  \]
  for suitable $c^{\zeta,j\gamma i}_{\zeta',j'\gamma'i'}\in\CI(\breve M)$. Therefore,
  \[
    \|\ell_{j\gamma i}[\sW^\zeta,(\rho R\pa_R)^j(\rho\sV)^\gamma(\rho\pa_\rho)^i]v\|_{R^\alpha\rho^{-N}L^2} \leq C\|\ell_{j\gamma i}\|_{L^\infty} \|v\|_{R^\alpha\rho^{-N}H_{0,\bop;\bop}^{1;k-1}};
  \]
  and $\|\ell_{j\gamma i}\|_{L^\infty}\leq C(1+\|h_0\|_{R^\alpha\Hb^{d_3+2}}+\|\tilde h\|_{R^\alpha\rho^\beta\Hb^{d_4+2}})$ by Sobolev embedding and~\eqref{EqE1StructTame}. Turning to $[\sW^\zeta,\ell_{j\gamma i}]$ and writing $\sW^\zeta=W_1\cdots W_k$ where $W_i\in\{\rho\pa_\rho,R\pa_R,V_1,V_2,V_3\}$, we note that, for any function $\ell$,
  \begin{equation}
  \label{EqE2CommExp}
    [\sW^\zeta,\ell] = \sum_{p=1}^k \sum_{\genfrac{}{}{0pt}{}{I\sqcup J=\{1,\ldots,k\}}{|I|=p}} (W_I\ell)W_J,
  \end{equation}
  where $I=\{i_1,\ldots,i_p\}$ with $1\leq i_1<i_2<\cdots<i_p\leq k$ and $J=\{j_1,\ldots,j_{k-p}\}$ with $1\leq j_1<j_2<\cdots<j_{k-p}\leq k$, and $W_I:=W_{i_1}W_{i_2}\cdots W_{i_p}$. Therefore, schematically writing $D_\bop^p$ for a $p$-fold composition of b-vector fields,
  \begin{align*}
    &\|[\sW^\zeta,\ell_{j\gamma i}](\rho R\pa_R)^j(\rho\sV)^\gamma(\rho\pa_\rho)^i v\|_{R^\alpha\rho^{-N}L^2} \\
    &\qquad \leq C\sum_{p=1}^k \| (D_\bop^p\ell_{j\gamma i}) (D_\bop^{k-p}(\rho R\pa_R)^j(\rho\sV)^\gamma(\rho\pa_\rho)^i v) \|_{L^2}.
  \end{align*}
  We split $\ell_{j\gamma i}$ as in~\eqref{EqE2ellSplit}. The contribution from $\ell_{0,j\gamma i}$ is bounded by $C\|v\|_{R^\alpha\rho^{-N}H_{0,\bop;\bop}^{1;k-1}}$. The contributions from $\ell_{(0),j\gamma i}$ and $\tilde\ell_{j\gamma i}$ can be bounded using Lemma~\ref{Lemma0bSob}\eqref{It0bSobProd} (applied with $a=p-1$, $v_1=D_\bop\ell_{(0),j\gamma i}$, resp.\ $u_1=D_\bop\tilde\ell_{j\gamma i}$ and $b=k-p$, $u_2=(\rho R\pa_R)^j(\rho\sV)^\gamma(\rho\pa_\rho)^i v$) by a constant times
  \begin{align*}
    &\bigl(\|D_\bop\ell_{(0),j\gamma i}\|_{\Hb^{d_3}} + \|D_\bop\tilde\ell_{j\gamma i}\|_{\Hb^{d_4}}\bigr) \|v\|_{R^\alpha\rho^{-N}H_{0,\bop;\bop}^{1;k-1}} \\
    &\qquad + \bigl(\|D_\bop\ell_{(0),j\gamma i}\|_{\Hb^{k-1+d_3}} + \|D_\bop\tilde\ell_{j\gamma i}\|_{\Hb^{k-1+d_4}}\bigr) \|v\|_{R^\alpha\rho^{-N}H_{0,\bop}^1},
  \end{align*}
  which in view of~\eqref{EqE1StructTame} and \eqref{EqE2h} is bounded by a constant times
  \[
    \|v\|_{R^\alpha\rho^{-N}H_{0,\bop;\bop}^{1;k-1}} + \bigl( \|h_0\|_{R^\alpha\Hb^{k+2+2 d_3}}+\|\tilde h\|_{R^\alpha\rho^\beta\Hb^{k+2+2 d_4}}\bigr) \|v\|_{R^\alpha\rho^{-N}H_{0,\bop}^1}.
  \]

  \pfsubstep{(2.2.2)}{Lower order terms II.} We now turn to the terms in the first sum on the right in~\eqref{EqE2CommEq} with $i+|\gamma|+j=2$; we expand $[\sW^\zeta,\ell_{j\gamma i}]$ using~\eqref{EqE2CommExp}. Those terms with $p=|I|\geq 2$ and thus $|J|\leq k-2$ can be estimated, using $\cV_{0,\bop}\subset\Vb$ and writing $D_{0,\bop}$ for a derivative along an element of $\cV_{0,\bop}$, by
  \[
    \|(W_I\ell_{j\gamma i}) W_J (\rho R\pa_R)^j(\rho\sV)^\gamma(\rho\pa_\rho)^i v\|_{R^\alpha\rho^{-N}L^2} \leq C\| (D_\bop^p\ell_{j\gamma i}) (D_\bop^{k-p+1} D_{0,\bop} v)\|_{R^\alpha\rho^{-N}L^2}.
  \]
  Lemma~\ref{Lemma0bSob}\eqref{It0bSobProd} (now with $a=p-2,v_1=D_\bop^2\ell_{(0),j\gamma i}$, resp.\ $u_1=D_\bop^2\tilde\ell_{j\gamma i}$) allows us to estimate this further by a constant times
  \begin{align*}
    &\bigl(\|D_\bop^2\ell_{(0),j\gamma i}\|_{\Hb^{d_3}} + \|D_\bop^2\tilde\ell_{j\gamma i}\|_{\Hb^{d_4}}\bigr) \|v\|_{R^\alpha\rho^{-N}H_{0,\bop;\bop}^{1;k-1}} \\
    &\qquad + \bigl(\|D_\bop^2\ell_{(0),j\gamma i}\|_{\Hb^{k-2+d_3}} + \|D_\bop^2\tilde\ell_{j\gamma i}\|_{\Hb^{k-2+d_4}}\bigr) \|v\|_{R^\alpha\rho^{-N}H_{0,\bop}^1} \\
    &\lesssim \|v\|_{R^\alpha\rho^{-N}H_{0,\bop;\bop}^{1;k-1}} + \bigl( \|h_0\|_{R^\alpha\Hb^{k+2+2 d_3}}+\|\tilde h\|_{R^\alpha\rho^\beta\Hb^{k+2+2 d_4}}\bigr) \|v\|_{R^\alpha\rho^{-N}H_{0,\bop}^1}.
  \end{align*}
  (We use here that $\|D_\bop^2\ell_{(0),j\gamma i}\|_{\Hb^{d_3}}\leq C\|h_0\|_{\Hb^{(d_3+2)+(d_3+2)}}$ by~\eqref{EqE1StructTame}, which is the origin for the assumption~\eqref{EqE2h}; similarly for $\tilde h$.)

  \pfsubstep{(2.2.3)}{Remaining terms; lower triangular structure.} Continuing the study of those terms in the first sum in~\eqref{EqE2CommEq} with $i+|\gamma|+j=2$, and using the notation introduced for~\eqref{EqE2CommExp}, it remains to deal with
  \[
    (W_q\ell_{j\gamma i}) W_1\cdots\wh{W_q}\cdots W_k (\rho R\pa_R)^j(\rho\sV)^\gamma(\rho\pa_\rho)^i v
  \]
  where the hat indicates the omission of a term. By Lemma~\ref{Lemma0bIdeal}, the commutator of $W_1\cdots\wh{W_q}\cdots W_k$ with $(\rho R\pa_R)^j(\rho\sV)^\gamma(\rho\pa_\rho)^i$ is schematically of the form $D_\bop^{k-2}D_{0,\bop}^2$ and thus a fortiori of the form $D_\bop^{k-1}D_{0,\bop}^1$. Therefore, its contribution is bounded by $\|D_\bop\ell_{j\gamma i}\|_{L^\infty}\|v\|_{R^\alpha\rho^{-N}H_{0,\bop;\bop}^{1;k-1}}$. Up to terms with these bounds, we can thus freely rearrange all vector fields. Suppose first that $W_q=\rho\pa_\rho$; then $W_q\ell_{j\gamma i}\in\rho\CI+R^\alpha\rho^\beta\Hb^\infty$, so we can write $(W_q\ell_{j\gamma i})(\rho R\pa_R)^j(\rho\sV)^\gamma(\rho\pa_\rho)^i W_1\cdots\wh{W_q}\cdots W_k v$ as the action of
  \[
    (\rho\pa_\rho\ell_{j\gamma i})D_{0,\bop}\in\rho\Diff_{0,\bop}^1+R^\alpha\rho^\beta\cC_\bop^0\Diff_{0,\bop}^1
  \]
  (which contributes to the appropriate $A_{\zeta\zeta'}$ term, $\zeta'\in K$, in~\eqref{EqE2LOp}) on $D_\bop^k v$. On the other hand, when $W_q=R\pa_R$ (and similarly when $W_q=V_1,V_2,V_3$), we need to distinguish two cases: the first case is that $i=2$, in which case we have the term
  \[
    (W_q\ell_{j\gamma i})\rho\pa_\rho( \rho\pa_\rho W_1\cdots\wh{W_q}\cdots W_k v),
  \]
  which contributes $(W_q\ell_{j\gamma i})\rho\pa_\rho$ to $A_{\zeta\zeta'}$ where $\zeta'=(m+1,n-1,\sigma)<\zeta$ (by which we man that $\zeta\leq\zeta'$ does not hold). (Thus $A_{\zeta\zeta'}$ is a strictly lower triangular term, with coefficients that need not vanish at $\rho=0$.) When $i\leq 1$, then among the two factors in $(\rho R\pa_R)^j(\rho\sV)^\gamma(\rho\pa_\rho)^i$ there is at least one (namely, one of $\rho R\pa_R$ and $\rho\sV$) of the form $\rho D_\bop$, and thus we get a term $(W_q\ell_{j\gamma i})D_{0,\bop}\rho (D_\bop^k v)$, which is again a trivial contribution to $\cA=(A_{\zeta\zeta'})$ due to the factor of $\rho$.

  Finally, consider the terms in the second sum in~\eqref{EqE2CommEq} with $i+|\gamma|+j=2$. Upon expanding the commutator, i.e.\ applying~\eqref{EqE2CommExp} with $W_I\ell$ for $\ell=(\rho R\pa_R)^j(\rho\sV)^\gamma(\rho\pa_\rho)^i$ now meaning $[W_{i_1},[W_{i_2}\cdots[W_{i_p},\ell]\cdots]]$, all terms with $|I|=p\geq 2$ are of the schematic form $D_{0,\bop}^2 D_\bop^{k-p}v$, so a fortiori $D_{0,\bop}D_\bop^{k-1}v$, and can thus be estimated by $\|\ell_{j\gamma i}D_{0,\bop}^1 D_\bop^{k-1}v\|_{R^\alpha\rho^{-N}L^2}\lesssim\|\ell_{j\gamma i}\|_{L^\infty}\|v\|_{H_{0,\bop;\bop}^{1;k-1}}$. It thus suffices to analyze the terms
  \begin{equation}
  \label{EqE2CommFinal}
    [W_q,(\rho R\pa_R)^j(\rho\sV)^\gamma(\rho\pa_\rho)^i] W_1\cdots\wh{W_q}\cdots W_k v.
  \end{equation}
  When $|\gamma|\geq 1$ and $W_q=V_1,V_2,V_3$, the commutator is a sum of terms which are of the form $(\rho R\pa_R)^j(\rho\sV)^{\gamma'}(\rho\pa_\rho)^i$ where $|\gamma'|=|\gamma|$. We then shift one factor of $\sV$ to the right and obtain a term of the form
  \[
    \rho\ell_{j\gamma i}(\rho R\pa_R)^j(\rho\sV)^{\gamma''}(\rho\pa_\rho)^i(V_a W_1\cdots\wh{W_q}\cdots W_k v),\qquad |\gamma''|=|\gamma|-1;
  \]
  thus $\rho\ell_{j\gamma i}(\rho R\pa_R)^j(\rho\sV)^{\gamma''}(\rho\pa_\rho)^i$ is a trivial contribution to the appropriate coefficient $A_{\zeta\zeta'}$. If $W_q=R\pa_R$, the commutator in~\eqref{EqE2CommFinal} vanishes. If $W_q=\rho\pa_\rho$, the commutator is equal to $(j+|\gamma|)(\rho R\pa_R)^j(\rho\sV)^\gamma(\rho\pa_\rho)^i$; for it to be nonzero, we must have $j+|\gamma|\geq 1$. Say $j\geq 1$; then we can shift one factor of $R\pa_R$ to the right and remain with $(j+|\gamma|)\rho\ell_{j\gamma i}(\rho R\pa_R)^{j-1}(\rho\sV)^\gamma(\rho\pa_\rho)^i\circ(R\pa_R W_1\cdots\wh{W_q}\cdots W_k v)$, with the operator on the left again giving a trivial contribution (due to the factor of $\rho$) to the appropriate $A_{\zeta\zeta'}$ (here $\zeta'=(m-1,n+1,\sigma)$)

  Altogether, we have shown that there exist $A_{\zeta\zeta'}=A_{0,\zeta\zeta'}+A_{(0),\zeta\zeta'}+\tilde A_{\zeta\zeta'}$ for $\zeta,\zeta'\in K$ with $A_{0,\zeta\zeta'}\in\Diff_{0,\bop}^1$, $A_{(0),\zeta\zeta'}\in R^\alpha\cC_\bop^0(\cI^+_{R_0})\Diff_{0,\bop}^1$, and $\tilde A_{\zeta\zeta'}\in R^\alpha\rho^\beta\cC_\bop^0(\Omega_{\rho_0,R_0})\Diff_{0,\bop}^1$, with the following properties:
  \begin{itemize}
  \item for $\zeta\leq\zeta'$, we have $A_{0,\zeta\zeta'}\in\rho\Diff_{0,\bop}^1$, $A_{(0),\zeta\zeta'}=0$;
  \item by Sobolev embedding to control $W_q\ell_{j\gamma i}$ in $L^\infty$ spaces, and again using~\eqref{EqE1StructTame},
    \[
      \|A_{(0),\zeta\zeta'}\|_{R^\alpha\cC_\bop^0\Diff_{0,\bop}^1},\ \|\tilde A_{\zeta\zeta'}\|_{R^\alpha\rho^\beta\cC_\bop^0\Diff_{0,\bop}^1} \leq C_k\bigl( \|h_0\|_{R^\alpha\Hb^{2 d_3+4}} + \|\tilde h\|_{R^\alpha\rho^\beta\Hb^{2 d_4+4}} \bigr),
    \]
    which are in turn bounded by~\eqref{EqE2h};
  \item set
    \begin{equation}
    \label{EqE2CommOp}
      \cL:=(\delta_{\zeta\zeta'}L_{h_0,\tilde h}-A_{\zeta\zeta'})_{\zeta,\zeta'\in K},\quad v':=(\sW^\zeta v)_{\zeta\in K},
    \end{equation}
    then $\cL v'=f'$ where $f'$ satisfies the bound
    \begin{align*}
      &\|f'\|_{R^\alpha\rho^{-N}L^2(\Omega_{\rho_0,R_0})} \\
      &\qquad \leq C_k\Bigl( \|f\|_{R^\alpha\rho^{-N}\Hb^k} + \bigl(\|h_0\|_{R^\alpha\Hb^{k+2+2 d_3}} + \|\tilde h\|_{R^\alpha\rho^\beta\Hb^{k+2+2 d_4}}\bigr)\|v\|_{R^\alpha\rho^{-N}H_{0,\bop}^1} \\
      &\qquad \hspace{25em} + \|v\|_{R^\alpha\rho^{-N}H_{0,\bop;\bop}^{1;k-1}}\Bigr).
    \end{align*}
  \end{itemize}

  \pfsubstep{(2.2.4)}{Initial data for the commuted equation.} It remains to control the Cauchy data of $v'$ in~\eqref{EqE2CommOp} at $\Sigma_{\rho_0,R_0}\subset\{\rho=\rho_0\}$. Since $R\pa_R,V_a$ are tangent to $\Sigma_{\rho_0,R_0}$, we only need to prove a tame estimate for $v_{0,p}:=(\rho\pa_\rho)^p v|_{\rho=\rho_0}$ in $R^\alpha\Hb^{k+1-p}(\Sigma_{\rho_0,R_0})$, $p=0,\ldots,k+1$. For $p=0,1$, we simply have $v_{0,p}=v_p$. For $p\geq 2$, we use the spacetime identity~\eqref{EqE2CauchyData}, written as
  \[
    (\rho\pa_\rho)^2 v=\frac{1}{\ell_{2 0 0}} \bigl( f - L_0 v - L_1(\rho\pa_\rho v) \bigr),\qquad L_q := \sum_{j+|\gamma|\leq 2-q} \ell_{j\gamma q}(\rho R\pa_R)^j(\rho\sV)^\gamma,\ \ q=0,1,
  \]
  to deduce that
  \begin{equation}
  \label{EqE2CommData}
    v_{0,p} = \sum_{a=0}^{p-2} \binom{p-2}{a} \Bigl((\rho\pa_\rho)^{p-2-a}\frac{1}{\ell_{2 0 0}}\Bigr) \bigl( (\rho\pa_\rho)^a f - (\rho\pa_\rho)^a L_0 v - (\rho\pa_\rho)^a L_1(\rho\pa_\rho v) \bigr) \Big|_{\Sigma_{\rho_0,R_0}}.
  \end{equation}
  We claim that
  \begin{equation}
  \label{EqE2DataEst}
  \begin{split}
    &\| v_{0,p} \|_{R^\alpha\Hb^{k+1-p}(\Sigma_{\rho_0,R_0})} \\
    &\qquad \leq C_k\Bigl( \|(f,v_0,v_1)\|_{D^{k,\alpha,-N}} \\
    &\qquad \hspace{4em} + \bigl( \|h_0\|_{R^\alpha\Hb^{k+2+2 d_3}}+\|\tilde h\|_{R^\alpha\rho^\beta\Hb^{k+2+2 d_4}}\bigr)\|(f,v_0,v_1)\|_{D^{0,\alpha,-N}} \Bigr).
  \end{split}
  \end{equation}
  We shall only prove this estimate for the term in~\eqref{EqE2CommData} with $a=p-2$, and indeed for $\ell_{2 0 0}:=1$; we leave the simple modifications required to treat the full expression (based on further applications of the tame product estimates of Lemma~\ref{Lemma0bSob}) to the reader. 

  The term in~\eqref{EqE2CommData} involving $f$ is bounded using Lemma~\ref{Lemma0bSob}\eqref{It0bSobTrace} by
  \begin{align*}
    \|(\rho\pa_\rho)^{p-2}f\|_{R^\alpha\Hb^{k+1-p}(\Sigma_{\rho_0,R_0})} &\leq C\|(\rho\pa_\rho)^{p-2}f\|_{R^\alpha\rho^{-N}\Hb^{k+2-p}(\Omega_{\rho_0,R_0})} \\
      &\leq C'\|f\|_{R^\alpha\rho^{-N}\Hb^k(\Omega_{\rho_0,R_0})}.
  \end{align*}
  For the estimate of the $R^\alpha\Hb^{k+1-p}$-norm of the term involving $L_1$ (the term $L_0$ is treated similarly and left to the reader), we only consider derivatives along $(R\pa_R)^{k+1-p}$; derivatives along $(R\pa_R)^q\sV^\gamma$ for $q+|\gamma|\leq k+1-p$ can then be handled in the same fashion with purely notational modifications. We thus need to prove a bound in $R^\alpha L^2(\Sigma_{\rho_0,R_0})$ for $(R\pa_R)^{k+1-p}(\rho\pa_\rho)^{p-2} L_1(\rho\pa_\rho v)$, which is a sum of terms of the form
  \[
    \bigl((R\pa_R)^{q'}(\rho\pa_\rho)^{p'}\ell_{j\gamma 1}\bigr) \cdot (R\pa_R)^{q''}(\rho R\pa_R)^j(\rho\sV)^\gamma (\rho\pa_\rho)^{p''+1}v,
  \]
  where $q'+q''=k+1-p$ and $p'+p''\leq p-2$, and $j+|\gamma|\leq 1$; schematically, this is thus
  \[
    (D_\bop^{q'+p'}\ell_{j\gamma 1}|_{\Sigma_{\rho_0,R_0}})(D_\bop^{q''+1} v_{0,p''+1}).
  \]
  Since $p''+1\leq p-1$, we can iteratively express $v_{0,p''+1}$ using the formula~\eqref{EqE2CommData}, and proceed in this fashion until we obtain an expression involving only $f$, $v_{0,0}=v_0$, $v_{0,1}=v_1$, and the coefficients of $L_{h_0,\tilde h}$. We discuss here only the case $p''=0$ and $p'=p-2$, in which case we can use Lemma~\ref{Lemma0bSob} (specifically, the estimate~\eqref{Eq0bSobProd2}, which also applies on $\cU=\Sigma_{\rho_0,R_0}$) and Lemma~\ref{Lemma0bSob}\eqref{It0bSobTrace} (which gives an estimate $\|\tilde\ell_{j\gamma 1}\|_{R^\alpha\Hb^m(\Sigma_{\rho_0,R_0})}\leq C\|\tilde\ell_{j\gamma 1}\|_{R^\alpha\rho^\beta\Hb^{m+1}(\Omega_{\rho_0,R_0})}$) to bound
  \begin{align*}
    &\|(D_\bop^{q'+p-2}\ell_{j\gamma 1}|_{\Sigma_{\rho_0,R_0}}) (D_\bop^{q''+1}v_{0,1}) \|_{R^\alpha L^2} \\
    &\qquad \lesssim \bigl(1+\|\ell_{(0),j\gamma 1}\|_{R^\alpha\Hb^{d_3}}+\|\tilde\ell_{j\gamma 1}\|_{R^\alpha\rho^\beta\Hb^{d_4}}\bigr) \|v_{0,1}\|_{R^\alpha\Hb^k} \\
    &\qquad \qquad + \bigl(1+\|\ell_{(0),j\gamma 1}\|_{R^\alpha\Hb^{k+d_3}}+\|\tilde\ell_{j\gamma 1}\|_{R^\alpha\rho^\beta\Hb^{k+d_4}}\bigr) \|v_{0,1}\|_{R^\alpha L^2} \\
    &\qquad \lesssim \|(0,0,v_1)\|_{D^{k,\alpha,-N}} + \bigl(\|h_0\|_{R^\alpha\Hb^{k+2+2 d_3}} + \|\tilde h\|_{R^\alpha\rho^\beta\Hb^{k+2+2 d_4}}\bigr) \|(0,0,v_1)\|_{D^{0,\alpha,-N}};
  \end{align*}
  here we use $(q'+p-2)+(q''+1)=k$. In this fashion one proves~\eqref{EqE2DataEst}.

  We can finally apply the estimate~\eqref{EqE2RegBasicEst} to the initial value problem for $\cL'v'=f'$ to finish the proof of~\eqref{EqE2TameHi} and thus of the Proposition.
\end{proof}

\subsubsection{Asymptotics and decay}
\label{SssE2Decay}

We continue assuming~\eqref{EqE2h}, and drop the bundle $\upbeta^*(S^2\,{}^0 T^*M)$ from the notation. We recall the cutoff $\chi$ from~\eqref{EqECutoffs}. Starting from the estimate~\eqref{EqE2RegEst} (for large $k$) for the solution of an initial value problem for $L_{h_0,\tilde h}v=f$, we now extract stronger information about the asymptotic behavior of $v$ near $\rho=0$ (assuming appropriate decay for $f$) using an indicial operator argument.

It is convenient to straighten out the domain $\Omega_{\rho_0,R_0}$ from Definition~\ref{Def0bDom}: introduce
\begin{equation}
\label{EqE2DecayCoordp}
  \rho'=\rho,\qquad R'=(1-2\rho)R,
\end{equation}
and set $R'_0:=(1-2\rho_0)R_0$, then
\[
  \Omega_{\rho_0,R_0}=\{\rho'\leq\rho_0,\ R'\leq R'_0\}=[0,\rho_0]_{\rho'}\times[0,R'_0]_{R'}\times\Sph^2.
\]
The product nature of $\Omega_{\rho_0,R_0}$ in these coordinates is closely related to the fact that the vector field
\begin{equation}
\label{EqE2DecayRhop}
  \rho'\pa_{\rho'} = \rho\pa_\rho + c(\rho)\rho R\pa_R,\qquad c(\rho):=\frac{2}{1-2\rho},
\end{equation}
is tangent to the (final spacelike) boundary hypersurface $\Sigma_{\rho_0,R_0}^+=\{(1-2\rho)R=(1-2\rho_0)R_0\}$ of $\Omega_{\rho_0,R_0}$. It is in these adapted coordinates that we now discuss the inversion of the indicial operator.

\begin{lemma}[Inversion of the indicial operator]
\label{LemmaE2Inv}
  Let $\alpha\in\R$ and $\rho_1\in(0,\rho_0)$, further $\eta_1<\eta_2<1$ with $\eta_1,\eta_2\neq 0$, and $k\in\N_0$. Recall the operator $I_{g_{(0)}}$ from Proposition~\usref{PropE1Struct}\eqref{ItE1StructInd}, where $g_{(0)}=\dd x^2+h_{(0)}$ is defined as in~\eqref{EqE1Structg0h0}. Suppose $v\in R^\alpha\rho^{\eta_1}\Hb^k(\Omega_{\rho_0,R_0})$ vanishes for $\rho\geq\rho_1>0$, and $I_{g_{(0)}}(\rho'\pa_{\rho'})v\in R^\alpha\rho^{\eta_2}\Hb^k(\Omega_{\rho_0,R_0})$.
  \begin{enumerate}
  \item\label{ItE2InvGrow}{\rm (Improving the weight.)} If $\eta_1<\eta_2<0$ or $0<\eta_1<\eta_2<1$, then $v\in R^\alpha\rho^{\eta_2}\Hb^k(\Omega_{\rho_0,R_0})$ and
    \begin{equation}
    \label{EqE2InvGrowEst}
      \|v\|_{R^\alpha\rho^{\eta_2}\Hb^k} \leq C_k\Bigl(\|I_{g_{(0)}}v\|_{R^\alpha\rho^{\eta_2}\Hb^k} + \|h_0\|_{R^\alpha\Hb^{k+d_3}}\|I_{g_{(0)}}v\|_{R^\alpha\rho^{\eta_2}L^2} \Bigr).
    \end{equation}
  \item\label{ItE2InvDecay}{\rm (Extracting asymptotics.)} If $\eta_1<0<\eta_2$, then there exist $v_0\in R^\alpha\Hb^k(\cI^+_{R_0};\tau^{-2}\ker\tr_{g_{(0)}})$, $\tilde v'\in R^\alpha\rho^{\eta_2}\Hb^k(\Omega_{\rho_0,R_0})$ so that
    \begin{equation}
    \label{EqE2InvDecay}
      v(\rho',R',\omega)=\chi(\rho')v_0(R',\omega)+\tilde v'(\rho',R',\omega),
    \end{equation}
    and $\|v_0\|_{R^\alpha\Hb^k}+\|\tilde v'\|_{R^\alpha\rho^{\eta_2}\Hb^k}$ is bounded by the right hand side of~\eqref{EqE2InvGrowEst}.
  \end{enumerate}
\end{lemma}

The proof, given below, relies on a contour shifting argument on the Mellin transform side. Our convention for the Mellin transform is
\[
  (\cM v)(\lambda,R',\omega) := \int_0^\infty \rho'{}^{-\lambda}v(\rho',R',\omega)\,\frac{\dd\rho'}{\rho'}.
\]
This intertwines $I_{g_{(0)}}(\rho'\pa_{\rho'})$ with $I_{g_{(0)}}(\lambda)$. The Plancherel theorem gives an isomorphism
\begin{equation}
\label{EqE2InvPlancherel}
  \cM \colon \rho'{}^\eta L^2\bigl((0,\infty)_{\rho'}\times\cI^+_{R'_0}\bigr) \to L^2\bigl( \{\Re\lambda=\eta\}; L^2(\cI^+_{R'_0})\bigr),
\end{equation}
where on $\cI^+_{R'_0}$ we use the density $\mu:=|\frac{\dd R'}{R'}\dd\slg|$, and on the left the density $|\frac{\dd\rho'}{\rho'}|\otimes\mu$. The inverse Mellin transform is
\[
  (\cM^{-1}_\eta w)(\rho',R',\omega) := \frac{1}{2\pi i} \int_{\eta-i\infty}^{\eta+i\infty} \rho'{}^\lambda w(\lambda,R',\omega)\,\dd\lambda.
\]
For $\lambda\in\C$, let us write $H_{\bop,\lambda}^{k,\alpha}(\cI^+_{R'_0})$ for the space $\Hb^{k,\alpha}(\cI^+_{R'_0})=R^\alpha\Hb^k(\cI^+_{R'_0})$ with norm
\[
  \| w \|_{H_{\bop,\lambda}^{k,\alpha}(\cI^+_{R'_0})}^2 := \sum_{j+|\gamma|+i\leq k} \| R^{-\alpha} (R\pa_R)^j\sV^\gamma\lambda^i w \|_{L^2(\cI^+_{R'_0})}^2.
\]
Then~\eqref{EqE2InvPlancherel} generalizes to the isomorphism
\begin{equation}
\label{EqE2Plancherel}
  \cM \colon R^\alpha\rho'{}^\eta\Hb^k([0,\infty)_{\rho'}\times\cI^+_{R'_0}) \to L^2\bigl(\{\Re\lambda=\eta\}; H_{\bop,\lambda}^{k,\alpha}(\cI^+_{R'_0})\bigr);
\end{equation}
the b-Sobolev space on the left is defined via testing with $\rho'\pa_{\rho'}$, $R'\pa_{R'}$, $V_a$ ($a=1,2,3$).

\begin{proof}[Proof of Lemma~\usref{LemmaE2Inv}]
  Write $f:=I_{g_{(0)}}(\rho'\pa_{\rho'})v$. For clarity, we write $I_{g_{(0)}}(\lambda,R',\omega)$ for the indicial family; this is, for fixed $R',\omega$, a linear map on the fiber of $\upbeta^*(S^2\,{}^0 T^*M)$ over $(R',\omega)\in\cI^+\subset\breve M$. Since by Lemma~\ref{LemmaEIndRoot} $I_{g_{(0)}}(\lambda,R',\omega)$ is invertible when $\Re\lambda<1$, $\lambda\neq 0$, and thus in particular for $\Re\lambda=\eta_1$, we can then express
  \begin{equation}
  \label{EqE2InvMellin}
    v(\rho',R',\omega) = \frac{1}{2\pi i}\int_{\eta_1-i\infty}^{\eta_1+\infty} \rho'{}^\lambda I_{g_{(0)}}(\lambda,R',\omega)^{-1} (\cM f)(\lambda,R',\omega)\,\dd\lambda.
  \end{equation}
  Since $f$ vanishes for large $\rho'$, its Mellin transform $\cM f(\lambda,\cdot)$ is holomorphic in $\Re\lambda<\eta_2$ with values in $R^\alpha\Hb^k(\cI^+_{R'_0})$.

  We aim to exploit the meromorphicity of $I_{g_{(0)}}(\lambda,R',\omega)^{-1}$ in $\Re\lambda<1$, with only a simple pole at $\lambda=0$. It is convenient to use the expression~\eqref{EqEIndRootPf} in the $g_{(0)}$-dependent splitting~\eqref{EqEIndRootSplit} of $S^2 T^*X$. Now, \eqref{EqEIndRootPf} has a $(3+1)\times(3+1)$ block structure, with a $3\times 3$ minor without poles in $\Re\lambda<1$, while the $(4,4)$ entry is $\lambda^{-1}(\lambda-3)^{-1}$. The map $S^2 T^*X\to\R g_{(0)}\oplus\ker\tr_{g_{(0)}}$ is given by
  \[
    S^2 T^*X \ni h \mapsto \bigl( \tfrac13 g_{(0)}\tr_{g_{(0)}}h,\ h-\tfrac13 g_{(0)}\tr_{g_{(0)}}h \bigr).
  \]
  Therefore, we can write
  \[
    I_{g_{(0)}}(\lambda,R',\omega)^{-1} = \lambda^{-1}(\lambda-3)^{-1}A_{g_{(0)}}(R',\omega) + B_{g_{(0)}}(\lambda,R',\omega)
  \]
  where $A_{g_{(0)}}=\diag(0,0,I-\frac13 g_{(0)}\tr_{g_{(0)}})$ in the splitting~\eqref{Eq0bST0Split}, while the matrix coefficients of $B_{g_{(0)}}$ are rational functions of $\lambda$ without poles in $\Re\lambda<1$ whose coefficients are linear combinations of constants, $g_{(0)}$ (third row), $\tr_{g_{(0)}}$ (third column), and $g_{(0)}\tr_{g_{(0)}}$ ($(3,3)$ entry). Fixing any fixed positive definite fiber inner product on $\upbeta^*(S^2\,{}^0 T^*M)$, we moreover have\footnote{The bound can be sharpened to $C(1+|\lambda|)^{-2}$, though this will not be of use in what follows.}
  \[
    \|B_{g_{(0)}}(\lambda,R',\omega)\|\leq C,\qquad \eta_1\leq\Re\lambda\leq\eta_2.
  \]
  Using Lemma~\ref{Lemma0bSob}\eqref{It0bSobProd}, specifically the estimate~\eqref{Eq0bSobProd2}, we can now estimate
  \begin{align*}
    &\bigl\|(R\pa_R)^j\sV^\gamma\bigl( B_{g_{(0)}}(\lambda,\cdot)(\cM f)(\lambda,\cdot) \bigr) \bigr\|_{R^\alpha L^2(\cI^+_{R'_0})} \\
    &\qquad\qquad \leq C\Bigl( (1+\|h_0\|_{R^\alpha\Hb^{d_3}}) \|(\cM f)(\lambda,\cdot)\|_{R^\alpha\Hb^{j+|\gamma|}} \\
    &\qquad\qquad \hspace{4em} + (1+\|h_0\|_{R^\alpha\Hb^{j+|\gamma|+d_3}}) \|(\cM f)(\lambda,\cdot)\|_{R^\alpha L^2} \Bigr).
  \end{align*}
  Multiplying this with $\lambda^i$ and summing over all $j,\gamma,i$ with $j+|\gamma|+i\leq k$, we obtain (using~\eqref{EqE2h})
  \begin{align*}
    &\| B_{g_{(0)}}(\lambda,\cdot)(\cM f)(\lambda,\cdot) \|_{H_{\bop,\lambda}^{k,\alpha}} \\
    &\qquad \leq C\Bigl( \|(\cM f)(\lambda,\cdot)\|_{H_{\bop,\lambda}^{k,\alpha}} + \sum_{p=0}^k \|h_0\|_{R^\alpha\Hb^{p+d_3}}\|(\cM f)(\lambda,\cdot)\|_{H_{\bop,\lambda}^{k-p,\alpha}} \Bigr).
  \end{align*}
  In the sum, the estimate~\eqref{Eq0bSobL2} shows that it suffices to keep the terms with $p=0,k$, and thus
  \[
    \| B_{g_{(0)}}(\lambda,\cdot)(\cM f)(\lambda,\cdot) \|_{H_{\bop,\lambda}^{k,\alpha}} \leq C\Bigl( \|(\cM f)(\lambda,\cdot)\|_{H_{\bop,\lambda}^{k,\alpha}} + \|h_0\|_{R^\alpha\Hb^{k+d_3}}\|(\cM f)(\lambda,\cdot)\|_{R^\alpha L^2} \Bigr).
  \]
  The same estimate applies for $A_{g_{(0)}}(\cdot)$ in place of $B_{g_{(0)}}(\lambda,\cdot)$.

  In part~\eqref{ItE2InvGrow} then, we shift the contour in the integral~\eqref{EqE2InvMellin} to $\eta_2+i(-\infty,\infty)$ and use these estimates together with~\eqref{EqE2Plancherel} to conclude. The proof of part~\eqref{ItE2InvDecay} is completely analogous, except now the pole of $I_{g_{(0)}}(\lambda,\cdot)^{-1}$ at $\lambda=0$ causes a contribution due to the residue theorem given by $-\frac13 A_{g_{(0)}}(R',\omega)(\cM f)(0,R',\omega)=:v_0(R',\omega)$, while the integral over the final contour $\eta_2+i(-\infty,\infty)$ gives rise to $\tilde w\in R^\alpha\rho'{}^{\eta_2}\Hb^k([0,\infty)_{\rho'}\times\cI^+_{R'_0})$; we then set $\tilde v':=\chi\tilde w+(1-\chi)v_0$ to conclude.
\end{proof}

In order to switch back to the original $\rho,R$ coordinates in~\eqref{EqE2InvDecay}, we first use Lemma~\ref{Lemma0bExt} to extend $v_0$ to an element of $R^\alpha\Hb^k(\cI^+_{R_0})$. We then have
\[
  v = \chi(\rho)v_0( (1-2\rho)R,\omega ) + \tilde v' = \chi v_0(R,\omega) + \tilde v,
\]
where\footnote{If one worked from the outset with the coordinates $\rho',R'$, the loss of one derivative here would be avoided.}
\begin{align*}
  \tilde v&:=\tilde v'+\chi\bigl(v_0((1-2\rho)R,\omega)-v_0(R,\omega)\bigr) = \tilde v' - 2\rho\chi \int_0^1 R\pa_R v_0((1-2\rho s)R,\omega)\,\dd s \\
   &\in R^\alpha\rho^{\eta_2}\Hb^k(\Omega_{\rho_0,R_0}) + \rho R^\alpha\Hb^{k-1}(\cI^+_{R_0}) \subset R^\alpha\rho^{\eta_2}\Hb^{k-1}(\Omega_{\rho_0,R_0});
\end{align*}
and we have the tame estimate
\begin{equation}
\label{EqE2InvDecay2}
\begin{split}
  &\|v_0\|_{R^\alpha\Hb^k(\cI^+_{R_0})} + \|\tilde v\|_{R^\alpha\rho^{\eta_2}\Hb^{k-1}(\Omega_{\rho_0,R_0})} \\
  &\qquad \leq C_k\Bigl(\|I_{g_{(0)}}v\|_{R^\alpha\rho^{\eta_2}\Hb^k} + \|h_0\|_{R^\alpha\Hb^{k+d_3}}\|I_{g_{(0)}}v\|_{R^\alpha\rho^{\eta_2}L^2} \Bigr).
\end{split}
\end{equation}

\begin{prop}[Tame bounds on decaying spaces]
\label{PropE2Decay}
  There exists $d\in\N$ so that the following holds whenever $\|h_0\|_{R^\alpha\Hb^d},\|\tilde h\|_{R^\alpha\rho^\beta\Hb^d}<1$. The unique solution $v$ of the initial value problem
  \begin{alignat*}{2}
    L_{h_0,\tilde h}v&=f &&\in R^\alpha\rho^\beta\Hb^\infty(\Omega_{\rho_0,R_0}), \\
    (v,\cL_{-\rho\pa_\rho}v)|_{\Sigma_{\rho_0,R_0}} &= (v_0,v_1) && \in R^\alpha\Hb^\infty(\Sigma_{\rho_0,R_0})\oplus R^\alpha\Hb^\infty(\Sigma_{\rho_0,R_0})
  \end{alignat*}
  can be written as
  \[
    v = \chi v_0 + \tilde v,
  \]
  where $v_0\in R^\alpha\Hb^\infty(\cI^+_{R_0})$ and $\tilde v\in R^\alpha\rho^\beta\Hb^\infty(\Omega_{\rho_0,R_0})$ satisfy for all $k\in\N_0$ a tame estimate
  \begin{equation}
  \label{EqE2DecayEst}
  \begin{split}
    \|v_0\|_{R^\alpha\Hb^k} + \|\tilde v\|_{R^\alpha\rho^\beta\Hb^k} &\leq C_k\Bigl( \|(f,v_0,v_1)\|_{D^{k+d,\alpha,\beta}} \\
      &\quad \hspace{3em} + \bigl( \|h_0\|_{R^\alpha\Hb^{k+d}} + \|\tilde h\|_{R^\alpha\rho^\beta\Hb^{k+d}} \bigr) \|(f,v_0,v_1)\|_{D^{0,\alpha,\beta}} \Bigr).
  \end{split}
  \end{equation}
\end{prop}

\begin{rmk}[Value of $d$]
\label{RmkE2Value}
  An inspection of the proof produces a concrete value for $d\geq 2 d_4+2$. For example, any number $d\geq 2\frac{N}{\beta}+16$ works. Thus, if in Proposition~\ref{PropE2Reg} one obtained a specific value for $N$ (by applying more care in the basic energy estimate for $\Box_g$), one could specify $d$ also here. (The value of $\beta$ can be fixed arbitrarily close to $1$, cf.\ the statement of Theorem~\ref{ThmESol}.)
\end{rmk}

\begin{proof}[Proof of Proposition~\usref{PropE2Decay}]
  Let $d'\in\N_0$ (chosen at various places in the argument below to ensure the positivity of all differentiability orders); and let $k\in\N_0$ be arbitrary. Given $(f,v_0,v_1)\in D^{k+d',\alpha,\beta}$, we have $v\in R^\alpha\rho^{-N}\Hb^{k+d'}(\Omega_{\rho_0,R_0})$, with a tame estimate~\eqref{EqE2RegEst}, so
  \begin{equation}
  \label{EqE2Decay0}
  \begin{split}
    &\|v\|_{R^\alpha\rho^{-N}\Hb^{k+d'}} \\
    &\quad \leq C_k\Bigl( \|(f,v_0,v_1)\|_{D^{k+d',\alpha,\beta}} \\
    &\quad \hspace{4em} + \bigl(\|h_0\|_{R^\alpha\Hb^{k+d'+2 d_3+2}}+\|\tilde h\|_{R^\alpha\rho^\beta\Hb^{k+d'+2 d_4+2}}\bigr)\|(f,v_0,v_1)\|_{D^{0,\alpha,\beta}} \Bigr).
  \end{split}
  \end{equation}

  In the notation used in Proposition~\ref{PropE1Struct}\eqref{ItE1StructInd}, we now rewrite the equation $L_{h_0,\tilde h}v=f$ as
  \begin{equation}
  \label{EqE2DecayInd}
    I_{g_{(0)}}(\rho\pa_\rho)(\chi v) = -\bigl(L_{h_0,\tilde h}-I_{g_{(0)}}(\rho\pa_\rho)\bigr)(\chi v) + [L_{h_0,\tilde h},\chi]v + \chi f.
  \end{equation}
  Replacing $\rho\pa_\rho$ on the left by $\rho'\pa_{\rho'}=\rho\pa_\rho+\rho\cR$, $\cR:=c(\rho)R\pa_R\in\Vb(\breve M)$ (cf.\ \eqref{EqE2DecayRhop}), creates a further error term given by the action on $\chi v$ of the operator
  \[
    I_{g_{(0)}}(\rho\pa_\rho) - I_{g_{(0)}}(\rho'\pa_{\rho'}) = I_{g_{(0)}}(\rho\pa_\rho)-I_{g_{(0)}}(\rho\pa_\rho+\rho\cR);
  \]
  this operator is of class $\rho\Diffb^2+\rho R^\alpha\Hb^k\Diffb^2$ (by inspection of~\eqref{EqE1StructIndOp}) and can be written and estimated in the same fashion as~\eqref{EqE1StructInd}--\eqref{EqE1StructIndTame}. In the estimates below, we continue writing $R^\alpha\rho^\beta$ for weights (for notational simplicity---the weight $R'{}^\alpha\rho'{}^{\beta}$ is a positive smooth multiple), but we write $I_{g_{(0)}}:=I_{g_{(0)}}(\rho'\pa_{\rho'})$.

  By slightly increasing $N$, we can ensure that for
  \[
    J := \left\lfloor\frac{N}{\beta}\right\rfloor,
  \]
  we have $-N+J\beta\in(-\beta,0)$. For easier bookkeeping, we moreover require $J\geq 1$.

  \pfstep{Step~1. Almost boundedness.} We shall prove that for all $j\in\N_0$, $j\leq J$ (so with $-N+j\beta<0$), we have $v\in R^\alpha\rho^{-N+j\beta}\Hb^{k+d'-2 j}$, with a tame estimate. For $j=0$, this is the content of~\eqref{EqE2Decay0}. For the inductive step, we assume that, for some $j\geq 1$, we have $v\in R^\alpha\rho^{-N+(j-1)\beta}\Hb^{k+d'-2(j-1)}$. We require $d'\geq 2 j$. We can then estimate the right hand side of~\eqref{EqE2DecayInd}, with $\rho'\pa_{\rho'}$ in place of $\rho\pa_\rho$, using~\eqref{EqE1StructInd}--\eqref{EqE1StructIndTame} by
  \begin{align*}
    &\|I_{g_{(0)}}(\chi v)\|_{R^\alpha\rho^{-N+j\beta}\Hb^{k+d'-2 j}} \\
    &\qquad \leq \|R_0(\chi v)\|_{R^\alpha\rho^{-N+j\beta}\Hb^{k+d'-2 j}} + \|\tilde R_{h_0,\tilde h}(\chi v)\|_{R^\alpha\rho^{-N+j\beta}\Hb^{k+d'-2 j}} \\
    &\qquad \hspace{4em} + \| [L_{h_0,\tilde h},\chi]v \|_{R^\alpha\rho^{-N+j\beta}\Hb^{k+d'-2 j}} + \|\chi f\|_{R^\alpha\rho^{-N+j\beta}\Hb^{k+d'-2 j}}.
  \end{align*}
  The first term is bounded by $C\|v\|_{R^\alpha\rho^{-N+j\beta-1}\Hb^{k+d'-2(j-1)}}$. Using Lemma~\ref{Lemma0bSob}, we can estimate the second term by
  \begin{align*}
    \|\tilde R_{h_0,\tilde h}(\chi v)\|_{R^\alpha\rho^{-N+j\beta}\Hb^{k+d'-2 j}} &\leq C_j\Bigl( \|\chi v\|_{R^\alpha\rho^{-N+(j-1)\beta}\Hb^{k+d'-2(j-1)}} \\
    &\qquad + \bigl( \|h_0\|_{R^\alpha\Hb^{k+d'+2+d_3}}+\|\tilde h\|_{R^\alpha\rho^\beta\Hb^{k+d'+2+d_4}} \bigr) \|\chi v\|_{R^\alpha\rho^{-N+(j-1)\beta}L^2}\Bigr).
  \end{align*}
  For the third term, we note that $[L_{h_0,\tilde h},\chi]\in\rho\Diffb^2+R^\alpha\rho^\beta\Hb^\infty\Diffb^2$, with the second summand obeying tame estimates by~\eqref{EqE1StructTame}. Altogether, we therefore obtain
  \begin{equation}
  \label{EqE2Decay1}
  \begin{split}
    &\|I_{g_{(0)}}(\chi v)\|_{R^\alpha\rho^{-N+j\beta}\Hb^{k+d'-2 j}} \\
    &\qquad \leq C_j\Bigl( \|f\|_{R^\alpha\rho^\beta\Hb^{k+d'-2 j}} + \|v\|_{R^\alpha\rho^{-N+(j-1)\beta}\Hb^{k+d'-2(j-1)}} \\
    &\qquad\hspace{4em} + \bigl( \|h_0\|_{R^\alpha\Hb^{k+d'+2+d_3}}+\|\tilde h\|_{R^\alpha\rho^\beta\Hb^{k+d'+2+d_4}} \bigr) \|v\|_{R^\alpha\rho^{-N+(j-1)\beta}L^2}\Bigr).
  \end{split}
  \end{equation}

  We can now apply Lemma~\ref{LemmaE2Inv}\eqref{ItE2InvGrow} and deduce that $\chi v\in R^\alpha\rho^{-N+j\beta}\Hb^{k+d'-2 j}$, with norm bounded by the right hand side of~\eqref{EqE2Decay1} but with $L^2$ replaced by $\Hb^2$ (arising from the low regularity term in~\eqref{EqE2InvGrowEst} which we estimate using~\eqref{EqE2Decay1} with $k=0$, $d'=2 j$). Since $(1-\chi)v\in R^\alpha\rho^{N'}\Hb^{k+d'-2(j-1)}$ for all $N'$, we conclude that
  \begin{equation}
  \label{EqE2Decay2}
  \begin{split}
    &\|v\|_{R^\alpha\rho^{-N+j\beta}\Hb^{k+d'-2 j}} \\
    &\qquad \leq C_j\Bigl( \|f\|_{R^\alpha\rho^\beta\Hb^{k+d'-2 j}} + \|v\|_{R^\alpha\rho^{-N+(j-1)\beta}\Hb^{k+d'-2(j-1)}} \\
    &\qquad\hspace{4em} + \bigl( \|h_0\|_{R^\alpha\Hb^{k+d'+2+d_3}}+\|\tilde h\|_{R^\alpha\rho^\beta\Hb^{k+d'+2+d_4}} \bigr) \|v\|_{R^\alpha\rho^{-N+(j-1)\beta}\Hb^2}\Bigr).
  \end{split}
  \end{equation}
  As a special case, for $j\geq 1$ we take $k=0$ and $d'=2 j+q$ to get the low regularity estimate
  \begin{equation}
  \label{EqE2Decay2Lo}
    \|v\|_{R^\alpha\rho^{-N+j\beta}\Hb^q} \leq C\Bigl( \|f\|_{R^\alpha\rho^\beta\Hb^q} + \bigl(1+\|h_0\|_{R^\alpha\Hb^{2 j+d_3+q}}+\|\tilde h\|_{R^\alpha\Hb^{2 j+d_4+q}}\bigr)\|v\|_{R^\alpha\rho^{-N+(j-1)\beta}\Hb^{q+2}} \Bigr).
  \end{equation}

  Consider now~\eqref{EqE2Decay2} for $j=J$ (which requires taking $d'\geq 2 J$). The high regularity norm (the second term on the right) will be bounded using the estimate for $j=J-1$. The low regularity norm (the norm on $v$ in the third term on the right) on the other hand can be bounded using~\eqref{EqE2Decay2Lo} with $j=J-1$, $q=2$ in terms of $\|v\|_{R^\alpha\rho^{-N+(J-1)\beta}\Hb^4}$, which again using~\eqref{EqE2Decay2Lo} with $j=J-2$, $q=4$ is bounded by $\|v\|_{R^\alpha\rho^{-N+(J-2)\beta}\Hb^6}$, and so on, until after $J$ such applications of~\eqref{EqE2Decay2Lo} we obtain a bound by $\|v\|_{R^\alpha\rho^{-N}\Hb^{2 J+2}}$, which in turn we bound using~\eqref{EqE2Decay0} with $d'=0$, $k=2 J+2$; in this last step, we use $2 J+2 d_4+4$ many derivatives on $\tilde h$. Altogether then, we have thus proved
  \begin{align*}
    \|v\|_{R^\alpha\rho^{-N+J\beta}\Hb^{k+d'-2 J}} &\leq C_k\Bigl( \|(f,v_0,v_1)\|_{D^{k+d',\alpha,\beta}} \\
    &\qquad\quad + \bigl(\|h_0\|_{R^\alpha\Hb^{k+d'+2+d_3}} + \|\tilde h\|_{R^\alpha\rho^\beta\Hb^{k+d'+2+d_4}}\bigr)\|(f,v_0,v_1)\|_{D^{0,\alpha,\beta}}\Bigr),
  \end{align*}
  where $d'\geq 2 J$ is any fixed integer; here $C_k$ is a constant which is allowed to depend on the low regularity norms $\|h_0\|_{R^\alpha\Hb^{2 J+4+2 d_3}}$, $\|\tilde h\|_{R^\alpha\rho^\beta\Hb^{2 J+4+2 d_4}}$. (In the statement of the Proposition, we shall thus in particular take $d\geq 2 J+4+2 d_4=2 J+10$.)

  \pfstep{Step~2. Leading order term and decaying remainder.} The estimate~\eqref{EqE2Decay1} remains valid for $j=J+1$, in which case we get an estimate for $I_{g_{(0)}}(\chi v)$ in a decaying space since $-N+(J+1)\beta\in(0,\beta)$. We can thus apply Lemma~\ref{LemmaE2Inv}\eqref{ItE2InvDecay} in the form~\eqref{EqE2InvDecay2} to show that
  \[
    v = \chi v_0 + \tilde v,
  \]
  where $v_0\in R^\alpha\Hb^{k+d'-2(J+1)}(\cI^+_{R_0};\tau^{-2}\ker\tr_{g_{(0)}})$ and $\tilde v\in R^\alpha\rho^{-N+(J+1)\beta}\Hb^{k+d'-2(J+1)-1}$, with norms obeying tame estimates. (Here we require $d'\geq 2(J+1)+1$.) Plugging this expression for $v$ into~\eqref{EqE2DecayInd}, the right hand side lies in $R^\alpha\rho^\beta\Hb^{k+d'-2(J+2)-1}$ (with tame estimates), where for the control of the first term we use~\eqref{Eq0bSobProd3}. Applying Lemma~\ref{LemmaE2Inv}\eqref{ItE2InvDecay} yet again thus shows that $\tilde v\in R^\alpha\rho^\beta\Hb^{k+d'-2(J+2)-2}$. Setting $d'=2(J+2)+2$ and $d=(2(J+2)+2)+4+2 d_4=2 J+16$, the proof is now complete.
\end{proof}

\subsection{Solution of the gauge-fixed Einstein equations: proof of Theorem~\ref{ThmESol}}
\label{SsEPf}

We begin by explaining how Proposition~\ref{PropE2Decay} fits into a solution scheme for the gauge-fixed Einstein equation. Recalling~\eqref{EqEEinOp}, we thus consider
\begin{align*}
  &P(h_0,\tilde h,\theta) = 2\Bigl(\Ric(g)-\Lambda g - \tilde\delta^*_g\bigl(\Ups(g;g_0)+E_{g_0}(g-g_0)-\tilde\chi\theta\bigr) \Bigr), \\
  &\qquad g := g_\sfb+\chi h_0+\tilde h,\quad g_0:=g_\sfb+\chi h_0,
\end{align*}
where $\chi,\tilde\chi$ are as in~\eqref{EqECutoffs}, and $h_0,\tilde h,\theta$ are as in~\eqref{EqESolh}. Let us write $D_2|_{\tilde h}P(h_0,v,\theta):=\frac{\dd}{\dd s}P(h_0,\tilde h+s v,\theta)$, similarly $D_{1,2}|_{\tilde h_0,\tilde h}P(v_0,\tilde v,\theta):=\frac{\dd}{\dd s}P(h_0+s v_0,\tilde h+s\tilde v,\theta)$ (which is the sum of $D_1|_{h_0}P(v_0,\tilde h,\theta)$ and $D_2|_{\tilde h}P(h_0,\tilde v,\theta)$), and so on. Then Proposition~\ref{PropE2Decay} shows that the solution of the initial value problem
\begin{equation}
\label{EqEPfLin2}
  D_2|_{\tilde h}P(h_0,v,\theta) = f
\end{equation}
with $f\in R^\alpha\rho^\beta\Hb^\infty(\Omega_{\rho_0,R_0})$ and initial data in $R^\alpha\Hb^\infty(\Sigma_{\rho_0,R_0})$ can be written as $v=\chi v_0+\tilde v$ where $v_0\in R^\alpha\Hb^\infty(\cI^+_{R_0};\upbeta^*(\tau^{-2}\ker\tr_{g_{(0)}}))$ and $\tilde v\in R^\alpha\rho^\beta\Hb^\infty(\Omega_{\rho_0,R_0})$; such an initial value problem will arise in a nonlinear iteration scheme. Since we require $\tilde h\in R^\alpha\rho^\beta\Hb^\infty$, the metric perturbation $v$ is not an acceptable correction to $\tilde h$; we thus need to rewrite~\eqref{EqEPfLin2} so that the only arguments of the linearization of $P$ are elements of the same spaces as $h_0,\tilde h,\theta$. To this end, we note that
\begin{align*}
  \frac12 D_1|_{h_0}P(v_0,\tilde h,\theta) &= D_g\Ric(\chi v_0) - \Lambda \chi v_0 - (D_g\tilde\delta^*_\cdot)(\chi v_0)\bigl(\Ups(g;g_0)+E_{g_0}(g-g_0)-\tilde\chi\theta\bigr) \\
    &\quad\hspace{1em} - \tilde\delta^*_g\Bigl( D_1|_g\Ups(\chi v_0;g_0) + D_2|_{g_0}\Ups(g;\chi v_0) + (D_{g_0}E_\cdot)(\chi v_0)(g-g_0) \Bigr), \\
  \frac12 D_2|_{\tilde h}P(h_0,\tilde v,\theta) &= D_g\Ric(\tilde v) - \Lambda\tilde v - (D_g\tilde\delta^*_\cdot)(\tilde v)\bigl(\Ups(g;g_0)+E_{g_0}(g-g_0)-\tilde\chi\theta\bigr) \\
    &\quad\hspace{1em} - \tilde\delta^*_g\bigl( D_1|_g\Ups(\tilde v;g_0) + E_{g_0}\tilde v\bigr), \\
  \frac12 D_3|_\theta P(h_0,\tilde h,\dot\theta) &= -\tilde\delta^*_g(\tilde\chi\dot\theta).
\end{align*}
Therefore,~\eqref{EqEPfLin2} is equivalent to
\begin{align}
\label{EqPfLin123}
  f &= D_{1,2}|_{\tilde h,h_0}P(v_0,\tilde v,\theta) + 2\tilde\delta_g^*\bigl(D_2|_{g_0}\Ups(g;\chi v_0) - E_{g_0}(\chi v_0) + (D_{g_0}E_\cdot)(\chi v_0)\tilde h \bigr) \\
\label{EqPfLin123theta}
    &= D_{\tilde h,h_0,\theta}P(v_0,\tilde v,\dot\theta), \qquad \dot\theta = -D_2|_{g_0}\Ups(g;\chi v_0) + E_{g_0}(\chi v_0) - (D_{g_0}E_\cdot)(\chi v_0)\tilde h.
\end{align}
Note here that for this definition of $\dot\theta$, we have $\supp\dot\theta\subset\supp\chi$ and thus $\tilde\chi\dot\theta=\dot\theta$. We make the following important observation regarding the size of $\dot\theta$.

\begin{lemma}[Bounds on the gauge modification]
\label{LemmaPfGaugeMod}
  Let $h_0\in R^\alpha\Hb^\infty(\cI^+_{R_0};\upbeta^*(S^2 T^*X))$ and $\tilde h\in R^\alpha\rho^\beta\Hb^\infty(\Omega_{\rho_0,R_0};\upbeta^*(S^2\,{}^0 T^*M))$; suppose that $\|h_0\|_{R^\alpha\Hb^{d_3+2}}<\delta_0$ and $\|\tilde h\|_{R^\alpha\rho^\beta\Hb^{d_4+2}}<\delta_0$ for some small $\delta_0>0$. Define $g_{(0)},h_{(0)}$ by~\eqref{EqE1Structg0h0}, and suppose that
  \[
    v_0\in R^\alpha\Hb^\infty\bigl(\cI^+_{R_0};\upbeta^*(\tau^{-2}\ker\tr_{g_{(0)}})\bigr).
  \]
  Define $\dot\theta$ by~\eqref{EqPfLin123theta}. Then $\dot\theta\in R^\alpha\rho^\beta\Hb^\infty(\Omega_{\rho_0,R_0};\upbeta^*({}^0 T^*M))$, and we have a tame estimate
  \begin{equation}
  \label{EqPfGaugeModTame}
    \|\dot\theta\|_{R^\alpha\rho^\beta\Hb^k} \leq C_k\Bigl( \|v_0\|_{R^\alpha\Hb^{k+1}} + \bigl(\|h_0\|_{R^\alpha\Hb^{k+1+d_3}}+\|\tilde h\|_{R^\alpha\rho^\beta\Hb^{k+1+d_4}}\bigr) \|v_0\|_{R^\alpha L^2}\Bigr).
  \end{equation}
\end{lemma}
\begin{proof}
  We write
  \begin{equation}
  \label{EqPfGaugeModUps}
    \dot\theta = -D_{(g;g_0)}\Ups(\chi v_0;\chi v_0) + D_1|_g\Ups(\chi v_0;g_0) + E_{g_0}(\chi v_0) - (D_{g_0}E_\cdot)(\chi v_0)\tilde h.
  \end{equation}
  The first summand is equal to minus
  \[
    \frac{\dd}{\dd s}\Ups\bigl(g_0+\chi(h_0+s v_0)+\tilde h; g_0+\chi(h_0+s v_0) \bigr)\big|_{s=0}.
  \]
  The metrics in both arguments of $\Ups(\cdot;\cdot)$ agree up to the term $\tilde h\in R^\alpha\rho^\beta\Hb^\infty$, and thus the same holds true for the inverse metrics. Evaluating the Christoffel symbols of the two metrics in the frame $e_\mu$ as in the proof of Proposition~\ref{PropE1Struct} and using the expression for $\Ups$ given in~\eqref{EqEGauge} implies that $-D_{(g;g_0)}\Ups(\chi v_0;\chi v_0)\in R^\alpha\rho^\beta\Hb^\infty$, and it satisfies the tame bound~\eqref{EqPfGaugeModTame}.

  We evaluate the second and third summand of~\eqref{EqPfGaugeModUps} using the formula~\eqref{EqE1D1Ups}. Write $v_{(0)}=\tau^2 v_0=v_0(\tau\pa_\mu,\tau\pa_\nu)\,\dd x^\mu\,\dd x^\nu$. The indicial operator of $D_1|_g\Ups(\cdot;g_0)+E_{g_0}$ can be computed using~\eqref{EqE1StructGg} and \eqref{EqE1StructBoxg} to be
  \begin{equation}
  \label{EqPfGaugeModPf}
    -\begin{pmatrix} e_0-3 & 0 & -\tr_{g_{(0)}} \\ 0 & e_0-4 & 0 \end{pmatrix}\begin{pmatrix} \frac12 & 0 & \frac12\tr_{g_{(0)}} \\ 0 & I & 0 \\ \frac12 g_{(0)} & 0 & I-\frac12 g_{(0)}\tr_{g_{(0)}} \end{pmatrix} + \begin{pmatrix} 1 & 0 & -2\tr_{g_{(0)}} \\ 0 & 0 & 0 \end{pmatrix};
  \end{equation}
  the key point is that this annihilates $(0,0,v_{(0)})$ since $\tr_{g_{(0)}}v_{(0)}=0$. The difference between $D_1|_g\Ups(\cdot;g_0)+E_{g_0}$ and its indicial operator is of class $(\rho\CI+R^\alpha\rho^\beta\Hb^\infty)\Diffb^1$ (with tame estimates for its coefficients), and thus the sum of the second and third summands of~\eqref{EqPfGaugeModUps} lies in $R^\alpha\rho^\beta\Hb^\infty$ indeed.

  Finally, the fourth summand of~\eqref{EqPfGaugeModUps} is of class $R^\alpha\rho^\beta\Hb^\infty$ as a consequence of $\tilde h\in R^\alpha\rho^\beta\Hb^\infty$. This proves~\eqref{EqPfGaugeModTame}.
\end{proof}

\begin{rmk}[Origin of the decay of $\dot\theta$]
\label{RmkPfGaugeModWhy}
  In view of the formula~\eqref{EqPfGaugeModUps}, the fact that $\dot\theta$ decays (even though the input $\chi v_0$ in~\eqref{EqPfLin123} does not) has an entirely conceptual explanation. First of all, the vanishing of the first and fourth terms in~\eqref{EqPfGaugeModUps} is \emph{automatic} since the solution metric and the background metric (i.e.\ the first and second argument of $\Ups$, and the argument of $E_{g_0}$) are changed in lockstep. The second ingredient is the fact that the indicial operator of $D_1|_g\Ups+E_{g_0}$ annihilates $v_0$. \emph{This is also automatic} by the following reasoning: the term $v_0$ arises as an indicial solution of $L_{h_0,\tilde h}$, corresponding to the indicial root $0$. By virtue of constraint damping, as discussed after~\eqref{EqE1AuxCD}, it must therefore necessarily satisfy the linearized gauge condition on the indicial operator level, i.e.\ $I(D_1|_g\Ups+E_{g_0},0)v_0=0$. (The computation~\eqref{EqPfGaugeModPf} merely verifies this through an explicit computation.)
\end{rmk}

We can now prove the main result of this section.

\begin{proof}[Proof of Theorem~\usref{ThmESol}]
  Let $d\in\N$ be as in Proposition~\ref{PropE2Decay}. For $k\in\N_0$, we define the spaces
  \begin{align*}
    \bfB^k&:=R^\alpha\Hb^k(\cI^+_{R_0};\upbeta^*(\tau^{-2}S^2 T^*X)) \oplus R^\alpha\rho^\beta\Hb^k\bigl(\Omega_{\rho_0,R_0};\upbeta^*(S^2\,{}^0 T^*M)\bigr) \\
      &\quad \hspace{16em} \oplus R^\alpha\rho^\beta\Hb^k\bigl(\Omega_{\rho_0,R_0};\upbeta^*({}^0 T^*X)\bigr), \\
    B^k&:=R^\alpha\rho^\beta\Hb^k\bigl(\Omega_{\rho_0,R_0};\upbeta^*(S^2\,{}^0 T^*M)\bigr) \\
      &\quad\hspace{6em} \oplus R^\alpha\Hb^k\bigl(\Sigma_{\rho_0,R_0};\upbeta^*(S^2\,{}^0 T^*M)\bigr) \oplus R^\alpha\Hb^k\bigl(\Sigma_{\rho_0,R_0};\upbeta^*(S^2\,{}^0 T^*M)\bigr).
  \end{align*}
  For $k\geq d$ and $(h_0,\tilde h,\theta)\in\bfB^\infty$ with $\|(h_0,\tilde h,\theta)\|_{\bfB^{3 d}}<\delta_0$ for $\delta_0>0$ sufficiently small, set
  \[
    \Phi(h_0,\tilde h,\theta) := \bigl( P(h_0,\tilde h,\theta),\ \tilde h|_{\Sigma_{\rho_0,R_0}},\ \cL_{-\rho\pa_\rho}\tilde h|_{\Sigma_{\rho_0,R_0}}\bigr) - (0,\ubar h_0,\ubar h_1),
  \]
  which is a map from a subset of $\bfB^\infty$ to $B^\infty$. The task is then to solve
  \begin{equation}
  \label{EqEStabEq}
    \Phi(h_0,\tilde h,\theta)=0.
  \end{equation}
  We accomplish this by applying the main result of \cite{SaintRaymondNashMoser} (see also \cite[Remark on p.~220]{SaintRaymondNashMoser}). The required tame estimates for $\Phi$ follow from Lemma~\ref{LemmaE1Map} and, for the initial data part of $\Phi$, from Lemma~\ref{Lemma0bSob}\eqref{It0bSobTrace}. The required low regularity estimates for the linearization and second derivative of $\Phi$ are straightforward consequences of the algebra properties of $\Hb^d$. The right inverse of $D_{h_0,\tilde h,\theta}\Phi$ is constructed using Proposition~\ref{PropE2Decay} via~\eqref{EqPfLin123}--\eqref{EqPfLin123theta} (using Lemma~\ref{LemmaPfGaugeMod}): this produces, for $f\in R^\alpha\rho^\beta\Hb^\infty(\Omega_{\rho_0,R_0})$ and $v_0,v_1\in R^\alpha\Hb^\infty(\Sigma_{\rho_0,R_0})$, tensors $v_0\in R^\alpha\Hb^\infty$, $\tilde v\in R^\alpha\rho^\beta\Hb^\infty$, and a 1-form $\dot\theta\in R^\alpha\rho^\beta\Hb^\infty$ so that
  \[
    D_{h_0,\tilde h,\theta}\Phi(v_0,\tilde v,\dot\theta) = \bigl(D_{h_0,\tilde h,\theta}P(v_0,\tilde v,\dot\theta), \tilde v|_{\Sigma_{\rho_0,R_0}}, \cL_{-\rho\pa_\rho}\tilde v|_{\Sigma_{\rho_0,R_0}}\bigr) = (f,v_0,v_1);
  \]
  and $v_0,\tilde v,\dot\theta$ satisfy tame estimates. Using the smoothing operators from Lemma~\ref{Lemma0bSmooth}, applied component-wise to $v_0,\tilde v,\dot\theta$ in the bundle splitting induced by the frame $e_\mu$, we can thus apply the Nash--Moser iteration of \cite{SaintRaymondNashMoser} to solve the desired equation~\eqref{EqEStabEq}.
\end{proof}

\section{Nonlinear stability; smoothness at the conformal boundary}
\label{SSt}

We continue working on the domains $\Omega_{\rho_0,R_0}$ from Definition~\ref{Def0bDom}. In this section, we will work with the gauge-fixed Einstein operator $P=P(h_0,\tilde h,\theta)$ from~\eqref{EqEEinOp} to solve initial value problems for the Einstein vacuum equations~\eqref{EqIEin}; the metric solving the Einstein equations will be $g=g_\sfb+\chi h_0+\tilde h$, with the KdS metric $g_\sfb$ defined by Lemma~\ref{Lemma0bStruct}, and with $\chi=\chi(\rho)\in\CIc([0,\frac12\rho_0))$ equal to $1$ on $[0,\frac14\rho_0]$ as in~\eqref{EqECutoffs}. We remark that in this section, we use Theorem~\ref{ThmESol} as a black box; our goal is to get sharper control on the solution $g$ given by this Theorem via careful modifications of the gauge.

Note that given a solution $h_0,\tilde h,\theta$ of the initial value problem for $P(h_0,\tilde h,\theta)=0$ as in Theorem~\ref{ThmESol}, the fact that $\theta$ lies in $R^\alpha\rho^\beta\Hb^\infty$ and in particular does not necessarily have an asymptotic expansion at $\cI^+$ means that also $\tilde h$ does not necessarily have an asymptotic expansion. When, however, $g=g_\sfb+\chi h_0+\tilde h$ solves the Einstein vacuum equations $\Ric(g)-\Lambda g=0$ (in the gauge $\Ups(g;g_0)+E_{g_0}(g-g_0)-\tilde\chi\theta=0$), we will begin by demonstrating how to exploit the diffeomorphism invariance of these equations and pull back $g$ by suitable diffeomorphisms to put it into the same type of gauge but now with $\theta$ vanishing to infinite order at $\rho=0$ (see Proposition~\ref{PropStGauge} below). In this new gauge then, we can prove that $g$ \emph{is} log-smooth at $\cI^+$ (Lemma~\ref{LemmaStLog}). We can then eliminate all logarithms via further pullbacks (Proposition~\ref{PropStCI}). For a comparison with the method of proof of \cite{ChruscielDelayLeeSkinnerConfCompReg} in the Riemannian setting, see Remarks~\ref{RmkStRiem1} and \ref{RmkStRiem2} below.

The initial data of $g$ at the Cauchy hypersurface $\Sigma_{\rho_0,R_0}$ of $\Omega_{\rho_0,R_0}$ are its first and second fundamental form, denoted $\gamma$ and $k$, respectively. To capture the behavior of $\gamma,k$ uniformly down to the boundary at infinity $\cK\cap\Sigma_{\rho_0,R_0}$ of $\Sigma_{\rho_0,R_0}$, we work with the b-cotangent bundle of $\Sigma_{\rho_0,R_0}$, which is
\[
  \Tb^*\Sigma_{\rho_0,R_0} = \R\frac{\dd R}{R}\oplus T^*\Sph^2.
\]

\begin{lemma}[Initial data of $g_\sfb$]
\label{LemmaStData}
  Denote by $\gamma_\sfb$ and $k_\sfb$ the first and second fundamental form of $g_\sfb$ at $\Sigma_{\rho_0,R_0}$, respectively. Then
  \[
    \gamma_\sfb,\ k_\sfb \in \CI(\Sigma_{\rho_0,R_0};S^2\,\Tb^*\Sigma_{\rho_0,R_0}),
  \]
  with $\gamma_\sfb$ positive definite.
\end{lemma}
\begin{proof}
  Recall that $g_\sfb\in\CI(\breve M;S^2\,{}^{0,\bop}T^*\breve M)$ is a Lorentzian signature section over $\Omega_{\rho_0,R_0}$. Near $\Sigma_{\rho_0,R_0}$, where $\rho$ is bounded away from $0$, local frames of ${}^{0,\bop}T^*\breve M$ are given by $\dd\rho$, $\frac{\dd R}{R}$, and a frame of $T^*\Sph^2$. Since $\dd\rho=\dd(\tilde r^{-1})$ is timelike for $g_\sfb$, with $0>g_\sfb^{-1}(\dd\rho,\dd\rho)\in\CI(\Sigma_{\rho_0,R_0})$, the future unit normal vector to $\Sigma_{\rho_0,R_0}$ is an element of $\CI(\Sigma_{\rho_0,R_0};{}^{0,\bop}T\breve M)$ with negative squared length. This implies the statement about $\gamma_\sfb$. The membership of $k_\sfb$ now follows, for example, from the smoothness of the Christoffel symbols in the (smooth) frame~\eqref{EqE1StructFrame} of ${}^{0,\bop}T\breve M$.
\end{proof}

Proposition~\ref{PropStCI}\eqref{ItStCICI}, proved in~\S\ref{SsStCI} (which is self-contained, i.e.\ does not rely on any other results proved here), produces a smooth diffeomorphism $\phi$ of $\breve M$ which preserves $\cI^+$ pointwise and maps $\cK$ to itself so that
\begin{equation}
\label{EqEStabFG}
  g_\sfb^{\rm FG} := \phi^*g_\sfb
\end{equation}
is in Fefferman--Graham form (see the explanation after Theorem~\ref{ThmIMain}) and still satisfies $g_\sfb^{\rm FG}-g_\dS\in\rho^3\CI(\Omega_{\rho_0,R_0};\upbeta^*(S^2\,{}^0 T^*M))$. We shall then prove:

\begin{thm}[Nonlinear stability of the cosmological region]
\label{ThmEStabEin}
  Let $R_0>0$ and $\rho_0\in(0,\bar\rho]$ in the notation of Definitions~\usref{Def0bMfd} and \usref{Def0bDom}. Let $\alpha>0$ and $d=d_4+2$ (or more generally $d\in\N$). Then there exists $D\in\N$ so that the following holds. For all $\beta\in(0,1)$ and $\delta_0>0$, there exists an $\eps>0$ so that if
  \[
    \gamma = \gamma_\sfb + \tilde\gamma,\quad
    k = k_\sfb + \tilde k,\qquad
    \tilde\gamma,\ \tilde k \in R^\alpha\Hb^\infty(\Sigma_{\rho_0,R_0};S^2\,\Tb^*\Sigma_{\rho_0,R_0}),
  \]
  with $\|\tilde\gamma\|_{R^\alpha\Hb^D}<\eps$ and $\|\tilde k\|_{R^\alpha\Hb^D}<\eps$, and with $(\gamma,k)$ satisfying the constraint equations, the maximal globally hyperbolic development of the initial data $\gamma,k$ contains a region isometric to
  \[
    (\Omega_{\rho_0,R_0},g),\qquad g=g_\sfb^{\rm FG}+h,
  \]
  where $h$ is as follows. There exist $h_{0,\mu\nu},\tilde h_{i,\mu\nu}\in R^\alpha\Hb^\infty(\cI^+_{R_0})$, $i=2,3,\ldots$, so that, in the frame $\{e_\mu\}=\{\tau\pa_\tau,\tau\pa_{x^1},\tau\pa_{x^2},\tau\pa_{x^3}\}$:
  \begin{enumerate}
  \item $h_0$ and $\tilde h_i$ are tangential-tangential tensors, i.e.\ $h_{0,\mu\nu}=\tilde h_{i,\mu\nu}=0$ unless $\mu,\nu\geq 1$;
  \item for all $N\in\N$, we have
    \[
      h_{\mu\nu}(\rho,R,\omega) - h_{0,\mu\nu}(R,\omega) - \sum_{i=2}^N \rho^i\tilde h_{i,\mu\nu}(R,\omega) \in R^\alpha\rho^N\Hb^\infty\bigl(\Omega_{\rho_0,R_0};\upbeta^*(S^2\,{}^0 T^*M)\bigr);
    \]
  \item write $h_{(0)}:=h_{0,i j}\dd x^i\otimes_s\dd x^j$. Then $\tilde h_2\neq 0$ unless $g_{(0)}:=\dd x^2+h_{(0)}$ is flat;
  \item $h_{(3)}:=\tilde h_{3,i j}\dd x^i\otimes_s\dd x^j$ is weighted transverse traceless, i.e.
    \begin{equation}
    \label{EqEStabEinTT}
      \tr_{g_{(0)}}h_{(3)}=\delta_{g_{(0)}}(|x|^{-3}h_{(3)})=0;
    \end{equation}
  \item\label{ItEStabEinSmall} $\|h_{0,\mu\nu}\|_{R^\alpha\Hb^d}<\delta_0$ and $\|h_{\mu\nu}-\chi h_{0,\mu\nu}\|_{R^\alpha\rho^\beta\Hb^d}<\delta_0$.
  \end{enumerate}
\end{thm}

The smallness of the low regularity norms in part~\eqref{ItEStabEinSmall} implies, by Sobolev embedding, that the geometry of $(\Omega_{\rho_0,R_0},g)$ is qualitatively the same as that of $(\Omega_{\rho_0,R_0},g_\sfb)$, so in particular the boundary hypersurfaces $\Sigma_{\rho_0,R_0}$ and $\Sigma^+_{\rho_0,R_0}$ are spacelike. As already remarked in~\S\ref{SI}, the weight $|x|^{-3}$ in~\eqref{EqEStabEinTT} is due to the fact that we Taylor expand not in the defining function $\tau$ of the conformal boundary of de~Sitter space, but in the defining function $\rho=\tau|x|^{-1}$ of the conformal boundary $\cI^+\subset\breve M$ of Kerr--de~Sitter.

\begin{proof}[Proof of Theorem~\usref{ThmEStabEin}]
  Write $\Sigma:=\Sigma_{\rho_0,R_0}$.

  \pfstep{Step~1. Construction of Cauchy data.} We need to prove the existence of
  \[
    \ubar h_0,\ \ubar h_1 \in R^\alpha\Hb^\infty\bigl(\Sigma;\upbeta^*(S^2\,{}^0 T^*M)\bigr),
  \]
  with $R^\alpha\Hb^{D-1}$-norms bounded by a constant times $\eps$, so that for a metric $g$ with $(g-g_\sfb,\cL_{-\rho\pa_\rho}(g-g_\sfb))|_\Sigma=(\ubar h_0,\ubar h_1)$, the first and second fundamental form of $g$ at $\Sigma$ are $\gamma$ and $k$, respectively, and furthermore the gauge condition $\Ups(g;g_\sfb)=0$ is satisfied at $\Sigma$.

  We proceed as in the proof of \cite[Proposition~3.10]{HintzVasyKdSStability}. We work in a product neighborhood $(-\frac12\rho_0,\frac12\rho_0)_\ft\times\Sigma$ of $\Sigma$ where $\ft:=\rho_0-\rho$; we can thus regard $S^2\Tb^*\Sigma$ as a subbundle of $S^2\,{}^{0,\bop}T^*_\Sigma\breve M$. Writing $g_\sfb=a\,\dd\ft^2+2\,\dd\ft\otimes_s b+\gamma_\sfb$ where $a\in\CI(\Sigma)$ and $b\in\CI(\Sigma;T^*\Sigma)$, we then set
  \[
    \ubar h_0 := \tilde\gamma \in R^\alpha\Hb^\infty(\Sigma;S^2\,{}^{0,\bop}T^*_\Sigma\breve M).
  \]
  Thus $g_0:=g_\sfb+\ubar h_0=a\,\dd\ft^2+2\,\dd\ft\otimes_s b+\gamma$ has first fundamental form $\gamma$ at $\ft=0$ indeed. If $\nu_\sfb$ and $\nu$ denote the future unit normals at $\Sigma$ for $g_\sfb$ and $g_0$, respectively, we then have $\nu=\nu_\sfb+\tilde\nu$ where $\tilde\nu\in R^\alpha\Hb^\infty(\Sigma;{}^{0,\bop}T_\Sigma\breve M)$. To match the desired second fundamental form, we require for all $X,Y\in\Tb\Sigma$ the equality
  \begin{align*}
    k(X,Y) &\stackrel{!}{=} g(\nabla^g_X Y,\nu) \\
      &= g_\sfb(\nabla_X^{g_\sfb}Y,\nu_\sfb) + \ubar h_0(\nabla^{g_\sfb}_X Y,\nu_\sfb) + g(\nabla_X^{g_\sfb}Y,\tilde\nu) + g\bigl((\nabla^g_X-\nabla^{g_\sfb}_X)Y,\nu\bigr),
  \end{align*}
  where we write $g:=g_\sfb+\ubar h_0+\ft \ubar h_1=g_0+\ft \ubar h_1$; equivalently,
  \begin{equation}
  \label{EqEStabEink}
    g_0\bigl( (\nabla_X^{g_0+\ft \ubar h_1}-\nabla_X^{g_0})Y, \nu\bigr) = \tilde k(X,Y) - \ubar h_0(\nabla^{g_\sfb}_X Y,\nu_\sfb) - g_0(\nabla^{g_\sfb}_X Y,\tilde\nu) - g_0\bigl((\nabla_X^{g_0}-\nabla_X^{g_\sfb})Y,\nu\bigr).
  \end{equation}
  Since $\tilde k,\ubar h_0\in R^\alpha\Hb^\infty$, the right hand side of this equation is the evaluation of an element of $R^\alpha\Hb^\infty(\Sigma;S^2\,\Tb^*\Sigma)$ on $(X,Y)$, which is moreover small in $R^\alpha\Hb^{D-1}$ due to the smallness assumption on $\tilde\gamma\in R^\alpha\Hb^D$. To ensure the desired gauge condition, we need
  \begin{equation}
  \label{EqEStabEinUps}
    \tr_{g_0} ( \nabla^{g_0+\ft \ubar h_1} - \nabla^{g_0} ) = 0.
  \end{equation}
  In any frame of ${}^{0,\bop}T\breve M$, we compute at $\ft=0$
  \[
    2\bigl(\Gamma(g_0+\ft \ubar h_1)_{\mu\nu}^\lambda - \Gamma(g_0)_{\mu\nu}^\lambda\bigr) = \ft_{;\mu}(\ubar h_1)_\nu{}^\lambda + \ft_{;\nu}(\ubar h_1)_\mu{}^\lambda - \ft^{;\lambda}(\ubar h_1)_{\mu\nu},
  \]
  where the indices are raised using $g_0$.

  Let us take a frame with $e_0=\nu$, while the $e_i$, $i=1,2,3$, span $e_0^\perp=\Tb\Sigma$; then $\ft_{;\mu}=0$ unless $\mu=0$, in which case $\ft_{;0}=e_0\ft>0$ is of class $\CI+R^\alpha\Hb^\infty$ and bounded away from zero. The left hand side of~\eqref{EqEStabEink}, for $X=e_i$ and $Y=e_j$, equates to $\frac12(e_0\ft)(\ubar h_1)_{i j}$; this thus uniquely determines $(\ubar h_1)_{i j}\in R^\alpha\Hb^\infty$. The $e_j$-component of the gauge condition~\eqref{EqEStabEinUps} reads
  \[
    0 = 2\ft^{;\mu}(\ubar h_1)_{\mu j} = 2\ft^{;0}(\ubar h_1)_{0 j} + 2\ft^{;i}(\ubar h_1)_{i j}.
  \]
  This uniquely determines $(\ubar h_1)_{0 j}\in R^\alpha\Hb^\infty$. The $e_0$-component of~\eqref{EqEStabEinUps} finally reads
  \[
    0 = 2\ft^{;\mu}(\ubar h_1)_{\mu 0} - \ft_{;0}(\ubar h_1)_\mu{}^\mu = -2\ft_{;0}(\ubar h_1)_{0 0} + 2\ft^{;i}(\ubar h_1)_{i 0} + \bigl(\ft_{;0}(\ubar h_1)_{0 0} - \ft_{;0}(\ubar h_1)_i{}^i\bigr),
  \]
  and this uniquely determines $(\ubar h_1)_{0 0}\in R^\alpha\Hb^\infty$.

  \pfstep{Step~2. Solution of the gauge-fixed equation.} We now use Theorem~\ref{ThmESol} to find $h_0,\tilde h,\theta$ so that $P(h_0,\tilde h,\theta)=0$, $(\tilde h,\cL_{-\rho\pa_\rho}\tilde h)|_\Sigma=(\ubar h_0,\ubar h_1)$. Writing
  \begin{equation}
  \label{EqEStabEingg0}
    g=g_\sfb+\chi h_0+\tilde h,\qquad g_0=g_\sfb+\chi h_0,
  \end{equation}
  we thus have
  \[
    \Ric(g)-\Lambda g - \tilde\delta_g^*\eta = 0,\qquad \eta := \Ups(g;g_0) + E_g(g-g_0) - \tilde\chi\theta.
  \]
  By the support conditions on $E_g,\tilde\chi$, we have $\eta|_\Sigma=0$. Since the initial data $\gamma,k$ satisfy the constraint equations, a standard argument (see e.g.\ \cite[Chapter~6, Lemma~8.2]{ChoquetBruhatGerochMGHD}) implies that also $\cL_{-\rho\pa_\rho}\eta=0$ at $\Sigma$. But since $\eta$ satisfies the homogeneous wave type equation $\delta_g\sfG_g\tilde\delta_g^*\eta=0$, we conclude that $\eta=0$ throughout $\Omega_{\rho_0,R_0}$, and therefore $\Ric(g)=\Lambda g$ on $\Omega_{\rho_0,R_0}$ as well.

  \pfstep{Step~3. Improved asymptotics in a new gauge.} As a starting point to improve the asymptotics of the spacetime metric, it is convenient to perform Steps~1 and 2 above for $g_\sfb^{\rm FG}$ in place of $g_\sfb$ (and in particular using $g_\sfb^{\rm FG}$ in the definition of the operator $P$ in~\eqref{EqEEinOp}). Applying Proposition~\ref{PropStGauge}, Lemma~\ref{LemmaStLog}, and Proposition~\ref{PropStCI} below then produces the desired solution $h$.
\end{proof}

\subsection{Improving the gauge condition; log-smoothness}

In order to facilitate the construction of suitable pullbacks of $g$ required to complete Step~3 of the above proof, we define a class of maps between (subsets of) $\breve M$ in the coordinates $\tau\geq 0$, $x\in\R^3$ on $M$ by\footnote{In the coordinates $\rho=\frac{\tau}{|x|}$, $R=|x|$, $\omega=\frac{x}{|x|}$, the map $\phi_{a,b}$ roughly maps $(\rho,R,\omega)$ into $(\rho',R',\omega')=(\rho+\rho\cO(a),R+\rho R\cO(b),\omega+\rho\cO(b))$.}
\begin{equation}
\label{EqStPhiab}
  \phi_{a,b} \colon (\tau,x) \mapsto \bigl( \tau(1+a(\tau,x)), x+\tau b(\tau,x) \bigr).
\end{equation}
We require $a$ (real-valued) and $b$ ($\R^3$-valued) to be conormal on $\breve M$ and to decay at $\cI^+$. Concretely, we will have $a,b\in\rho^3\CI+R^\alpha\rho^\beta\Hb^\infty(\Omega_{\rho_0,R_0})$, \emph{and we always tacitly require that $a,b$ vanish near $\rho=\rho_0$}. If $a,b$ are small in $\rho^3\cC^1+R^\alpha\rho^\beta\cC_\bop^1$, then an inverse function argument shows that $\phi_{a,b}$ restricted to $\Omega_{\rho_0,R_0}\setminus(\cI^+\cup\cK)$ is a diffeomorphism onto its image (whose inverse is then again of the form $\phi_{a',b'}$ for some $a',b'$. Note that the vector field
\[
  V_{a,b} := \frac{\dd}{\dd s}\phi_{s a,s b}\Big|_{s=0} = a(\tau,x)\tau\pa_\tau + b(\tau,x)\cdot\tau\pa_x
\]
is then a section of the bundle $\upbeta^*({}^0 T^*M)={}^{0,\bop}T^*\breve M$ of class $\rho^3\CI+R^\alpha\rho^\beta\Hb^\infty$; this is the natural bundle of which generators of diffeomorphisms of the metric $g$ (which is a nondegenerate section of $\upbeta^*(S^2\,{}^0 T^*M)$) are sections of. In an iterative construction of $a,b$, we need the following result.

\begin{lemma}[Pullbacks along $\phi_{a,b}$]
\label{LemmaStPullback}
  Let $R_1\in(0,R_0)$, $\alpha,\beta>0$,and $\dot\alpha\geq\alpha$, $\dot\beta\geq\beta$. Suppose that $a,b\in R^\alpha\rho^\beta\Hb^\infty(\Omega_{\rho_0,R_1})$ and $\dot a,\dot b\in R^{\dot\alpha}\rho^{\dot\beta}\Hb^\infty(\Omega_{\rho_0,R_1})$ are small in $R^\alpha\rho^\beta\cC_\bop^1$ so that $\phi_{a,b}$ and $\phi_{a+\dot a,b+\dot b}$ map $\Omega_{\rho_0,R_1}$ into $\Omega_{\rho_0,R_0}$. Let $u\in R^{\tilde\alpha}\rho^{\tilde\beta}\Hb^\infty(\Omega_{\rho_0,R_0})$ where $\tilde\alpha,\tilde\beta\in\R$. Then
  \[
    \phi_{a+\dot a,b+\dot b}^*u - \phi_{a,b}^*u \equiv V_{\dot a,\dot b}u \bmod R^{\alpha+\dot\alpha+\tilde\alpha}\rho^{\beta+\dot\beta+\tilde\beta}\Hb^\infty(\Omega_{\rho_0,R_1}).
  \]
\end{lemma}
\begin{proof}
  This follows from a second order Taylor expansion. To wit,
  \begin{align*}
    &(\phi_{a+\dot a,b+\dot b}^*u)(\tau,x) = u(\tau(1+a+\dot a),x+\tau(b+\dot b)) \\
    &\quad= u(\tau(1+a),x+\tau b) + \dot a(\tau,x)(\tau\pa_\tau u)(\tau(1+a),x+\tau b) \\
    &\qquad \hspace{8.6em} + \dot b(\tau,x)\cdot(\tau\pa_x u)(\tau(1+a),x+\tau b) \\
    &\qquad + \sum_{|\gamma|=2} \frac{2}{\gamma!}\int_0^1 (1-s) (\tau\dot a,\tau\dot b)^\gamma ((\pa_\tau,\pa_x)^\gamma u)(\tau(1+a+s\dot a),x+\tau(b+s\dot b))\,\dd s.
  \end{align*}
  The first term on the right is $\phi_{a,b}^*u$; the last term involves quadratic expressions in $\dot a$ and $\dot b$ and thus is of class $R^{\tilde\alpha+2\dot\alpha}\rho^{\tilde\beta+2\dot\beta}\Hb^\infty$. We rewrite the second term as
  \[
    (\tau\pa_\tau u)(\tau(1+a),x+\tau b) = (\tau\pa_\tau u)(\tau,x) + \int_0^1 ((a\tau\pa_\tau+b\cdot\tau\pa_x)\tau\pa_\tau u)(\tau(1+s a),x+\tau s b)\,\dd s.
  \]
  Multiplied by $\dot a$, the integral contributes a term of class $R^{\alpha+\dot\alpha+\tilde\alpha}\rho^{\beta+\dot\beta+\tilde\beta}\Hb^\infty$. We rewrite the third term in an analogous fashion.
\end{proof}

Analogous results, with the same proof, can be obtained for other classes of $u$. Of particular relevance for us is the following instance, with $\dot\alpha=\alpha$, in the notation of~\eqref{EqEStabEingg0}:
\begin{equation}
\label{EqStPullback}
\begin{split}
  &\phi_{a+\dot a,b+\dot b}^*g - \phi_{a,b}^*g \equiv \delta_{\omega_{\dot a,\dot b}}^* g_0 \bmod R^{2\alpha}\rho^{\beta+\dot\beta}\Hb^\infty, \\
  &\hspace{10em} \omega_{\dot a,\dot b}:=2 g_0(V_{\dot a,\dot b},\cdot) \in R^\alpha\rho^{\dot\beta}\Hb^\infty\bigl(\Omega_{\rho_0,R_1};\upbeta^*({}^0 T^*M)\bigr).
\end{split}
\end{equation}
(The $\rho$-weight arises from the fact that $g$ is equal to a \emph{non-decaying} tensor $g_0$ plus a decaying correction term.)

\begin{prop}[Improving the gauge]
\label{PropStGauge}
  Suppose $h_0,\tilde h,\theta$ are as in~\eqref{EqESolh} with small weighted $\Hb^d$-norms, and let $R_1\in(0,R_0)$. Let $g=g_\sfb+\chi h_0+\tilde h$ and $g_0=g_\sfb+\chi h_0$, and suppose that $\Ric(g)-\Lambda g=0$ and $\Ups(g;g_0)+E_g(g-g_0)-\tilde\chi\theta=0$. Let $\beta^-\in(0,\beta)$ and $\eps>0$. Then there exist $a\in R^\alpha\rho^\beta\Hb^\infty(\Omega_{\rho_0,R_0})$ and $b\in R^\alpha\rho^\beta\Hb^\infty(\Omega_{\rho_0,R_0};\R^3)$ with $R^\alpha\rho^{\beta^-}\cC_\bop^1$-norms less than $\eps$ (which implies that $\phi_{a,b}$ maps $\Omega_{\rho_0,R_1}$ diffeomorphically into a subset of $\Omega_{\rho_0,R_0}$) so that, for $g':=\phi_{a,b}^*g$, we have\footnote{Of course, we have $\Ric(g')-\Lambda g'=0$ as well.}
  \begin{equation}
  \label{EqStGauge}
    \theta' := \Ups(g';g_0) + E_{g'}(g'-g_0) \in R^\alpha\rho^\infty\Hb^\infty\bigl(\Omega_{\rho_0,R_1};\upbeta^*({}^0 T^*M)\bigr).
  \end{equation}
\end{prop}

\begin{rmk}[Eliminating $\theta'$ altogether]
\label{RmkStGaugeTheta}
  While the infinite order vanishing of $\theta'$ at $\rho=0$ suffices for our purposes, one may ask whether one can choose $a,b$ so that, in fact, $\theta'=0$. We expect this to be possible by solving a suitable wave map equation backwards from $\cI^+$ using extensions of the methods of \cite{HintzAsymptoticallydS,BernhardtLinearExpanding}, but do not pursue this question further here.
\end{rmk}

\begin{rmk}[Presentation of the KdS metric]
\label{RmkStGaugeKdS}
  The same conclusions hold, by the same proof, for $g_\sfb^{\rm FG}$ in place of $g_\sfb$. Since we construct $g_\sfb^{\rm FG}$ only later, we formulate Proposition~\ref{PropStGauge} with $g_\sfb$.
\end{rmk}

\begin{proof}[Proof of Proposition~\usref{PropStGauge}]
  We may assume that $\beta$ is irrational by reducing it by an arbitrarily small amount. (This ensures that $k\beta\notin\N$ for all $k\in\N$, and thus avoids integer coincidences with indicial roots.) We will iteratively construct
  \[
    \dot a_k \in R^\alpha\rho^{(k+1)\beta}\Hb^\infty(\Omega_{\rho_0,R_0}),\qquad
    \dot b_k \in R^\alpha\rho^{(k+1)\beta}\Hb^\infty(\Omega_{\rho_0,R_0};\R^3)
  \]
  with the following properties for $\delta:=\beta-\beta^->0$ and for all $k\in\N_0$:
  \begin{enumerate}
  \item $\dot a_k$ and $\dot b_k$ are supported in $\rho<\frac12\rho_0$;
  \item $\|\dot a_k\|_{R^\alpha\rho^{(k+1)\beta-\delta}\Hb^{d_4+1+k}}<\eps 2^{-k}$ and $\|\dot b_k\|_{R^\alpha\rho^{(k+1)\beta-\delta}\Hb^{d_4+1+k}}<\eps 2^{-k}$;
  \item setting $a_k:=\sum_{i=0}^{k-1}\dot a_i$ and $b_k:=\sum_{i=0}^{k-1}\dot b_i$, as well as $g'_k:=\phi_{a_k,b_k}^*g$, we have
    \begin{equation}
    \label{EqStGaugeImpr}
      \theta'_k := \Ups_E(g'_k;g_0) := \Ups(g'_k;g_0) + E_{g'_k}(g'_k-g_0) \in R^\alpha\rho^{(k+1)\beta}\Hb^\infty(\Omega_{\rho_0,R'_k})
    \end{equation}
    where $R'_0:=R_0>R'_1>R'_2>\cdots>R_1$.
  \end{enumerate}
  The radii $R'_k$ are chosen so that $\phi_{a_k,b_k}$ maps $\Omega_{\rho_0,R'_k}$ into $\Omega_{\rho_0,R_0}$; in view of the smallness requirement on the $\dot a_k,\dot b_k$ in $\cC_\bop^1$, a sequence $R'_k$ with the required properties will indeed exist. We also note that $a_k,b_k\in R^\alpha\rho^\beta\Hb^\infty$. The domains of definition of the $\dot a_k,\dot b_k$ can be fixed to be $\Omega_{\rho_0,R_0}$ since even if they are initially defined on $\Omega_{\rho_0,R'_k}$ they can be extended, with controlled norms, using Lemma~\ref{Lemma0bExt}. We shall thus omit the specification of domains in what follows.

  Now, if for some $k\in\N_0$ the functions $\dot a_i,\dot b_i$, $i=0,\ldots,k-1$, have already been constructed, we need to find $\dot a_k,\dot b_k$ with
  \[
    \Ups_E(\phi_{a_k+\dot a_k,b_k+\dot b_k}^*g;g_0) - \Ups_E(\phi_{a_k,b_k}^*g;g_0) - \theta'_k \in R^\alpha\rho^{(k+2)\beta}\Hb^\infty.
  \]
  Using~\eqref{EqStPullback} with $\dot\beta=(k+1)\beta$, this is equal to
  \[
    \Ups_E(\phi_{a_k,b_k}^*g + \delta_{g_0}^*\omega_{\dot a_k,\dot b_k} + h_k; g_0) - \Ups_E(\phi_{a_k,b_k}^*g;g_0) - \theta'_k
  \]
  where $\delta_{\omega_{\dot a_k,\dot b_k}}^*g_0\in R^\alpha\rho^{(k+1)\beta}\Hb^\infty$ and $h_k\in R^\alpha\rho^{(k+2)\beta}\Hb^\infty$. Expanding $\Ups_E$ in the first argument, this is further equal to
  \[
    D_1|_{\phi_{a_k,b_k}^*g}\Ups_E(\delta_{g_0}^*\omega_{\dot a_k,\dot b_k};g_0) - \theta'_k
  \]
  modulo $R^\alpha\rho^{(k+2)\beta}\Hb^\infty + R^{2\alpha}\rho^{2(k+1)\beta}\Hb^\infty=R^\alpha\rho^{(k+2)\beta}\Hb^\infty$; and finally we can replace the point of linearization $\phi_{a_k,b_k}^*g\equiv g_0\bmod R^\alpha\rho^\beta\Hb^\infty$ by $g_0$ upon committing another error in $R^\alpha\rho^{(k+2)\beta}\Hb^\infty$. We must therefore construct $\omega\in R^\alpha\rho^{(k+1)\beta}\Hb^\infty$ so that
  \begin{equation}
  \label{EqStGaugeEq}
    D_1|_{g_0}\Ups_E(\delta_{g_0}^*\omega;g_0)\equiv\theta'_k\bmod R^\alpha\rho^{(k+2)\beta}\Hb^\infty.
  \end{equation}
  Once we have such an $\omega$, the cut-off 1-form $\chi(\rho/\eps_k)\omega$ satisfies the same equation and, when $\eps_k$ is chosen sufficiently small, it moreover has small norm in $R^\alpha\rho^{(k+1)\beta-\delta}\Hb^{d_4+1+k}$.\footnote{Such an argument is familiar from proofs of Borel's lemma.} One can then read off $\dot a_k,\dot b_k$ with the analogous properties from the coefficients of $\chi(\rho/\eps_k)\omega$.

  In order to solve~\eqref{EqStGaugeEq}, we can further replace $D_1|_{g_0}\Ups_E(\delta_{g_0}^*\cdot;g_0)$ (the gauge potential wave operator) by its indicial operator
  \[
    I(D_1|_{g_0}\Ups_E(\delta_{g_0}^*\cdot;g_0))=I(-\delta_{g_0}\sfG_{g_0}+E_{g_0})\circ I(\delta_{g_0}^*)
  \]
  whose indicial family, which we denote here by $I(\lambda)$, was computed in~\eqref{EqE1ChoiceBoxGauge}; in particular, $I(\lambda)$ is invertible for $\lambda\notin\{-1,2,3,4\}$. Passing to the Mellin transform side, we are thus led to set
  \[
    \hat\omega(\lambda,R,\omega) := I(\lambda)^{-1} (\cM\theta'_k)(\lambda,R,\omega),\qquad \Re\lambda=(k+1)\beta.
  \]
  In view of the isomorphism~\eqref{EqE2Plancherel}, we then have $\omega:=\cM_{(k+1)\beta}^{-1}\hat\omega\in R^\alpha\rho^{(k+1)\beta}\Hb^\infty$, as desired.

  To complete the proof, it remains to set $a:=\sum_{i=0}^\infty\dot a_i$ and $b:=\sum_{i=0}^\infty\dot b_i$. By construction, both sums converge in every $R^\alpha\rho^\beta\Hb^N$-norm and define elements of $R^\alpha\rho^{\beta-\delta}\cC_\bop^1$ with small norm. Since $\Ups_E(\phi_{a,b}^*g;g_0)\equiv\Ups_E(\phi_{a_k,b_k}^*g;g_0)\bmod R^\alpha\rho^{(k+1)\beta}\Hb^\infty$ lies in $R^\alpha\rho^{(k+1)\beta}\Hb^\infty$ for all $k$, we obtain~\eqref{EqStGauge}.
\end{proof}

Note that $g'$ can be written as $g'=g_\sfb+\chi h_0+\tilde h$ where $h_0,\tilde h$ are of the same class~\eqref{EqESolh} as before (and in fact only $\tilde h$ is changed). To avoid cumbersome notation, we now relabel $g',R_1$ as $g,R_0$. 

\begin{lemma}[Log-smoothness in the improved gauge]
\label{LemmaStLog}
  Suppose $h_0,\tilde h$ are as in~\eqref{EqESolh}, with small weighted $\Hb^d$-norms. Suppose that
  \begin{equation}
  \label{EqStLogAssm}
    P(h_0,\tilde h,\theta) = 0,\qquad  \theta\in R^\alpha\rho^\infty\Hb^\infty(\Omega_{\rho_0,R_0};\upbeta^*({}^0 T^*M)).
  \end{equation}
  Then $\tilde h$ is log-smooth down to $\cI^+$; that is, for each $i\in\N$ there exist $m_i\in\N_0$ and $\tilde h_{i,m}\in R^\alpha\Hb^\infty(\cI^+_{R_0};\upbeta^*(S^2\,{}^0 T^*M))$ so that, for all $N\in\N$,
  \begin{equation}
  \label{EqStLogExp}
    \tilde h(\rho,R,\omega) - \sum_{i=1}^N\sum_{m=0}^{m_i} \rho^i(\log\rho)^m\tilde h_{i,m}(R,\omega) \in R^\alpha\rho^N\Hb^\infty\bigl(\Omega_{\rho_0,R_0};\upbeta^*(S^2\,{}^0 T^*M)\bigr).
  \end{equation}
\end{lemma}

The assumption~\eqref{EqStLogAssm} is satisfied by the metric produced by Proposition~\ref{PropStGauge}; in fact, both lines of~\eqref{EqEEinOp} vanish separately. The conclusion holds (with the same proof) assuming only that $P(h_0,\tilde h,\theta)\in R^\alpha\rho^\infty\Hb^\infty$.

\begin{rmk}[Integer indicial roots]
\label{RmkStIndInt}
  The fact that the indicial roots of $L_{h_0,\tilde h}$ are integers is the reason for the log-smoothness of $\tilde h$. This fact should, however, be regarded as coincidental. If we used a different gauge for which, say, the indicial root $0$ and the corresponding space of indicial solutions was the same, but the remaining indicial roots in $\Re\lambda>0$ were different (non-integers, and possibly even in complex conjugate pairs), then $\tilde h$ would be polyhomogeneous. This would still suffice for the proof of Proposition~\ref{PropStCI} below to go through, as follows from part~\eqref{ItScCINothing} of Lemma~\ref{LemmaStCISeq} below.
\end{rmk}

\begin{proof}[Proof of Lemma~\usref{LemmaStLog}]
  We again arrange for $\beta\in(0,1)$ to be irrational by reducing it slightly if necessary. Write $g_0=g_\sfb+\chi h_0$. In the computations below, we write `$\equiv$' for equality modulo $R^\alpha\rho^\infty\Hb^\infty$. We thus have
  \begin{equation}
  \label{EqStLogEqn}
    0 = P(h_0,\tilde h,\theta) = P(h_0,0,\theta) + \int_0^1 L_{h_0,s\tilde h}\tilde h\,\dd s \equiv 2\bigl( \Ric(g_0)-\Lambda g_0 \bigr) + \int_0^1 L_{h_0,s\tilde h}\tilde h\,\dd s.
  \end{equation}
  The computations in the proof of Proposition~\ref{PropE1Struct} imply that $\Ric(g_0)-\Lambda g_0$ vanishes at $\rho=0$. Working in the splitting $[0,\rho_0]_\rho\times\cI_{R_0}^+$ and noting that $g_0\in\CI([0,\rho_0];R^\alpha\Hb^\infty(\cI^+_{R_0}))$, we thus have
  \[
    f := 2(\Ric(g_0)-\Lambda g_0) \in \rho\CI([0,\rho_0];R^\alpha\Hb^\infty).
  \]
  (See Lemma~\ref{LemmaStCIg0} below for a more precise statement.) Using~\eqref{EqE1StructInd} with $g_{(0)}=\dd x^2+\tau^2 h_0$ (cf.\ \eqref{EqE1Structg0h0}), we thus have
  \[
    I_{g_{(0)}}(\rho\pa_\rho)(\chi\tilde h) \equiv -f + \tilde f,
  \]
  where $\tilde f$ arises from the action of $R_0,\int_0^1 \tilde R_{h_0,s\tilde h}\,\dd s$ on $\tilde h$ and thus lies in $R^\alpha\rho^{\beta+1}\Hb^\infty+R^{2\alpha}\rho^{2\beta}\Hb^\infty\subset R^\alpha\rho^{2\beta}\Hb^\infty$. (The insertion of the cutoff $\chi$ produces an error $[I_{g_{(0)}}(\rho\pa_\rho),\chi]\tilde h$ which vanishes near $\rho=0$ and thus lies in $R^\alpha\rho^\infty\Hb^\infty$.) We can replace $\rho\pa_\rho$ by $\rho'\pa_{\rho'}$ (see~\eqref{EqE2DecayCoordp}--\eqref{EqE2DecayRhop}) upon committing a further error of the schematic form $\rho\Diffb^2([0,\infty)\times\cI^+_{R_0})\tilde h$, which thus lies in $R^\alpha\rho^{\beta+1}\Hb^\infty$.

  At this point, we pass to the Mellin transform in $\rho'$. The Mellin transform of $f$ is meromorphic with poles at the positive integers $1,2,3,\ldots$. Using the indicial root computation of Lemma~\ref{LemmaEIndRoot} and a contour shifting argument in the inverse Mellin transform similarly to the proof of Lemma~\ref{LemmaE2Inv} (but no longer keeping track of tame estimates), we can now draw the follow conclusions. If $\beta\in(0,\frac12)$, then $\chi\tilde h\in R^\alpha\rho^{2\beta}$. If $\beta\in(\frac12,1)$, then
  \begin{equation}
  \label{EqStLogLot}
    \chi\tilde h=\chi\sum_{m=0}^{m_1}\rho(\log\rho)^m\tilde h_{1,m}+\tilde h^\flat
  \end{equation}
  where
  \begin{equation}
  \label{EqStLogRemainder}
    \tilde h_{1,m}\in R^\alpha\Hb^\infty(\cI^+_{R_0}),\qquad \tilde h^\flat\in R^\alpha\rho^{2\beta}\Hb^\infty.
  \end{equation}
  (The logarithmic terms arise from the residue theorem when on the Mellin transform side there is a pole of order $\geq 2$ at $\lambda=1$.) In the former case, we repeat the same argument with $2\beta$ in place of $\beta$ until, after finitely many steps, we are in the latter case and thus extract the first term in the expansion of $\tilde h$.

  With $2\beta\in(1,2)$, we now proceed inductively to extract an expansion for $\tilde h^\flat$. The key point is that the coefficients of $L_{h_0,s\tilde h}$ in~\eqref{EqStLogEqn} inherit the partial log-smoothness of $\tilde h$, as follows by inspection of the explicit computations in the proof of Proposition~\ref{PropE1Struct}. To wit, armed with~\eqref{EqStLogLot}, we first use that the coefficients of $L_{h_0,s\tilde h}$ are log-smooth modulo $R^\alpha\rho^{2\beta}\Hb^\infty$, and therefore the action of this operator on $\chi\rho(\log\rho)^m\tilde h_{1,m}$ is log-smooth modulo $R^\alpha\rho^{2\beta+1-\delta}\Hb^\infty$ for any $\delta>0$ (or $\delta=0$ when $m=0$). Since $(L_{h_0,s\tilde h}-I_{g_{(0)}})\tilde h^\flat\in R^\alpha\rho^{2\beta+1}\Hb^\infty$, we thus conclude that
  \[
    I_{g_{(0)}}(\rho'\pa_\rho')\tilde h^\flat \equiv \tilde f^\flat \bmod R^\alpha\rho^{2\beta+1-\delta}\Hb^\infty,\quad \delta>0,
  \]
  where $\tilde f^\flat$ is log-smooth in $\rho$ and of class $R^\alpha\Hb^\infty$ in $(R,\omega)$. Using the (inverse) Mellin transform as before, this equation allows us to extract $\rho^2(\log\rho)^m$ leading order terms of $\tilde h^\flat$, with a remainder of class $R^\alpha\rho^{2\beta+1-\delta}\Hb^\infty$ which thus vanishes to almost a full order more at $\rho=0$ compared to~\eqref{EqStLogRemainder}. Proceeding iteratively in this fashion produces the expansion~\eqref{EqStLogExp} and finishes the proof.
\end{proof}

\begin{rmk}[Comparison with the Riemannian setting, I]
\label{RmkStRiem1}
  The idea to put $g$ into a convenient gauge condition in order to improve its asymptotic behavior is used in \cite[\S4]{ChruscielDelayLeeSkinnerConfCompReg} (where a harmonic map gauge is used), with polyhomogeneity being deduced in \cite[\S5]{ChruscielDelayLeeSkinnerConfCompReg} using \cite{AnderssonChruscielHyp}.
\end{rmk}

\subsection{Smoothness and precise Taylor expansion at the conformal boundary}
\label{SsStCI}

The gauge condition serves no further purpose now: we only used it as a means to ensure that $\tilde h$ has simple asymptotics (namely, log-smoothness) at $\rho=0$. We shall now reduce the task of further sharpening the asymptotic behavior of $\tilde h$ to the level of indicial operators of the (ungauged) linearized Einstein vacuum equations and the symmetric gradient (or Lie derivative), somewhat analogously to (but simpler than) the analysis in~\cite[\S{7}]{HintzGlueLocI}. We work in the splitting~\eqref{Eq0bT0Split} of ${}^0 T^*M$ and (via combining the splittings~\eqref{Eq0bST0Split} with~\eqref{EqE1AuxSplitRef})
\[
  S^2\,{}^0 T^*M = \R\frac{\dd\tau^2}{\tau^2} \oplus \Bigl(2\frac{\dd\tau}{\tau}\otimes_s \tau^{-1}T^*X\Bigr) \oplus \R\tau^{-2}g_{(0)} \oplus \tau^{-2}\ker\tr_{g_{(0)}},
\]
and recall $\Lambda=3$. We thus proceed to analyze the kernel of
\begin{equation}
\label{EqStCIIndRic}
\begin{split}
  2 I(D\Ric-\Lambda,\lambda) &:= I\bigl(\Box_{g_0} - 2\delta_{g_0}^*\delta_{g_0}\sfG_{g_0} + 2\sR_{g_0} - 2\Lambda,\lambda\bigr) \\
    &= \begin{pmatrix}
         3\lambda-6 & 0 & -3\lambda^2+6\lambda & 0 \\
         0 & 0 & 0 & 0 \\
         -\lambda+6 & 0 & \lambda^2-6\lambda & 0 \\
         0 & 0 & 0 & \lambda^2-3\lambda
       \end{pmatrix}
\end{split}
\end{equation}
and its relationship with the range of
\[
  I(\delta^*,\lambda) := I(\delta_{g_0}^*,\lambda) = \begin{pmatrix} \lambda & 0 \\ 0 & \frac12(\lambda+1) \\ 1 & 0 \\ 0 & 0 \end{pmatrix}.
\]
(We use~\eqref{EqE1StructGg}, \eqref{EqE1StructR}, \eqref{EqE1StructBoxg}, and \eqref{EqE1Structdelstar} to derive these expressions.) Dually, we study the range of $I(D\Ric-\Lambda,\lambda)$ and its relationship with the kernel of
\begin{equation}
\label{EqStCIdelG}
  2 I(\delta\sfG,\lambda) := 2 I(\delta_{g_0}\sfG_{g_0},\lambda) = \begin{pmatrix} \lambda-6 & 0 & 3\lambda-6 & 0 \\ 0 & 2\lambda-8 & 0 & 0 \end{pmatrix}.
\end{equation}
The dependence of these operators on $R,\omega$ is only through $g_{(0)}$, i.e.\ through the bundle splitting. Note that $I(D\Ric-\Lambda,\lambda)\circ I(\delta^*,\lambda)=0$, which is due to the diffeomorphism covariance of the Einstein operator, and $I(\delta\sfG,\lambda)\circ I(D\Ric-\Lambda,\lambda)=0$, which arises from the linearized second Bianchi identity.

We proceed to analyze the above $\lambda$-dependent $4\times 4$ and $4\times 2$ matrices, which define linear maps $\C^4\to\C^4$ and $\C^2\to\C^4$ (denoted by the same symbols). Acting on functions $h=h(\rho)$, note that $I(D\Ric-\Lambda,\rho\pa_\rho)(\rho^\lambda h)=\rho^\lambda I(D\Ric-\Lambda,\rho\pa_\rho)h$, and thus
\begin{align*}
  I(D\Ric-\Lambda,\rho\pa_\rho) (\rho^\lambda(\log\rho)^m h) &= \pa_\lambda^m\bigl( \rho^\lambda I(D\Ric-\Lambda,\lambda)h \bigr) \\
    &= \sum_{m'=0}^m \binom{m}{m'} \rho^\lambda(\log\rho)^{m'} \pa_\lambda^{m-m'}I(D\Ric-\Lambda,\lambda)h
\end{align*}
In particular, when $k\geq 1$, we have the implication
\begin{equation}
\label{EqStCILog}
\begin{split}
  &I(D\Ric-\Lambda,\rho\pa_\rho) \sum_{m=0}^k \rho^\lambda(\log\rho)^m h_m = 0 \\
  &\qquad \implies h_k\in\ker I(D\Ric-\Lambda,\lambda),\quad \pa_\lambda I(D\Ric-\Lambda,\lambda)h_k \in \ran I(D\Ric-\Lambda,\lambda).
\end{split}
\end{equation}

\begin{lemma}[Kernel of linearized Einstein modulo pure gauge]
\label{LemmaStCISeq}
  Let $\lambda\in\C$ and $h\in\C^4$; suppose that $I(D\Ric-\Lambda,\lambda)h=0$.
  \begin{enumerate}
  \item\label{ItScCINothing} If $\lambda\neq -1,0,3$, then $h\in\ran I(\delta^*,\lambda)$.
  \item\label{ItScCINoLog} If $\lambda=0,3$ and $\pa_\lambda I(D\Ric-\Lambda,\lambda)h\in\ran I(D\Ric-\Lambda,\lambda)$, then $h\in\ran I(\delta^*,\lambda)$. Moreover,
  \begin{equation}
  \label{EqStCISeq3}
    \ker I(D\Ric-\Lambda,\lambda)/\ran I(\delta^*,\lambda) = \mathspan\{(0,0,0,1)\},\qquad \lambda=0,3.
  \end{equation}
  \item\label{ItScCIMinus1} If $\lambda=-1$, then $h=\pa_\lambda I(\delta^*,-1)\omega_0+I(\delta^*,0)\omega_1$ for some $\omega_0,\omega_1\in\C^2$ with $\omega_0\in\ker I(\delta^*,-1)$.
  \end{enumerate}
  When $\lambda$ is real, all statements hold also for real vectors.
\end{lemma}
\begin{proof}
  If $\lambda\neq 0,3$, we have
  \begin{equation}
  \label{EqStCISeqKer}
    \ker I(D\Ric-\Lambda,\lambda) = \mathspan\{ (\lambda,0,1,0),\ (0,1,0,0) \}.
  \end{equation}
  This uses that when $\lambda\neq 2$, resp.\ $\lambda\neq 6$, the first, resp.\ third row of $I(D\Ric-\Lambda,\lambda)$ is a nonzero multiple of $(1,0,-\lambda,0)$. For $\lambda\neq -1$, this equals the range of $I(\delta^*,\lambda)$.

  We next compute
  \[
    2\pa_\lambda I(D\Ric-\Lambda,\lambda) = \begin{pmatrix} 3 & 0 & -6\lambda+6 & 0 \\ 0 & 0 & 0 & 0 \\ -1 & 0 & 2\lambda-6 & 0 \\ 0 & 0 & 0 & 2\lambda-3 \end{pmatrix}.
  \]
  Consider the case $\lambda=0$. The basis elements of
  \[
    \ker I(D\Ric-\Lambda,0) = \mathspan\{ (0,1,0,0),\ (0,0,1,0),\ (0,0,0,1) \}
  \]
  get mapped by $2\pa_\lambda I(D\Ric-\Lambda,0)$ to $(0,0,0,0)$, $(6,0,-6,0)$ $(0,0,0,-3)$. But $\ran I(D\Ric-\Lambda,0)=\mathspan\{(1,0,-1,0)\}$. It remains to observe that $(0,1,0,0)$, $(0,0,1,0)\in\ran I(\delta^*,0)$.

  The arguments for $\lambda=3$ are similar: now the basis elements of
  \[
    \ker I(D\Ric-\Lambda,3) = \mathspan\{ (0,1,0,0),\ (3,0,1,0),\ (0,0,0,1) \},
  \]
  get mapped by $2\pa_\lambda I(D\Ric-\Lambda,3)$ to $(0,0,0,0)$, $(-3,0,-3,0)$, $(0,0,0,3)$. But $\ran I(D\Ric-\Lambda,3)=\mathspan\{(1,0,1,0)\}$. The claim then follows from $(0,1,0,0),(3,0,1,0)\in\ran I(\delta^*,3)$.

  The final part follows again from~\eqref{EqStCISeqKer}, now for $\lambda=-1$, and the observation that $(-1,0,1,0)\in\ran I(\delta^*,-1)$, while $(0,1,0,0)=\pa_\lambda I(\delta^*,-1)(0,2)$ with $(0,2)\in\ker I(\delta^*,-1)$.
\end{proof}

\begin{rmk}[Indicial roots modulo pure gauge and stability of de~Sitter space]
\label{RmkStCIStab}
  Lemma~\ref{LemmaStCISeq} is a mode stability statement for de~Sitter space: all modes, i.e.\ here indicial solutions of the linearization of $\Ric-\Lambda$, with $\Re\lambda\leq 0$, $\lambda\neq 0$, are pure gauge; and in fact the only modes which are not pure gauge occur at $\lambda=0,3$. Moreover, modulo pure gauge solutions, the $\lambda=3$ mode lies in $\ker\tr_{g_{(0)}}$, i.e.\ it is a trace-free tangential-tangential tensor. Once one puts back the $(R,\omega)$-dependence, one can draw further conclusions related to the results in \cite{FriedrichDeSitterPastSimple} (see also \cite{FeffermanGrahamAmbient,FeffermanGrahamAmbientBook,HintzAsymptoticallydS}) concerning the asymptotic degrees of freedom of asymptotically de~Sitter type metrics solving the Einstein vacuum equations, which are given by a Riemannian metric ($g_{(0)}$) and a transverse-traceless tensor (i.e.\ an element of $\ker\tr_{g_{(0)}}\cap\ker\delta_{g_{(0)}}$) on the conformal boundary; we recover one direction of this in~\eqref{EqStCI} below.
\end{rmk}

Note that by basic linear algebra (or by an inspection of the proof), Lemma~\ref{LemmaStCISeq} also applies to families. Thus, if $h$ depends on a parameter $(R,\omega)\in\cI^+_{R_0}$ in an $R^\alpha\Hb^\infty$ fashion, then in part~\eqref{ItScCINothing}, one can find $\omega$ with the same parameter dependence so that $I(\delta_{g_0}^*,\lambda)\omega=h$; similarly in the other parts.

For brevity, we now focus on $\lambda>0$, as positive indicial roots are the only ones of interest in our quest to improve the asymptotic behavior of the decaying tensor $\tilde h$. We study generalized mode solutions, i.e.\ those which may feature $\log\rho$ factors.

\begin{cor}[Quasihomogeneous nullspace modulo pure gauge]
\label{CorStCISeq}
  Let $\lambda>0$. Let $h_0,\ldots,h_k\in\R^4$ and set $h(\rho):=\sum_{m=0}^k\rho^\lambda(\log\rho)^m h_m$. Suppose that $I(D\Ric-\Lambda,\rho\pa_\rho)h=0$. Then there exist $\omega_0,\ldots,\omega_k\in\R^2$ so that for $\omega(\rho):=\sum_{m=0}^k\rho^\lambda(\log\rho)^m\omega_m$, the following holds.
  \begin{enumerate}
  \item In the case $\lambda\neq 3$: $h=I(\delta^*,\rho\pa_\rho)\omega$.
  \item In the case $\lambda=3$: $h-I(\delta^*,\rho\pa_\rho)\omega$ is a scalar multiple of $\{(0,0,0,1)\}$.
  \end{enumerate}
\end{cor}
\begin{proof}
  We have $h_k\in\ker I(D\Ric-\Lambda,\lambda)$. When $\lambda\neq 3$, this implies the existence of $\omega_k\in\R^2$ with $h_k=I(\delta^*,\lambda)\omega_k$. Therefore,
  \begin{align*}
    h - I(\delta^*,\rho\pa_\rho)\bigl(\rho^\lambda(\log\rho)^k\omega_k\bigr) &= \rho^\lambda(\log\rho)^k\bigl(\,\underbrace{h_m-I(\delta^*,\lambda)\omega_m}_{=0}\,\bigr) \\
      &\quad\qquad + \sum_{m=0}^{k-1}\rho^\lambda(\log\rho)^m h_m - k \rho^\lambda(\log\rho)^{k-1}\pa_\lambda I(\delta^*,\lambda)\omega_k
  \end{align*}
  is of the same form as $h$ except with $k$ reduced by $1$. An iterative argument thus finishes the proof in this case.

  For $\lambda=3$, the same argument eliminates $h_k$, \emph{provided} that $k\geq 1$. We thus find $\omega_k,\ldots,\omega_1\in\R^2$ so that $h-I(\delta^*,\rho\pa_\rho)(\sum_{m=1}^k\rho^\lambda(\log\rho)^m\omega_m)=:h_0$ is $\rho$-independent. Using then~\eqref{EqStCISeq3}, we can find $\omega_0\in\R^2$ so that $h_0-I(\delta^*,\rho\pa_\rho)\omega_0$ lies in the span of $(0,0,0,1)$.
\end{proof}

We now turn to the range of $I(D\Ric-\Lambda,\lambda)$, which is necessarily contained in $\ker I(\delta\sfG,\lambda)$. Again we only consider $\lambda>0$.

\begin{lemma}[Solvability of linearized Einstein]
\label{LemmaStCISolv}
  Let $\lambda>0$ and $f\in\R^4$; suppose that $I(\delta\sfG,\lambda)f=0$.
  \begin{enumerate}
  \item If $\lambda\neq 3,4$, then $f\in\ran I(D\Ric-\Lambda,\lambda)$, and indeed we can find a solution $h$ of $f=I(D\Ric-\Lambda,\lambda)h$ of the form $h=(0,0,h_3,h_4)$.
  \item\label{ItStCISolv3} If $\lambda=3$ and $f=(0,f_2,0,0)$, then $f_2=0$.
  \item\label{ItStCISolv4} If $\lambda=4$ and $f=(f_1,0,f_3,f_4)$, then we can write $f=I(D\Ric-\Lambda,\lambda)h$ for some $h$ of the form $h=(0,0,h_3,h_4)$.
  \end{enumerate}
\end{lemma}

The assumptions on the form of $f$ in the second and third part will arise from evenness considerations.

\begin{proof}[Proof of Lemma~\usref{LemmaStCISolv}]
  For $\lambda\neq 3$, the range of $I(D\Ric-\Lambda,\lambda)$ is $2$-dimensional due to~\eqref{EqStCISeqKer}, and so is the kernel of $I(\delta\sfG,\lambda)$ when $\lambda\neq 4$ due to~\eqref{EqStCIdelG}. To prove the full statement of the first part, we compute for $\lambda\neq 3,4$
  \[
    \ker I(\delta\sfG,\lambda) = \mathspan\{ (3\lambda-6,0,-\lambda+6,0),\ (0,0,0,1) \}
  \]
  and note that
  \begin{equation}
  \label{EqStCISolv}
  \begin{split}
    (3\lambda-6,0,-\lambda+6,0)&=2 I(D\Ric-\Lambda,\lambda)(0,0,-\lambda^{-1},0), \\
    (0,0,0,1)&=2 I(D\Ric-\Lambda,\lambda)(0,0,0,(\lambda^2-3\lambda)^{-1}).
  \end{split}
  \end{equation}

  For $\lambda=3$, the non-vanishing of $2\lambda-8$ in~\eqref{EqStCIdelG} implies the second part. For $\lambda=4$, the assumptions on $f$ imply $-2 f_1+6 f_3=0$, so $f\in\mathspan\{(3,0,1,0),(0,0,0,1)\}\subset\ran I(D\Ric-\Lambda,4)$; and the existence of $h=(0,0,h_3,h_4)$ solving $I(D\Ric-\Lambda,4)h=f$ follows from~\eqref{EqStCISolv}.
\end{proof}

As a final preparation, we record:
\begin{lemma}[Ricci tensor of $g_0$]
\label{LemmaStCIg0}
  We work in the frame~\eqref{EqE1StructFrame} on $\breve M$, and with the frame $\dd x^1,\dd x^2,\dd x^3$ on $\cI^+$. Consider a tensor $h_0\in R^\alpha\Hb^\infty(\cI^+_{R_0};\upbeta^*(\tau^{-2}S^2 T^*X))$ with small $R^\alpha\cC_\bop^2$-norm. Set $g_0=g'_\sfb+\chi h_0$ where\footnote{We shall first apply this with $g'_\sfb=g_\sfb$; this is the case that will be used for the proof of the existence of $g_\sfb^{\rm FG}$ in Proposition~\ref{PropStCI}\eqref{ItStCICI} below. Only once $g_\sfb^{\rm FG}$ has been constructed will we use Lemma~\ref{LemmaStCIg0} with $g'_\sfb=g_\sfb^{\rm FG}$.} $g'_\sfb\equiv g_\dS\bmod\rho^3\CI(\Omega_{\rho_0;R_0};\upbeta^*(S^2\,{}^0 T^*M))$ and $g_{(0)}=g_{(0)}(x,\dd x)=\dd x^2+h_{(0)}$ where $h_{(0)}=h(\tau\pa_\mu,\tau\pa_\nu)\,\dd x^\mu\,\dd x^\nu$. Modulo $\rho^3\CI+R^\alpha\rho^\infty\Hb^\infty(\Omega_{\rho_0,R_0})$, we then have
  \begin{equation}
  \label{EqStCIg0}
    (\Ric(g_0) - \Lambda g_0)_{\mu\nu} \equiv \begin{cases} \tau^2\Ric(g_{(0)})_{i j}, & (\mu,\nu)=(i,j), \\ 0 & \text{otherwise}. \end{cases}
  \end{equation}
  (We recall that the indices $\mu,\nu$ run from $0$ to $3$, and the indices $i,j$ from $1$ to $3$.)
\end{lemma}
\begin{proof}
  For the computation, we can drop the cutoff $\chi$ and consider $g_0=g_\sfb+h_0$. The result then follows from the expressions in the proof of Proposition~\ref{PropE1Struct}, which give, modulo $\rho^3\CI$, $\Gamma(g_0)_{\lambda\mu\nu}\equiv 0$ except for $\Gamma(g_0)_{\ell i 0}\equiv -(g_{(0)})_{i\ell}$, $\Gamma(g_0)_{0 i j}\equiv(g_{(0)})_{i j}$, and $\Gamma(g_0)_{\ell i j}\equiv\tau\Gamma(g_{(0)})_{\ell i j}$, and therefore $\Gamma(g_0)^\lambda_{\mu\nu}\equiv 0$ except for $\Gamma(g_0)^\ell_{i 0}\equiv-\delta_i^\ell$, $\Gamma(g_0)^0_{i j}\equiv-(g_{(0)})_{i j}$, and $\Gamma(g_0)^\ell_{i j}\equiv\tau\Gamma(g_{(0)})^\ell_{i j}$. This gives~\eqref{EqStCIg0} after a short computation.
\end{proof}

\begin{prop}[Taylor expansion at the conformal boundary]
\label{PropStCI}
  Let
  \[
    h_0\in R^\alpha\Hb^\infty(\cI^+_{R_0};\upbeta^*(\tau^{-2}S^2 T^*X)),
  \]
  with small $R^\alpha\Hb^d$-norm. Suppose $\tilde h\in R^\alpha\rho^\beta\Hb^\infty(\Omega_{\rho_0,R_0};\upbeta^*(S^2\,{}^0 T^*M))$ (where $\beta\in(0,1)$ is arbitrary) is log-smooth at $\cI^+$, i.e.\ it satisfies~\eqref{EqStLogExp} for all $N$. Suppose that $g:=g_\sfb+\chi h_0+\tilde h$ satisfies
  \[
    \Ric(g) - \Lambda g = 0\ \ \text{on}\ \ \Omega_{\rho_0,R_0}.
  \]
  Let $R_1\in(0,R_0)$. Then there exist
  \begin{align*}
    a &\in \rho^3\CI(\Omega_{\rho_0,R_0})+R^\alpha\rho^\beta\Hb^\infty(\Omega_{\rho_0,R_0}), \\
    b &\in \rho^3\CI(\Omega_{\rho_0,R_0};\R^3) + R^\alpha\rho^\beta\Hb^\infty(\Omega_{\rho_0,R_0};\R^3)
  \end{align*}
  which are log-smooth at $\cI^+$, small in $\rho^\beta\cC_\bop^1$, and vanish for $\rho\geq\frac{\rho_0}{2}$ so that (recalling the definition of $\phi_{a,b}$ from~\eqref{EqStPhiab}) the pullback metric $\phi_{a,b}^*g\in\CI+R^\alpha\rho^\beta\Hb^\infty$ is \emph{smooth} down to $\cI^+$. More precisely, there exist
  \[
    \tilde h_2\in R^\alpha\Hb^\infty(\cI^+_{R_0};\upbeta^*(\tau^{-2}S^2 T^*X)),\quad \tilde h_i\in(\CI+R^\alpha\Hb^\infty)(\cI_{R_0}^+;\tau^{-2}S^2 T^*X),\ \ i=3,4,\ldots,
  \]
  so that, for all $N\in\N$,
  \begin{equation}
  \label{EqStCI}
    (\phi_{a,b}^*g)(\rho,R,\omega)-g_\dS(\rho,R,\omega) - \sum_{i=2}^N \rho^i \tilde h_i(R,\omega) \in R^\alpha\rho^N\Hb^\infty\bigl(\Omega_{\rho_0,R_1};\upbeta^*(S^2\,{}^0 T^*M)\bigr),
  \end{equation}
  and so that $g_{(3)}:=\tilde h_3(\tau\pa_\mu,\tau\pa_\nu)\,\dd x^\mu\,\dd x^\nu$ is a weighted transverse-traceless tensor, that is,
  \begin{equation}
  \label{EqStCIh3}
    g_{(3)}\in (\CI+R^\alpha\Hb^\infty)\bigl(\cI^+_{R_0};\ker\tr_{g_{(0)}}\bigr),\quad
    \delta_{g_{(0)}}\bigl(|x|^{-3}g_{(3)}\bigr)=0.
  \end{equation}
  Moreover,~\eqref{EqStCI} is sharp in the sense that $\tilde h_2\neq 0$ unless the metric $g_{(0)}=\dd x^2+h_{(0)}$ on $\cI^+$ is flat; here $h_{(0)}=h_0(\tau\pa_\mu,\tau\pa_\nu)\,\dd x^\mu\,\dd x^\nu$. Furthermore:
  \begin{enumerate}
  \item\label{ItStCICI} if $h_0=\tilde h=0$, so $g=g_\sfb$, then the conclusions hold with $a,b\in\rho^3\CI$, $\tilde h_2=0$, and $\tilde h_i\in\CI$, $i\geq 3$. This produces $g_\sfb^{\rm FG}:=\phi_{a,b}^*g_\sfb$ with $g_\sfb^{\rm FG}-g_\sfb\in\rho^3\CI(\Omega_{\rho_0,R_0};\upbeta^*(S^2\,{}^0 T^*M))$;
  \item\label{ItStCIHb} if $g=g_\sfb^{\rm FG}+\chi h_0+\tilde h$, then the conclusions hold with $a,b\in R^\alpha\rho^\beta\Hb^\infty$ and $\tilde h_i\in R^\alpha\Hb^\infty$, $i=2,3,4,\ldots$, with $g_\dS$ in~\eqref{EqStCI} replaced by $g_\sfb^{\rm FG}$, and with $g_{(3)}$ equal to the $\rho^3$ coefficient of $g$.\footnote{Since the $\rho^3$ coefficient of $g_\dS$ vanishes, this is consistent with the definition of $g_{(3)}$ in the other settings considered in this Proposition.}
  \end{enumerate}
\end{prop}
\begin{proof}
  We discuss the general case and scenario~\eqref{ItStCICI} simultaneously. In scenario~\eqref{ItStCICI}, the arguments simplify since $g_\sfb$ has no logarithmic terms in its Taylor expansion at $\rho=0$, and all tensors on $\cI^+_{R_0}$ arising in the proof in this case are smooth. Once $g_\sfb^{\rm FG}$ has been constructed, the same arguments then apply if we replace $g_\sfb$ by $g_\sfb^{\rm FG}$ throughout the proof; all smooth terms, starting with $\dot a_{b,3},\dot b_{b,3}\in\CI$ prior to~\eqref{EqStCIg3} below, can then be taken to be equal to $0$ since the relevant smooth metric coefficients (which are the Taylor coefficients of $g_\sfb^{\rm FG}$) are already free of logarithmic terms and valued in $\upbeta^*(\tau^{-2}S^2 T^*X)$ (and $\upbeta^*(\tau^{-2}\ker\tr_{g_{(0)}})$ in the case of the $\rho^3$ term). With these modifications in mind, the reader may thus read the following proof as is, or with $h_0=0$, $\tilde h=0$ (for scenario~\eqref{ItStCICI}), or with $g_\sfb^{\rm FG}$ in place of $g_\sfb$ (for scenario~\eqref{ItStCIHb}) throughout.

  We write $g_0=g_\sfb+\chi h_0$ and
  \begin{equation}
  \label{EqStCIRic0}
    0 = \Ric(g_0+\tilde h) - \Lambda(g_0+\tilde h) = (\Ric(g_0)-\Lambda g_0) + \int_0^1 D_{g_0+s\tilde h}\Ric(\tilde h)-\Lambda\tilde h\,\dd s.
  \end{equation}

  \pfstep{Step~1. Eliminating the $\rho^1$ terms.} We work modulo log-smooth terms with almost $\rho^2$ decay at $\cI^+$ (and $R^\alpha\Hb^\infty$ behavior in $(R,\omega)$)---we shall write `almost-$\cO(\rho^2)$' in short (and omit the `almost' if there are no log terms at leading order). We can thus replace $D_{g_0+s\tilde h}\Ric-\Lambda$ in~\eqref{EqStCIRic0} by the indicial operator $I(D\Ric-\Lambda,\rho\pa_\rho)$. In view of Lemma~\ref{LemmaStCIg0}, we obtain the equation
  \[
    I(D_{g_0}\Ric-\Lambda,\rho\pa_\rho)\sum_{m=0}^{m_1}\rho(\log\rho)^m\tilde h_{1,m} = 0.
  \]
  Corollary~\ref{CorStCISeq}, applied with $R^\alpha\Hb^\infty$-dependence on $(R,\omega)\in\cI^+_{R_0}$, then produces 1-forms $\omega_{1,m}\in R^\alpha\Hb^\infty(\cI^+_{R_0};\upbeta^*({}^0 T^*M))$, $m=0,\ldots,m_1$, so that
  \[
    \sum_{m=0}^{m_1} \rho(\log\rho)^m\tilde h_{1,m} = -I(\delta_{g_0}^*,\rho\pa_\rho)\sum_{m=0}^{m_1}\rho(\log\rho)^m\omega_{1,m}.
  \]
  Writing $\omega_{1,m}=\omega_{\dot a_{1,m},\dot b_{1,m}}$ in the notation~\eqref{EqStPullback}, with $\dot a_{1,m},\dot b_{1,m}\in R^\alpha\Hb^\infty(\cI^+_{R_0})$, we then set
  \[
    \dot a_1 = \sum_{m=0}^{m_1} \rho(\log\rho)^m \dot a_{1,m},\qquad
    \dot b_1 = \sum_{m=0}^{m_1} \rho(\log\rho)^m \dot b_{1,m},
  \]
  and $a_1=\dot a_1$, $b_1=\dot b_1$. But then, modulo almost-$\cO(\rho^2)$,
  \[
    \phi_{a_1,b_1}^*g \equiv g + \delta_{g_0}^*\omega_{a_1,b_1} \equiv g_0 + \bigl(\tilde h + I(\delta_{g_0}^*,\rho\pa_\rho)\omega_{a_1,b_1}\bigr).
  \]
  The term in parentheses is log-smooth, and its generalized Taylor expansion now starts with $\rho^2$ (times logarithmic factors).

  \pfstep{Step~2. Simplification of the $\rho^2$ terms.} Repeating Step~1 but starting with $\phi_{a_1,b_1}^*g-g_0$ in place of $\tilde h$ in~\eqref{EqStCIRic0}, the main change is now that $\tau^2\Ric(g_{(0)})$ from~\eqref{EqStCIg0} gives rise to a forcing term, i.e.\ we need to analyze
  \begin{equation}
  \label{EqStCIStep2}
    I(D_{g_0}\Ric-\Lambda,\rho\pa_\rho) \sum_{m=0}^{m_2}\rho^2(\log\rho)^m\tilde h_{2,m} = \rho^2 f,
  \end{equation}
  where $f=f(R,\omega)\in R^\alpha\Hb^\infty(\cI^+_{R_0};\upbeta^*(\tau^{-2}S^2 T^*X))$. We have $f\in\ker I(\delta_{g_0}\sfG_{g_0},2)$: this follows directly from~\eqref{EqStCIdelG}, and more conceptually from the second Bianchi identity. By Lemma~\ref{LemmaStCISolv}, we can find
  \begin{equation}
  \label{EqStCItildeh2}
    \tilde h_2 \in R^\alpha\Hb^\infty\bigl(\cI^+_{R_0};\upbeta^*(\tau^{-2}S^2 T^*X)\bigr)
  \end{equation}
  with $I(D_{g_0}\Ric-\Lambda,2)\tilde h_2=f$. By Corollary~\ref{CorStCISeq}, $\sum_{m=0}^{m_2}\rho^2(\log\rho)^m\tilde h_{2,m}-\rho^2\tilde h_2\in\ker I(D_{g_0}\Ric-\Lambda,\rho\pa_\rho)$ can be expressed as $-I(\delta_{g_0}^*,\rho\pa_\rho)\sum_{m=0}^{m_2}\rho^2(\log\rho)^m\omega_{2,m}$. Extracting $\dot a_2,\dot b_2$ from the 1-forms $\omega_{2,m}$ as above and setting $a_2=a_1+\dot a_2$, $b_2=b_1+\dot b_2$, we then find that, modulo almost-$\cO(\rho^3)$,
  \begin{equation}
  \label{EqStCIgp2}
    g'_2 := \phi_{a_2,b_2}^*g \equiv \phi_{a_1,b_1}^*g + \delta_{g_0}^*\omega_{\dot a_2,\dot b_2} \equiv g_0+\rho^2\tilde h_2.
  \end{equation}
  Note also that if $\Ric(g_{(0)})=0$, which due to $\dim\cI^+=3$ is equivalent to $\Riem(g_{(0)})=0$, i.e.\ $g_{(0)}$ being flat, then $f=0$, and thus we can take $\tilde h_2=0$.

  \pfstep{Step~3. Simplification of the $\rho^3$ terms.} The deviation of the KdS metric from the dS metric appears at this stage in view of Lemma~\ref{Lemma0bStruct}. This will be the reason why (except in setting~\eqref{ItStCIHb}) the corrections to the diffeomorphism $\phi_{a_2,b_2}$ already constructed will involve $\CI(\cI^+_{R_0})$-terms (with powers $\rho^3,\rho^4$, etc.), in addition to the log-smooth $R^\alpha\Hb^\infty$-terms which already appeared in previous steps. We thus write
  \[
    g_0 = g_\sfb + \chi h_0 = g'_0 + h_\sfb,\qquad g'_0:=g_\dS+\chi h_0,\quad h_\sfb=g_\sfb-g_\dS\in\rho^3\CI.
  \]
  From~\eqref{EqStCIgp2}, we thus get $g'_2=g'_0+\rho^2\tilde h_2+\tilde h^\flat$ where $\tilde h^\flat=\tilde h^\flat(\rho,R,\omega)$ is thus the sum of $h_\sfb\in\rho^3\CI$ and an almost-$\cO(\rho^3)$ term. Now,
  \begin{equation}
  \label{EqStCItildehflat}
    0 = \Ric(g'_2) - \Lambda g'_2 \equiv \Ric(g'_0+\rho^2\tilde h_2) - \Lambda(g'_0+\rho^2\tilde h_2) + I(D_{g'_0}\Ric-\Lambda,\rho\pa_\rho)\tilde h^\flat
  \end{equation}
  modulo almost-$\cO(\rho^4)$ terms. In particular, setting $g_2:=g'_0+\rho^2\tilde h_2$,
  \[
    f := \Ric(g_2) - \Lambda g_2
  \]
  is almost-$\cO(\rho^3)$; since $g_2$ is smooth in $\rho$, we in fact have
  \[
    f \equiv \rho^3 f_3,\qquad f_3\in R^\alpha\Hb^\infty(\cI^+_{R_0};\upbeta^*(S^2\,{}^0 T^*M)),
  \]
  modulo $\cO(\rho^4)$. In the frame~\eqref{EqE1StructFrame}, we moreover claim that
  \begin{equation}
  \label{EqStCIf3}
    (f_3)_{0 0}=(f_3)_{i j}=0,\qquad 1\leq i,j\leq 3.
  \end{equation}
  This uses the information~\eqref{EqStCItildeh2}, which implies that the $(0,j)$ and $(i,0)$ components of $g_2$ vanish: from~\eqref{EqE1StructChristoffel}, we then see that $\Gamma(g_2)_{\lambda\mu\nu}$ is even, resp.\ odd in $\rho$ if and only if the number of indices $\lambda,\mu,\nu$ which are $\geq 1$ is even, resp.\ odd. The same is then true for $\Gamma(g_2)_{\mu\nu}^\lambda$. Consider then
  \begin{equation}
  \label{EqStCIRic}
    \Ric(g_2)_{\kappa\nu} = e_\lambda\Gamma_{\nu\kappa}^\lambda - e_\nu\Gamma_{\lambda\kappa}^\lambda + \Gamma_{\lambda\rho}^\lambda\Gamma_{\nu\kappa}^\rho - \Gamma_{\nu\rho}^\lambda\Gamma_{\lambda\kappa}^\rho.
  \end{equation}
  Note that $e_0$ preserves parity (i.e.\ evenness and oddness), whereas $e_i$, $i=1,2,3$, maps even functions to odd functions and vice versa. For $\kappa=\nu=0$, the term $e_\lambda\Gamma_{0 0}^\lambda=e_0\Gamma_{0 0}^0+e_i\Gamma_{0 0}^i$ is thus even, likewise for $e_0\Gamma_{\lambda 0}^\lambda$. Similarly, $\Gamma_{\lambda\rho}^\lambda\Gamma_{0 0}^\rho$ is even since, upon expanding the sum, every term features an even number of indices $\geq 1$; likewise for the final term $\Gamma_{0\rho}^\lambda\Gamma_{\lambda 0}^\rho$. The same considerations apply in the case $(\kappa,\nu)=(i,j)$ since now $2$ (thus, an even number of) further indices are $\geq 1$. Thus $\Ric(g_2)_{0 0}$ and $\Ric(g_2)_{i j}$ are even. Since $(\rho^3 f_3)_{0 0}$ and $(\rho^3 f_3)_{i j}$ are evidently odd,~\eqref{EqStCIf3} follows. (For an alternative argument, see~\cite[(A.5)]{HintzAsymptoticallydS}.)

  Returning to~\eqref{EqStCItildehflat}, we have, modulo $\rho^4\CI$ plus almost-$\cO(\rho^4)$,
  \[
    I(D_{g'_0}\Ric-\Lambda,\rho\pa_\rho)\tilde h^\flat \equiv -\rho^3 f_3.
  \]
  The second Bianchi identity gives $\delta_{g_2}\sfG_{g_2}f=0$, which implies $f_3\in\ker I(\delta_{g'_0}\sfG_{g'_0},3)$. In view of~\eqref{EqStCIf3}, we can apply Lemma~\ref{LemmaStCISolv}\eqref{ItStCISolv3} to conclude that, in fact, $f_3=0$. Therefore, the $\rho^3$ leading order term $\sum_{m=0}^{m_3} \rho^3(\log\rho)^m \tilde h^\flat_m$ of $\tilde h^\flat$ satisfies
  \[
    I(D_{g'_0}\Ric-\Lambda,\rho\pa_\rho) \sum_{m=0}^{m_3}\rho^3(\log\rho)^m \tilde h^\flat_m = 0.
  \]
  In view of Corollary~\ref{CorStCISeq}, we can thus find coefficients $\dot a_3=\dot a_{\sfb,3}+\sum_{m=0}^{m_3}\rho^3(\log\rho)^m\dot a_{3,m}$ and $\dot b_3=\dot b_{\sfb,3}+\sum_{m=0}^{m_3}\rho^3(\log\rho)^m\dot b_{3,m}$, with $\dot a_{\sfb,3},\dot b_{\sfb,3}\in\CI(\cI^+_{R_0})$ and $\dot a_{3,m},\dot b_{3,m}\in R^\alpha\Hb^\infty(\cI^+_{R_0})$, so that, for $a_3=a_2+\dot a_3$ and $b_3=b_2+\dot b_3$,
  \begin{equation}
  \label{EqStCIg3}
    g'_3 := \phi_{a_3,b_3}^*g \equiv g'_0 + \rho^2\tilde h_2 + \rho^3\tilde h_3,\quad
    \tilde h_3\in (\CI+R^\alpha\Hb^\infty)\bigl(\cI^+_{R_0};\upbeta^*(\tau^{-2}\ker\tr_{g_{(0)}})\bigr),
  \end{equation}
  modulo $\rho^4\CI$ plus almost-$\cO(\rho^4)$. Henceforth, we re-define `almost-$\cO(\rho^k)$' to mean `$\rho^k\CI$ plus almost-$\cO(\rho^k)$'.

  \pfstep{Step~4. Simplification of the $\rho^4$ terms.} Writing $g'_3:=g_2+\rho^3\tilde h_3+\tilde h^\flat$ where now $\tilde h^\flat$ is almost-$\cO(\rho^4)$, we have
  \begin{align}
    0 &= \Ric(g'_3)-\Lambda g'_3 \nonumber\\
  \label{EqStCIRic31}
      &\equiv \Ric(g_2+\rho^3\tilde h_3)-\Lambda(g_2+\rho^3\tilde h_3) + I(D_{g'_0}\Ric-\Lambda,\rho\pa_\rho)\tilde h^\flat \\
  \label{EqStCIRic32}
      &\equiv (\Ric(g_2)-\Lambda g_2) + \bigl(D_{g'_0}\Ric(\rho^3\tilde h_3) - \Lambda\rho^3\tilde h_3\bigr) + I(D_{g'_0}\Ric-\Lambda,\rho\pa_\rho)\tilde h^\flat
  \end{align}
  modulo almost-$\cO(\rho^5)$.

  The parity arguments used to show~\eqref{EqStCIf3} imply that the $\rho^4$-coefficient $f_4$ of the first parenthesis $\Ric(g_2)-\Lambda g_2$ on the right in~\eqref{EqStCIRic32} has mixed coefficients equal to $0$, i.e.\ $(f_4)_{i 0}=0$ for all $1\leq i\leq 3$. Since every element in the range of $I(D_{g'_0}\Ric-\Lambda,\lambda)$ has vanishing mixed coefficients as well (cf.\ the vanishing of the second row of~\eqref{EqStCIIndRic}), we conclude that the same must be true for the $\rho^4$-coefficient
  \[
    \tilde f_4=\bigl(\rho^{-4}(D_{g'_0}\Ric(\rho^3\tilde h_3)-\Lambda\rho^3\tilde h_3)\bigr)\big|_{\rho=0}
  \]
  of the second parenthesis in~\eqref{EqStCIRic32}. Note first, however, that the sum of the first two parentheses on the right in~\eqref{EqStCIRic32}, which is almost-$\cO(\rho^4)$, gets mapped to almost-$\cO(\rho^5)$ by $I(\delta_{g'_0}\sfG_{g'_0},\rho\pa_\rho)$ due to the second Bianchi identity for $g_2+\rho^3\tilde h_3$ (applied to~\eqref{EqStCIRic31}). The hypotheses of Lemma~\ref{LemmaStCISolv}\eqref{ItStCISolv4} are thus satisfied for $f:=f_4+\tilde f_4$. Arguing as after~\eqref{EqStCIStep2}, we can thus use Corollary~\ref{CorStCISeq} to find $\dot a_4,\dot b_4$ (quasi-homogeneous of degree $4$ in $\rho$) so that, for $a_4=a_3+\dot a_4$, $b_4=b_3+\dot b_4$, we have
  \[
    g'_4 := \phi_{a_4,b_4}^*g \equiv g'_0 + \rho^2\tilde h_2 + \rho^3\tilde h_3 + \rho^4\tilde h_4,\qquad \tilde h_4\in (\CI+R^\alpha\Hb^\infty)\bigl(\cI^+_{R_0};\upbeta^*(\tau^{-2}S^2 T^*X)\bigr),
  \]
  modulo almost-$\cO(\rho^5)$.

  We extract further information from vanishing of the mixed (i.e.\ $(i,0)$-)coefficients of $\tilde f_4$. These are equal to the mixed coefficients of the $\rho^4$ term of $\Ric(g'_0+\rho^3\tilde h_3)-\Ric(g'_0)$. We compute, using~\eqref{EqE1StructChristoffel} for $g'_0+\rho^3\tilde h_3$ and $g'_0$, that $\Gamma^\lambda_{\mu\nu}=0$ except for
  \begin{alignat*}{2}
    \Gamma^0_{i j} &= \tfrac12\rho^3 (g_{(3)})_{i j}-(g_{(0)})_{i j},&\qquad
    \Gamma^\ell_{i 0} &\equiv \tfrac12\rho^3(g_{(3)})_i{}^\ell - \delta_i^\ell, \\
    \Gamma^\ell_{0 j} &\equiv \tfrac32\rho^3(g_{(3)})_j{}^\ell, &\qquad
    \Gamma^\ell_{i j} &\equiv \tau\Gamma(g_{(0)})^\ell_{i j} + \cO(\rho^4),
  \end{alignat*}
  where we raise indices using $g_{(0)}$, and `$\equiv$' means equality modulo $\cO(\rho^6)$. Therefore, we can use~\eqref{EqStCIRic} for $g'_0+\rho^3\tilde h_3$ and $g'_0$ to compute, modulo $\cO(\rho^5)$,\footnote{The right hand side of \cite[(A.6)]{HintzAsymptoticallydS} should feature $-\frac{n}{2}$ instead of $n$.}
  \begin{align*}
    &\Ric(g'_0+\rho^3\tilde h_3)_{i 0} - \Ric(g'_0)_{i 0} \\
    &\qquad\equiv \tau\pa_\ell\bigl(\tfrac32\rho^3(g_{(3)})_i{}^\ell\bigr) + \tau\Gamma(g_{(0)})^\ell_{\ell k}\tfrac32\rho^3(g_{(3)})_i{}^k - \tau\Gamma(g_{(0)})^k_{\ell i}\tfrac32\rho^3(g_{(3)})_k{}^\ell \\
    &= -\tfrac32\tau\delta_{g_{(0)}}(\rho^3 g_{(3)})_i.
  \end{align*}
  Recalling that $\pa_\ell$ are (spatial) coordinate derivatives in the coordinates $\tau,x^1,x^2,x^3$ valid on the interior of $\breve M$, and recalling moreover that $\rho=\tau|x|^{-1}$, the $\rho^4$-coefficient of this tensor vanishes if and only if $\delta_{g_{(0)}}(|x|^{-3} g_{(3)})=0$, which is~\eqref{EqStCIh3}.

  \pfstep{Step~5. Simplification of the remaining terms; completion of the construction.} From this point onwards, we do not need any further careful considerations since we can use the simplest parts (namely, those which apply for $\lambda>4$) of Corollary~\ref{CorStCISeq} and Lemma~\ref{LemmaStCISolv}. To wit, the almost-$\cO(\rho^5)$-term $\tilde h^\flat:=g'_4-(g'_0+\rho^2\tilde h_2+\rho^3\tilde h_3+\rho^4\tilde h_4)$ satisfies an equation
  \[
    I(D_{g'_0}\Ric-\Lambda,\rho\pa_\rho)\tilde h^\flat \equiv \rho^5 f
  \]
  modulo almost-$\cO(\rho^6)$, where $f\in (\CI+R^\alpha\Hb^\infty)\cap\ker I(\delta_{g'_0}\sfG_{g'_0},5)$. We can therefore find $\dot a_5,\dot b_5$ so that, upon setting $a_5=a_4+\dot a_5$ and $b_5=b_4+\dot b_5$,
  \[
    \phi_{a_5,b_5}^*g \equiv g'_0 + \sum_{i=2}^5 \rho^i\tilde h_i
  \]
  modulo almost-$\cO(\rho^6)$. Proceeding in this fashion produces $\dot a_i,\dot b_i$, $i=6,7,8,\ldots$. Taking $a$ and $b$ to be asymptotic sums of $\dot a_1,\dot a_2,\ldots$ and $\dot b_1,\dot b_2,\ldots$, respectively, the conclusions of the Proposition follow.
\end{proof}

With this result, the proof of Theorem~\ref{ThmEStabEin} is now complete, with a small caveat: fixing any $\rho_1\in(0,\rho_0)$ and $R_1\in(0,R_0)$, the arguments thus far produce a metric $g$ which is defined only on the union of $\Omega_{\rho_0,R_0}\cap\{\rho\geq\rho_1\}$ and $\{\rho\leq\rho_1,\ R\leq R_1\}$ where $\rho_1\in(0,\rho_0)$ and $R_1\in(0,R_0)$ are fixed but arbitrary. For $R_1$ close to $R_0$ and $\rho_1>0$ small, this domain is thus slightly smaller than $\Omega_{\rho_0,R_0}$ itself. This can easily be fixed. One way is to pull back $g$ along a diffeomorphism which is the identity near $\cK\cup\Sigma_{\rho_0,R_0}$ and a map $(\tau,x)\mapsto(\tau,\lambda(x)x)$ near $\cI^+$ for a suitable map $\lambda$ with $\lambda(x)=1$ for small $|x|$ and $\lambda(x)<1$ for $|x|$ near $R_0$. Another way is to solve the original gauge-fixed Einstein equations from the outset on a slightly larger domain (replacing $\frac12$ in~\eqref{Eq0bDom} by $\frac34$, say); then the pullback arguments in this section produce a metric which is well-defined on $\Omega_{\rho_0,R_0}$.

\begin{rmk}[Comparison with the Riemannian setting, II]
\label{RmkStRiem2}
  The construction of a coordinate system in which the conformally rescaled metric, initially only known to be polyhomogeneous, has optimal regularity is done in the Riemannian setting in \cite[\S6]{ChruscielDelayLeeSkinnerConfCompReg} in one go via Fermi (or boundary normal) coordinates and an appeal to the formal computations of \cite{FeffermanGrahamAmbient,FeffermanGrahamAmbientBook}. Our approach avoids an analysis of special coordinate systems (which would be delicate near $\cI^+\cap\cK$) in favor of a more direct argument which in particular directly handles the relevant parts of the Fefferman--Graham argument.
\end{rmk}

\bibliographystyle{alphaurl}


\end{document}